\documentclass[onecolumn, conference, 11pt]{IEEEtran}

\usepackage[english]{babel}

\usepackage{cite}

\usepackage{hyperref}

\usepackage{paralist}

\usepackage[final]{microtype}

\hypersetup
  {
  pdfsubject    = {},
  pdfkeywords   = {},
  pdfproducer   = {},
  pdfcreator    = {},
  pdfpagemode   = {UseNone},
  pdfstartview  = {FitH},
  colorlinks
  }

% ...

\usepackage{amsmath}
\usepackage{amssymb}
\usepackage{amsthm}
\usepackage{array}

% ...

% ...

% ...

\usepackage{etoolbox}

  {

  \usepackage[switch, columnwise, displaymath, mathlines]{lineno}\linenumbers

  \newcommand*\linenomathpatch[1]{%
    \cspreto{#1}{\linenomath}%
    \cspreto{#1*}{\linenomath}%
    \cspreto{end#1}{\endlinenomath}%
    \cspreto{end#1*}{\endlinenomath}%
  }
  \newcommand*\linenomathpatchAMS[1]{%
    \cspreto{#1}{\linenomathAMS}%
    \cspreto{#1*}{\linenomathAMS}%
    \csappto{end#1}{\endlinenomath}%
    \csappto{end#1*}{\endlinenomath}%
  }

  \expandafter\ifx\linenomath\linenomathWithnumbers
    \let\linenomathAMS\linenomathWithnumbers
    \patchcmd\linenomathAMS{\advance\postdisplaypenalty\linenopenalty}{}{}{}
  \else
    \let\linenomathAMS\linenomathNonumbers
  \fi

  \linenomathpatch{equation}
  \linenomathpatchAMS{gather}
  \linenomathpatchAMS{multline}
  \linenomathpatchAMS{align}
  \linenomathpatchAMS{alignat}
  \linenomathpatchAMS{flalign}

  }

\usepackage{xargs}
\usepackage{xspace}
\usepackage{xstring}
\usepackage{stringstrings}

\ifdef{\mathscr}{}{\DeclareMathAlphabet{\mathscr}{U}{rsfs}{m}{n}}
\DeclareMathAlphabet{\mathpzc}{T1}{pzc}{m}{it}
\DeclareMathAlphabet{\mathpzcx}{U}{eus}{m}{n}
\DeclareMathAlphabet{\mathbbo}{U}{bbold}{m}{n}

\newcommand{\argemp}[2]
  {\if&#1&\else#2\fi}

\newcommand{\argdef}[2]
  {\if&#1&#2\else#1\fi}

\newcommand{\arglef}[2]
  {\argemp{#2}{#1\allowbreak#2}}

\newcommand{\argrig}[2]
  {\argemp{#1}{#1\allowbreak#2}}

\newcommand{\argmid}[3]
  {\argemp{#2}{#1\allowbreak#2\allowbreak#3}}

\newcommand{\argsep}[3]
  {\if&#1&#3\else#1\allowbreak\arglef{#2}{#3}\fi}

\newcommand{\argvarcmd}[6]
  {\expandafter\newcommand\csname gobble#1arg\endcsname[2]
    {\csname check#1arg\endcsname{\argsep{##1}{#4}{##2}}}%
  \expandafter\newcommand\csname check#1arg\endcsname[1]
    {\csname @ifnextchar\endcsname%
      \bgroup{\csname gobble#1arg\endcsname{##1}}{#2{##1#5}#6}}%
  \expandafter\newcommand\csname#1\endcsname[1]
    {\csname check#1arg\endcsname{#3##1}}}

\def\forallcmd#1%
  {\ifx#1\forallcmd\else\defcmd#1\expandafter\forallcmd\fi}

\newcommandx{\defenv}[6][2=, 3=, 4=, 5=]
  {\def\thestring{}%
  \if&#4&%
    \begin{#1}[#2]%
      \argemp{#5}{\caselower[q]{\argemp{#3}{#3:}#5}\label{\thestring}}%
      \if&\thestring&%
        \def\labelx##1{\label{##1}}%
      \else%
        \def\labelx##1{\label{\thestring(##1)}}%
      \fi%
      #6%
    \end{#1}%
  \else%
    \csedef{#3#4Environment}{#1}%
    \csedef{#3#4Title}{#2}%
    \caselower[q]{\argemp{#3}{#3:}\argdef{#5}{#4}}
    \csedef{#3#4Label}{\thestring}%
    \csdef{#3#4Content}{#6}%
    \begin{#1}[#2]%
      \label{\thestring}
      \def\labelx##1{\label{\thestring(##1)}}%
      \csname#3#4Content\endcsname%
    \end{#1}%
  \fi\def\labelx##1{}}

\newcommand{\renewcounter}[2]
  {\setcounterref{#1}{#2}\addtocounter{#1}{-1}}

\newcommandx{\recenv}[1]
  {\ifcsname#1Environment\endcsname%
    \renewcounter{\csname#1Environment\endcsname}{\csname#1Label\endcsname}%
    \if&\csname#1Title\endcsname&%
      \begin{\csname#1Environment\endcsname}%
        \csname#1Content\endcsname%
      \end{\csname#1Environment\endcsname}
    \else%
      \begin{\csname#1Environment\endcsname}[\csname#1Title\endcsname]%
        \csname#1Content\endcsname%
      \end{\csname#1Environment\endcsname}
    \fi%
  \else%
    \textbf{[Undefined environment: #1!]}%
  \fi}

\newcommandx{\sidenote}[2][1=]
  {\marginpar
    [\textbf{\argrig{#1}{:}\raggedleft\small\sffamily #2\\}]
    {\textbf{\argrig{#1}{:}\raggedright\small\sffamily #2\\}}}

\newcommand{\txtsubsup}[3][]
  {\ensuremath{\argemp{#2}{_{\text{#1#2}}}\argemp{#3}{^{\text{#1#3}}}}}
\newcommandx{\txtnew}[5][1=, 3=, 4=, 5=]
  {\text{#1#2\txtsubsup[#1]{#3}{#4}#5}\xspace}
\newcommandx{\txtargnew}[7][1=, 3=, 4=, 5=, 7=]
  {\txtnew[#1]{#2}[#3][#4][#5\argmid{(}{#6}{)}#7]}
\newcommandx{\txtparnew}[7][1=, 3=, 4=, 5=, 7=]
  {\txtnew[#1]{#2}[#3][#4][#5\argmid{[}{#6}{]}#7]}
\newcommandx{\txtnewsty}[2][2=]
  {\txtnew[\argdef{#2}{#1}]}
\newcommandx{\txtargnewsty}[2][2=]
  {\txtargnew[\argdef{#2}{#1}]}
\newcommandx{\txtparnewsty}[2][2=]
  {\txtparnew[\argdef{#2}{#1}]}

\newcommand{\deftxt}[1]
  {\csdef{txt#1}{\txtnewsty{\csname txtsty#1\endcsname}}%
  \csdef{txtarg#1}{\txtargnewsty{\csname txtsty#1\endcsname}}%
  \csdef{txtpar#1}{\txtparnewsty{\csname txtsty#1\endcsname}}}

\newcommandx{\deftxtusr}[4][4=]
  {\csdef{#1#2}{\csname txt#3\endcsname{\argdef{#4}{#1}}}}

\AfterEndPreamble
  {

  \deftxt{def}
  \newcommandx{\deftxtdef}[2][2=]
    {\deftxtusr{#1}{}{def}[#2]}

  \deftxt{abr}
  \newcommandx{\deftxtabr}[2][2=]
    {\deftxtusr{#1}{}{abr}[#2]}

  \deftxt{name}
  \newcommandx{\deftxtname}[2][2=]
    {\deftxtusr{#1}{}{name}[#2]}

  \deftxt{com}
  \newcommandx{\deftxtcom}[2][2=]
    {\deftxtusr{#1}{}{com}[#2]}

  }

\newcommand{\mthsubsup}[2]
  {\argemp{#1}{_{#1}}\argemp{#2}{^{#2}}}
\newcommandx{\mthnew}[5][1=, 3=, 4=, 5=]
  {{\ensuremath{\csname#1\endcsname{#2}\mthsubsup{#3}{#4}#5}}}
\newcommandx{\mthargnew}[7][1=, 3=, 4=, 5=, 7=]
  {\mthnew[#1]{#2}[#3][#4][#5\argmid{\!\left(}{#6}{\right)}#7]}
\newcommandx{\mthparnew}[7][1=, 3=, 4=, 5=, 7=]
  {\mthnew[#1]{#2}[#3][#4][#5\argmid{\!\left[}{#6}{\right]}#7]}
\newcommandx{\mthnewsty}[2][2=]
  {\mthnew[\argdef{#2}{#1}]}
\newcommandx{\mthargnewsty}[2][2=]
  {\mthargnew[\argdef{#2}{#1}]}
\newcommandx{\mthparnewsty}[2][2=]
  {\mthparnew[\argdef{#2}{#1}]}

\newcommand{\mthsty}{}
\newcommand{\mth}
  {\mthnewsty{\mthsty}}

\newcommand{\defmth}[1]
  {\csdef{mth#1}{\mthnewsty{mthsty#1}}%
  \csdef{mtharg#1}{\mthargnewsty{mthsty#1}}%
  \csdef{mthpar#1}{\mthparnewsty{mthsty#1}}}

\newcommand{\defmthseq}[3]
  {\def\defcmd##1{\csdef{##1#1}{\csname mth#2\endcsname{##1}}}%
  \forallcmd#3\forallcmd}
\newcommand{\defmthupp}[2]
  {\defmthseq{#1}{#2}{ABCDEFGHIJKLMNOPQRSTUVWXYZ}}
\newcommand{\defmthlow}[2]
  {\defmthseq{#1}{#2}{abcdefghijklmnopqrstuvwxyz}}
\newcommand{\defmthall}[2]
  {\defmthupp{#1}{#2}\defmthlow{#1}{#2}}

\newcommandx{\defmthusr}[4][4=]
  {\csdef{#1#2}{\csname mth#3\endcsname{\argdef{#4}{#1}}}}
\newcommandx{\defmthusrupp}[4][4=]
  {\defmthusr{#1}{#2}{#3}[#4]%
  \defmthupp{#1#2}{#3}}
\newcommandx{\defmthusrlow}[4][4=]
  {\defmthusr{#1}{#2}{#3}[#4]%
  \defmthlow{#1#2}{#3}}
\newcommandx{\defmthusrall}[4][4=]
  {\defmthusr{#1}{#2}{#3}[#4]%
  \defmthall{#1#2}{#3}}

\AfterEndPreamble
  {

  \defmth{name}
  \defmthupp{Name}{name}
  \newcommandx{\defmthname}[2][2=]
    {\defmthusr{#1}{Name}{name}[#2]}

  \defmth{fam}
  \defmthupp{Fam}{fam}
  \newcommandx{\defmthfam}[2][2=]
    {\defmthusr{#1}{Fam}{fam}[#2]}

  \defmth{cls}
  \defmthall{Cls}{cls}
  \newcommandx{\defmthcls}[2][2=]
    {\defmthusr{#1}{Cls}{cls}[#2]}

  \defmth{sig}
  \defmthall{Sig}{sig}
  \newcommandx{\defmthsig}[2][2=]
    {\defmthusr{#1}{Sig}{sig}[#2]}

  \defmth{str}
  \defmthall{Str}{str}
  \newcommandx{\defmthstr}[2][2=]
    {\defmthusr{#1}{Str}{str}[#2]}

  \defmth{set}
  \defmthall{Set}{set}
  \newcommandx{\defmthset}[2][2=]
    {\defmthusr{#1}{Set}{set}[#2]}
  \newcommandx{\defmthsetext}[3][2=, 3=]
    {\defmthset{#1}[#2]%
    \caselower[q]{#1}%
    \defmthusrall{\thestring}{Sym}{sym}[\argdef{#3}{\lowercase{#2}}]%
    \defmthusrall{\thestring}{Elm}{elm}[\argdef{#3}{\lowercase{#2}}]}

  \defmth{rel}
  \defmthall{Rel}{rel}
  \newcommandx{\defmthrel}[2][2=]
    {\defmthusr{#1}{Rel}{rel}[#2]}

  \defmth{fun}
  \defmthall{Fun}{fun}
  \newcommandx{\defmthfun}[2][2=]
    {\defmthusr{#1}{Fun}{fun}[#2]}

  \defmth{sym}
  \defmthall{Sym}{sym}
  \newcommandx{\defmthsym}[2][2=]
    {\defmthusr{#1}{Sym}{sym}[#2]}

  \defmth{elm}
  \defmthall{Elm}{elm}
  \newcommandx{\defmthelm}[2][2=]
    {\defmthusr{#1}{Elm}{elm}[#2]}

  \newcommandx{\defmthsymelm}[2][2=]
    {\defmthsym{#1}[#2]%
    \defmthelm{#1}[#2]}

  \defmth{mat}
  \defmthall{Mat}{mat}
  \newcommandx{\defmthmat}[2][2=]
    {\defmthusr{#1}{Mat}{mat}[#2]}

  \defmth{vec}
  \defmthall{Vec}{vec}
  \newcommandx{\defmthvec}[2][2=]
    {\defmthusr{#1}{Vec}{vec}[#2]}

  \defmth{snt}
  % ...
  \newcommandx{\defmthsnt}[2][2=]
    {\defmthusr{#1}{Snt}{snt}[#2]}

  \defmth{frm}
  % ...
  \newcommandx{\defmthfrm}[2][2=]
    {\defmthusr{#1}{Frm}{frm}[#2]}

  }

\AfterEndPreamble
  {

  \deftxtabr{adhoc}[ad hoc]

  \deftxtabr{apriori}[a priori]

  \deftxtabr{afortiori}[a fortiori]

  \deftxtabr{cf}[cf.]

  \deftxtabr{dedicto}[de dicto]

  \deftxtabr{defacto}[de facto]

  \deftxtabr{dere}[de re]

  \deftxtabr{divideetimpera}[divide et impera]

  \deftxtabr{eg}[e.g.]

  \deftxtabr{ergo}[ergo]

  \deftxtabr{errata}[errata]

  \deftxtabr{erratum}[erratum]

  \deftxtabr{etc}[etc.]

  \deftxtabr{etal}[et al.]

  \deftxtabr{ie}[i.e.]

  \deftxtabr{mutatismutandis}[mutatis mutandis]

  \deftxtabr{viceversa}[vice versa]

  \deftxtabr{vs}[vs.]

  \deftxtabr{viz}[viz.]

  \deftxtabr{Dedicto}[De dicto]

  \deftxtabr{Defacto}[De facto]

  \deftxtabr{Dere}[De re]

  \deftxtabr{Divideetimpera}[Divide et impera]

  \deftxtabr{Eg}[E.g.]

  \deftxtabr{Errata}[Errata]

  \deftxtabr{Erratum}[Erratum]

  \deftxtabr{Percontra}[Per contra]

  \deftxtabr{Viceversa}[Vice versa]

  }

\AfterEndPreamble
  {

  % ...

  % ...

  }

\AfterEndPreamble
  {

  \deftxtabr{role}[r\^{o}le]

  \deftxtabr{Role}[R\^{o}le]

  }

\AfterEndPreamble
  {

  \deftxtabr{aka}[a.k.a.]

  \deftxtabr{contd}[contd.]

  \deftxtabr{iff}[iff]

  \deftxtabr{stx}[s.t.]

  \deftxtabr{resp}[resp.]

  \deftxtabr{wrt}[w.r.t.]

  \deftxtabr{wlogx}[w.l.o.g.]

  \deftxtabr{Contd}[Contd.]

  \deftxtabr{Wlogx}[W.l.o.g.]

  }

\DeclareRobustCommand{\implies}
  {\ensuremath{\Rightarrow}}

\DeclareMathOperator{\defeq}
  {\ensuremath{\triangleq}}
\DeclareMathOperator{\seteq}
  {\ensuremath{:\!=\,}}

\DeclareRobustCommand{\cmodels}
  {\mth{\,\models}}
\DeclareRobustCommand{\notcmodels}
  {\mth{\,\not\models}}

\DeclareRobustCommand{\cequiv}
  {\mth{\,\equiv}}

\DeclareRobustCommand{\cimplies}
  {\mth{\,\implies}}

% ...

\newcommand{\dual}[1]
  {\mth{\overline{#1}}}

\newcommand{\der}[1]
  {\mth{\widehat{#1}}}
\newcommand{\trn}[1]
  {\mth{\widetilde{#1}}}

\let\oldvec\vec
\renewcommand{\vec}[1]
  {\mth{\oldvec{#1}}}

% ...

\argvarcmd{tuple}{\mth}{\left\langle}{,}{\right\rangle}{}
\argvarcmd{tuplex}{\mth}{\left\langle}{;}{\right\rangle}{}

\newcommand{\set}[2]
  {\mth{\argmid{\{\,}{\argsep{#1}{\allowbreak\mid\allowbreak}{#2}}{\,\}}}}

\newcommand{\card}[1]
  {\mth{\argmid{\lvert}{#1}{\rvert}}}

\newcommand{\pow}[1]
  {\ensuremath{2^{\argdef{#1}{\cdot}}}}

% ...

% ...

\DeclareRobustCommand{\dom}
  {\mthargfun{dom}}

\DeclareRobustCommand{\img}
  {\mthargfun{img}}

\DeclareMathOperator{\rst}
  {\ensuremath{\upharpoonright}}

% ...

\DeclareRobustCommand{\emptyfun}
  {\mth{\varnothing}}

\DeclareMathOperator{\pto}
  {\ensuremath{\rightharpoonup}}
\DeclareMathOperator{\pmapsto}
  {\ensuremath{\mathrel{\raisebox{0.5ex}{\footnotesize${\llcorner}$}%
    \kern-1.5ex\rightharpoonup}}}

% ...

\DeclareRobustCommand{\AOmicron}
  {\mthargset{O}}

\newcommand{\SetB}
  {\mthset[mathbb]{B}}

\newcommand{\SetN}
  {\mthset[mathbb]{N}}

\newcommand{\numcc}[2]
  {\mth{[\argsep{#1}{,}{#2}]}}
\newcommand{\numco}[2]
  {\mth{[\argsep{#1}{,}{#2})}}

\DeclareRobustCommand{\fst}
  {\mthargfun{fst}}
\DeclareRobustCommand{\lst}
  {\mthargfun{lst}}

\newcommandx{\defcomcls}[2][2=]
  {\defcomclssem{#1}{\argdef{#2}{#1}}%
  \defcomclssem{#1}{\argdef{#2}{#1}}[Co]}
\newcommandx{\defcomclssem}[3][3=]
  {\defcomclsred{#3#1}{#2}[#3]%
  \defcomclsred{#3U#1}{#2}[#3U]%
  \defcomclsred{#3N#1}{#2}[#3N]%
  \defcomclsred{#3A#1}{#2}[#3A]}
\newcommandx{\defcomclsred}[3][3=]
  {\defcomclscmd{#1}{#2}[#3]%
  \defcomclscmd{#1E}{#2}[#3][-easy]%
  \defcomclscmd{#1H}{#2}[#3][-hard]%
  \defcomclscmd{#1C}{#2}[#3][-complete]}%
\newcommandx{\defcomclscmd}[4][3=, 4=]
  {\csdef{#1}{\txtcom{#3#2#4}}}

\AfterEndPreamble
  {

  \defcomcls{Time}
  \defcomcls{Space}

  \defcomcls{LogTime}
  \defcomcls{LogSpace}

  \defcomcls{PTime}
  \defcomcls{PSpace}

  \defcomcls{ExpTime}
  \defcomcls{ExpSpace}

  }

\newcommandx{\SATG}[3][1=, 2=, 3=]
  {\txtname{Sat}[#1][#2][\argmid{[}{#3}{]}]}

\newcommandx{\VALG}[3][1=, 2=, 3=]
  {\txtname{Val}[#1][#2][\argmid{[}{#3}{]}]}

\newcommandx{\EVLG}[3][1=, 2=, 3=]
  {\txtname{Evl}[#1][#2][\argmid{[}{#3}{]}]}

\newcommandx{\MCG}[3][1=, 2=, 3=]
  {\txtname{MC}[#1][#2][\argmid{[}{#3}{]}]}

\newcommandx{\EFG}[3][1=, 2=, 3=]
  {\txtname{EF}[#1][#2][\argmid{[}{#3}{]}]}

\AfterEndPreamble
  {

  \newcommand{\plrsym}{E}
  \defmthsym{Plr}[\plrsym]
  \newcommand{\oppsym}{A}
  \defmthsym{Opp}[\oppsym]

  \newcommand{\arenaname}{A}
  \defmthusrupp{Arena}{}{name}[\arenaname]

  \newcommand{\possym}{v}
  \defmthsetext{Pos}[Ps][\possym]
  \defmthsymelm{ipos}[\possym_{I}]
  \defmthsymelm{fpos}[\possym_{F}]
  \defmthset{PPos}[Ps_{\PlrSym}]
  \defmthsymelm{ppos}[\possym_{\PlrSym}]
  \defmthset{OPos}[Ps_{\OppSym}]
  \defmthsymelm{opos}[\possym_{\OppSym}]
  \defmthrel{Mov}[Mv]

  \newcommand{\gamename}{G}
  \defmthusrupp{Game}{}{name}[\gamename]

  \defmthset{Win}[Wn]

  \defmthset{Obs}[Ob]
  \defmthfun{obs}

  \newcommand{\pthsym}{\pi}
  \defmthsetext{Pth}[][\pthsym]
  \defmthfun{pth}

  \newcommand{\hstsym}{\rho}
  \defmthsetext{Hst}[][\hstsym]
  \defmthset{PHst}[Hst_{\PlrSym}]
  \defmthsymelm{phst}[\hstsym_{\PlrSym}]
  \defmthset{OHst}[Hst_{\OppSym}]
  \defmthsymelm{ohst}[\hstsym_{\OppSym}]
  \defmthfun{hst}

  \newcommand{\playsym}{\pi}
  \defmthsetext{Play}[Play][\playsym]
  \defmthfun{play}

  \newcommand{\strsym}{\sigma}
  \defmthsetext{Str}[][\strsym]
  \defmthset{PStr}[Str_{\PlrSym}]
  \defmthsymelm{pstr}[\strsym_{\PlrSym}]
  \defmthset{OStr}[Str_{\OppSym}]
  \defmthsymelm{ostr}[\strsym_{\OppSym}]

  \newcommand{\prfsym}{\xi}
  \defmthsetext{Prf}[][\prfsym]

  \defmthfun{prd}
  \defmthfun{suc}

  \defmthfun{atr}
  \defmthfun{rch}

  \defmthfun{sol}

  }

\newcommandx{\BG}[3][1=, 2=, 3=]
  {\txtname{BG}[#1][#2][\argmid{[}{#3}{]}]}

\newcommandx{\CG}[3][1=, 2=, 3=]
  {\txtname{CG}[#1][#2][\argmid{[}{#3}{]}]}

\newcommandx{\PG}[3][1=, 2=, 3=]
  {\txtname{PG}[#1][#2][\argmid{[}{#3}{]}]}

\newcommandx{\RG}[3][1=, 2=, 3=]
  {\txtname{RG}[#1][#2][\argmid{[}{#3}{]}]}

\newcommandx{\SG}[3][1=, 2=, 3=]
  {\txtname{SG}[#1][#2][\argmid{[}{#3}{]}]}

\newcommandx{\MG}[3][1=, 2=, 3=]
  {\txtname{MG}[#1][#2][\argmid{[}{#3}{]}]}

\AfterEndPreamble
  {

  \newcommand{\prtsym}{p}
  \defmthsetext{Prt}[Pr][\prtsym]
  \defmthfun{prt}[pr]

  % ...

  }

\newcommandx{\EG}[3][1=, 2=, 3=]
  {\txtname{EG}[#1][#2][\argmid{[}{#3}{]}]}

\newcommandx{\MPG}[3][1=, 2=, 3=]
  {\txtname{MPG}[#1][#2][\argmid{[}{#3}{]}]}

\newcommandx{\DPG}[3][1=, 2=, 3=]
  {\txtname{DPG}[#1][#2][\argmid{[}{#3}{]}]}

\AfterEndPreamble
  {

  \newcommand{\wghsym}{w}
  \defmthsetext{Wgh}[Wg][\wghsym]
  \defmthfun{wgh}[wg]

  % ...

  }

\newcommandx{\BF}[3][1=, 2=, 3=]
  {\txtname{BF}[#1][#2][\argmid{[}{#3}{]}]}

\newcommandx{\QBF}
  {{\txtname{Q}}\BF}
\newcommandx{\EBF}
  {\ensuremath{\exists}\BF}
\newcommandx{\UBF}
  {\ensuremath{\forall}\BF}

\AfterEndPreamble
  {

  \defmthsig{Log}[L]

  \defmthusr{Tt}{}{sym}[\top]
  \defmthusr{Ff}{}{sym}[\bot]

  \newcommand{\apsym}{p}
  \defmthsetext{AP}[][\apsym]
  \defmthfun{ap}\defmthusr{ap}{}{argfun}

  \defmthusr{Con}{}{sym}[Cn]

  \defmthusr{sub}{}{argfun}

  \defmthusr{Qnt}{}{sym}[Q]

  \newcommand{\qntsym}{\wp}
  \defmthsetext{Qnt}[Qn][\qntsym]

  \defmthusr{free}{}{argfun}

  \defmthusr{dep}{}{argfun}
  \defmthusr{alt}{}{argfun}

  \defmthstr{Log}[L]

  \newcommand{\valsym}{\xi}
  \defmthsetext{Val}[][\valsym]

  \newcommand{\asgsym}{\chi}
  \defmthsetext{Asg}[][\asgsym]

  }

\newcommandx{\FOL}[3][1=, 2=, 3=]
  {\txtname{Fol}[#1][#2][\argmid{[}{#3}{]}]}

\newcommandx{\MFOL}
  {{\txtname{M}}\FOL}

\newcommandx{\IF}[3][1=, 2=, 3=]
  {\txtname{IF}[#1][#2][\argmid{[}{#3}{]}]}

\AfterEndPreamble
  {

  \defmthsig{Var}[V]
  \newcommand{\varsym}{x}
  \defmthsetext{Var}[Vr][\varsym]
  \defmthusr{var}{}{argfun}[vr]
  \defmthfun{dim}[dm]\defmthusr{dim}{}{argfun}[dm]

  \defmthsig{Con}[C]
  \newcommand{\consym}{c}
  \defmthsetext{Con}[Cn][\consym]
  \defmthusr{con}{}{argfun}[cn]

  \defmthsig{Fun}[F]
  \newcommand{\funsym}{f}
  \defmthsetext{Fun}[Fn][\funsym]
  \defmthusr{fun}{}{argfun}[fn]
  \defmthfun{art}[ar]\defmthusr{art}{}{argfun}[ar]

  \defmthsig{Ter}[T]
  \newcommand{\tersym}{t}
  \defmthsetext{Ter}[Tr][\tersym]
  \defmthusr{ter}{}{argfun}

  \defmthsig{Rel}[R]
  \newcommand{\relsym}{r}
  \defmthsetext{Rel}[Rl][\relsym]
  \defmthusr{rel}{}{argfun}[rl]

  \defmthusr{skm}{}{argfun}[skm]

  \defmthstr{Con}[C]
  \defmthstr{Fun}[F]
  \defmthstr{Ter}[T]

  \defmthstr{Rel}[R]

  }

\newcommandx{\SOL}[3][1=, 2=, 3=]
  {\txtname{Sol}[#1][#2][\argmid{[}{#3}{]}]}

\newcommandx{\MSOL}
  {{\txtname{M}}\SOL}

\AfterEndPreamble
  {

  \newcommand{\fvarsym}{x}
  \defmthsetext{FVar}[FVr][\fvarsym]

  \newcommand{\svarsym}{X}
  \defmthsetext{SVar}[SVr][\svarsym]

  % ...

  }

\newcommandx{\TL}[3][1=, 2=, 3=]
  {\txtname{TL}[#1][#2][\argmid{[}{#3}{]}]}

\newcommandx{\MTL}
  {{\txtname{M}}\TL}

\newcommandx{\PL}[3][1=, 2=, 3=]
  {\txtname{PL}[#1][#2][\argmid{[}{#3}{]}]}

\newcommandx{\MPL}
  {{\txtname{M}}\PL}

\AfterEndPreamble
  {

  % ...

  % ...

  }

\newcommandx{\ML}[3][1=, 2=, 3=]
  {\txtname{ML}[#1][#2][\argmid{[}{#3}{]}]}

\newcommandx{\QML}
  {{\txtname{Q}}\ML}
\newcommandx{\EML}
  {\ensuremath{\exists}\ML}
\newcommandx{\UML}
  {\ensuremath{\forall}\ML}

\AfterEndPreamble
  {

  % ...

  \newcommand{\wrlsym}{w}
  \defmthsetext{Wrl}[W][\wrlsym]
  \defmthsymelm{iwrl}[\wrlsym_{I}]
  \defmthrel{Trn}[R]

  \defmthfun{lab}[\lambda]

  }

\newcommandx{\MC}[3][1=, 2=, 3=]
  {\ensuremath{\mu}-\txtname{Calculus}[#1][#2][\argmid{[}{#3}{]}]}

\newcommandx{\QMC}
  {{\txtname{Q}}\MC}
\newcommandx{\EMC}
  {\ensuremath{\exists}\MC}
\newcommandx{\UMC}
  {\ensuremath{\forall}\MC}

\AfterEndPreamble
  {

  % ...

  % ...

  }

\newcommandx{\PTL}[3][1=, 2=, 3=]
  {\txtname{PTL}[#1][#2][\argmid{[}{#3}{]}]}

\newcommandx{\QPTL}
  {{\txtname{Q}}\PTL}
\newcommandx{\EPTL}
  {\ensuremath{\exists}\PTL}
\newcommandx{\UPTL}
  {\ensuremath{\forall}\PTL}

\newcommandx{\LTL}[3][1=, 2=, 3=]
  {\txtname{LTL}[#1][#2][\argmid{[}{#3}{]}]}

\newcommandx{\QLTL}
  {{\txtname{Q}}\LTL}
\newcommandx{\ELTL}
  {\ensuremath{\exists}\LTL}
\newcommandx{\ULTL}
  {\ensuremath{\forall}\LTL}

\AfterEndPreamble
  {

  \defmthusr{Opr}{}{sym}[Op]

  \defmthusr{X}{}{sym}
  \defmthusr{F}{}{sym}
  \defmthusr{G}{}{sym}
  \defmthusr{U}{}{sym}
  \defmthusr{R}{}{sym}

  \defmthusr{Y}{}{sym}
  \defmthusr{P}{}{sym}
  \defmthusr{H}{}{sym}
  \defmthusr{S}{}{sym}
  \defmthusr{B}{}{sym}

  % ...

  }

\newcommandx{\CTL}[3][1=, 2=, 3=]
  {\txtname{CTL}[#1][#2][\argmid{[}{#3}{]}]}

\newcommandx{\QCTL}
  {{\txtname{Q}}\CTL}
\newcommandx{\ECTL}
  {\ensuremath{\exists}\CTL}
\newcommandx{\UCTL}
  {\ensuremath{\forall}\CTL}

\newcommandx{\CTLP}[3][1=, 2=, 3=]
  {\txtname{CTL$^{+}$}[#1][#2][\argmid{[}{#3}{]}]}

\newcommandx{\QCTLP}
  {{\txtname{Q}}\CTLP}
\newcommandx{\ECTLP}
  {\ensuremath{\exists}\CTLP}
\newcommandx{\UCTLP}
  {\ensuremath{\forall}\CTLP}

\newcommandx{\CTLS}[3][1=, 2=, 3=]
  {\txtname{CTL*}[#1][#2][\argmid{[}{#3}{]}]}

\newcommandx{\QCTLS}
  {{\txtname{Q}}\CTLS}
\newcommandx{\ECTLS}
  {\ensuremath{\exists}\CTLS}
\newcommandx{\UCTLS}
  {\ensuremath{\forall}\CTLS}

\AfterEndPreamble
  {

  \defmthusr{E}{}{sym}
  \defmthusr{A}{}{sym}

  % ...

  }

\newcommandx{\ATL}[3][1=, 2=, 3=]
  {\txtname{ATL}[#1][#2][\argmid{[}{#3}{]}]}

\newcommandx{\QATL}
  {{\txtname{Q}}\ATL}
\newcommandx{\EATL}
  {\ensuremath{\exists}\ATL}
\newcommandx{\UATL}
  {\ensuremath{\forall}\ATL}

\newcommandx{\ATLP}[3][1=, 2=, 3=]
  {\txtname{ATL$^{+}$}[#1][#2][\argmid{[}{#3}{]}]}

\newcommandx{\QATLP}
  {{\txtname{Q}}\ATLP}
\newcommandx{\EATLP}
  {\ensuremath{\exists}\ATLP}
\newcommandx{\UATLP}
  {\ensuremath{\forall}\ATLP}

\newcommandx{\ATLS}[3][1=, 2=, 3=]
  {\txtname{ATL*}[#1][#2][\argmid{[}{#3}{]}]}

\newcommandx{\QATLS}
  {{\txtname{Q}}\ATLS}
\newcommandx{\EATLS}
  {\ensuremath{\exists}\ATLS}
\newcommandx{\UATLS}
  {\ensuremath{\forall}\ATLS}

\AfterEndPreamble
  {

  % ...

  \deftxtname{CGS}

  \newcommand{\cgsstr}{G}
  \defmthusrupp{CGS}{Str}{str}[\cgsstr]

  \newcommand{\agnsym}{a}
  \defmthsetext{Agn}[Ag][\agnsym]
  \newcommand{\actsym}{c}
  \defmthsetext{Act}[Ac][\actsym]
  \newcommand{\decsym}{d}
  \defmthsetext{Dec}[Dc][\decsym]

  \defmthfun{mov}[\tau]

  }

\newcommandx{\SL}[3][1=, 2=, 3=]
  {\txtname{SL}[#1][#2][\argmid{[}{#3}{]}]}

\newcommandx{\ESL}
  {\ensuremath{\exists}\SL}
\newcommandx{\USL}
  {\ensuremath{\forall}\SL}

\newcommandx{\FSL}
  {{\txtname{F}}\SL}

\newcommandx{\EFSL}
  {\ensuremath{\exists}\FSL}
\newcommandx{\UFSL}
  {\ensuremath{\forall}\FSL}

\newcommandx{\OGSL}[3][1=, 2=, 3=]
  {\SL[#1][#2][1g\arglef{,}{#3}]}

\newcommandx{\EOGSL}
  {\ensuremath{\exists}\OGSL}
\newcommandx{\UOGSL}
  {\ensuremath{\forall}\OGSL}

\newcommandx{\FOGSL}
  {{\txtname{F}}\OGSL}

\newcommandx{\EFOGSL}
  {\ensuremath{\exists}\FOGSL}
\newcommandx{\UFOGSL}
  {\ensuremath{\forall}\FOGSL}

\newcommandx{\CGSL}[3][1=, 2=, 3=]
  {\SL[#1][#2][cg\arglef{,}{#3}]}

\newcommandx{\ECGSL}
  {\ensuremath{\exists}\CGSL}
\newcommandx{\UCGSL}
  {\ensuremath{\forall}\CGSL}

\newcommandx{\FCGSL}
  {{\txtname{F}}\xGSL}

\newcommandx{\EFCGSL}
  {\ensuremath{\exists}\FCGSL}
\newcommandx{\UFCGSL}
  {\ensuremath{\forall}\FCGSL}

\newcommandx{\DGSL}[3][1=, 2=, 3=]
  {\SL[#1][#2][dg\arglef{,}{#3}]}

\newcommandx{\EDGSL}
  {\ensuremath{\exists}\DGSL}
\newcommandx{\UDGSL}
  {\ensuremath{\forall}\DGSL}

\newcommandx{\FDGSL}
  {{\txtname{F}}\xGSL}

\newcommandx{\EFDGSL}
  {\ensuremath{\exists}\FDGSL}
\newcommandx{\UFDGSL}
  {\ensuremath{\forall}\FDGSL}

\newcommandx{\AGSL}[3][1=, 2=, 3=]
  {\SL[#1][#2][ag\arglef{,}{#3}]}

\newcommandx{\EAGSL}
  {\ensuremath{\exists}\AGSL}
\newcommandx{\UAGSL}
  {\ensuremath{\forall}\AGSL}

\newcommandx{\FAGSL}
  {{\txtname{F}}\xGSL}

\newcommandx{\EFAGSL}
  {\ensuremath{\exists}\FAGSL}
\newcommandx{\UFAGSL}
  {\ensuremath{\forall}\FAGSL}

\newcommandx{\EGSL}[3][1=, 2=, 3=]
  {\SL[#1][#2][eg\arglef{,}{#3}]}

\newcommandx{\EEGSL}
  {\ensuremath{\exists}\EGSL}
\newcommandx{\UEGSL}
  {\ensuremath{\forall}\EGSL}

\newcommandx{\FEGSL}
  {{\txtname{F}}\xGSL}

\newcommandx{\EFEGSL}
  {\ensuremath{\exists}\FEGSL}
\newcommandx{\UFEGSL}
  {\ensuremath{\forall}\FEGSL}

\newcommandx{\BGSL}[3][1=, 2=, 3=]
  {\SL[#1][#2][bg\arglef{,}{#3}]}

\newcommandx{\EBGSL}
  {\ensuremath{\exists}\BGSL}
\newcommandx{\UBGSL}
  {\ensuremath{\forall}\BGSL}

\newcommandx{\FBGSL}
  {{\txtname{F}}\xGSL}

\newcommandx{\EFBGSL}
  {\ensuremath{\exists}\FBGSL}
\newcommandx{\UFBGSL}
  {\ensuremath{\forall}\FBGSL}

\newcommandx{\NGSL}[3][1=, 2=, 3=]
  {\SL[#1][#2][ng\arglef{,}{#3}]}

\newcommandx{\ENGSL}
  {\ensuremath{\exists}\NGSL}
\newcommandx{\UNGSL}
  {\ensuremath{\forall}\NGSL}

\newcommandx{\FNGSL}
  {{\txtname{F}}\xGSL}

\newcommandx{\EFNGSL}
  {\ensuremath{\exists}\FNGSL}
\newcommandx{\UFNGSL}
  {\ensuremath{\forall}\FNGSL}

\newcommandx{\XGSL}[3][1=, 2=, 3=]
  {\SL[#1][#2][xg\arglef{,}{#3}]}

\newcommandx{\EXGSL}
  {\ensuremath{\exists}\XGSL}
\newcommandx{\UXGSL}
  {\ensuremath{\forall}\XGSL}

\newcommandx{\FXGSL}
  {{\txtname{F}}\xGSL}

\newcommandx{\EFXGSL}
  {\ensuremath{\exists}\FXGSL}
\newcommandx{\UFXGSL}
  {\ensuremath{\forall}\FXGSL}

\AfterEndPreamble
  {

  \newcommand{\bndsym}{\flat}
  \defmthsetext{Bnd}[Bn][\bndsym]
  \defmthusr{bnd}{}{argfun}

  \defmthusr{psn}{}{argfun}

  \defmthfun{nxt}

  }

\AfterEndPreamble
  {

  \deftxtname{DWA}\deftxtname{NWA}\deftxtname{UWA}\deftxtname{AWA}

  \deftxtname{DFW}\deftxtname{NFW}\deftxtname{UFW}\deftxtname{AFW}
  \deftxtname{DBW}\deftxtname{NBW}\deftxtname{UBW}\deftxtname{ABW}
  \deftxtname{DCW}\deftxtname{NCW}\deftxtname{UCW}\deftxtname{ACW}
  \deftxtname{DPW}\deftxtname{NPW}\deftxtname{UPW}\deftxtname{APW}
  \deftxtname{DRW}\deftxtname{NRW}\deftxtname{URW}\deftxtname{ARW}
  \deftxtname{DSW}\deftxtname{NSW}\deftxtname{USW}\deftxtname{ASW}
  \deftxtname{DMW}\deftxtname{NMW}\deftxtname{UMW}\deftxtname{AMW}

  \deftxtname{PD}

  \deftxtname{GFG}

  \defmthset{Aut}[Aut]

  \defmthset{WAut}[WAut]

  \newcommand{\autname}{A}
  \defmthusrupp{Aut}{Name}{name}[\autname]

  \newcommand{\sttsym}{q}
  \defmthsetext{Stt}[Q][\sttsym]
  \defmthset{IStt}[Q_{I}]
  \defmthsymelm{istt}[\sttsym_{I}]
  \defmthset{FStt}[Q_{F}]
  \defmthsymelm{fstt}[\sttsym_{F}]

  \newcommand{\symsym}{\sigma}
  \defmthsetext{Sym}[\Sigma][\symsym]

  \defmthfun{trn}[\delta]

  \defmthfun{Lang}[L]

  \newcommand{\wrdsym}{w}
  \defmthsetext{Wrd}[Wr][\wrdsym]

  }

\AfterEndPreamble
  {

  \deftxtname{DTA}\deftxtname{NTA}\deftxtname{UTA}\deftxtname{ATA}

  \deftxtname{DFT}\deftxtname{NFT}\deftxtname{UFT}\deftxtname{AFT}
  \deftxtname{DBT}\deftxtname{NBT}\deftxtname{UBT}\deftxtname{ABT}
  \deftxtname{DCT}\deftxtname{NCT}\deftxtname{UCT}\deftxtname{ACT}
  \deftxtname{DPT}\deftxtname{NPT}\deftxtname{UPT}\deftxtname{APT}
  \deftxtname{DRT}\deftxtname{NRT}\deftxtname{URT}\deftxtname{ART}
  \deftxtname{DST}\deftxtname{NST}\deftxtname{UST}\deftxtname{AST}
  \deftxtname{DMT}\deftxtname{NMT}\deftxtname{UMT}\deftxtname{AMT}

  \defmthset{TAut}[TAut]

  \newcommand{\dirsym}{d}
  \defmthsetext{Dir}[\Lambda][\dirsym]

  \newcommand{\treesym}{T}
  \defmthsetext{Tree}[Tr][\treesym]

  \defmthfun{wot}

  }

\newtheorem{definition}{Definition}
\newtheorem{construction}{Construction}

\newtheorem{claim}{Claim}
\newtheorem{theorem}{Theorem}

\newtheorem{remark}{Remark}
\newtheorem{example}{Example}

\usepackage{xcolor}

% ...

\newcommand{\ignore}[1]{}
\newcommand{\iffExpl}[1]
  {\mathrel{\stackrel{\text{\tiny #1}}{\Leftrightarrow}}}
\newcommand{\flipFun}[1]{\dual{#1}}

\AfterEndPreamble
  {

  \deftxtabr{pnf}[pnf]

  \defmthset{HAsg}
  \defmthusr{Asg}{Fam}{str}[X]

  \defmthsetext{Spc}[\mathbbo{\Theta}][\varTheta]

  \defmthsetext{Fnc}[][F]

  \defmthusr{par}{Fun}{argfun}
  \defmthusr{ext}{Fun}{argfun}
  \defmthusr{fnc}{Fun}{argfun}

  \renewcommand{\dep}[2]
    {\genfrac{<}{>}{0pt}{}{\BSym : #1}{\SSym : #2}}
  \newcommand{\bdep}[1]
    {\left< \BSym : #1 \right>}
  \newcommand{\sdep}[1]
    {\left< \SSym : #1 \right>}

  \defmthfun{wrd}

  \defmthfun{chc}[\Gamma]
  \defmthusr{Chc}{Set}{argset}

  \defmthfun{evl}
  \defmthfun{nevl}

  \defmthfun{shr}[\jmath]

  \defmthusr{QAE}{}{sym}[\forall\exists]
  \defmthusr{QEA}{}{sym}[\exists\forall]

  }

% ...

% ...
\newcommand{\Automaton}{\DName[\psi]}
\newcommand{\Astates}{Q}
\newcommand{\Astate}{q}
\newcommand{\Ainitialstate}{q_0}
\newcommand{\Aalphabet}{\Sigma}
\newcommand{\Atransition}{\delta}
\newcommand{\Aacceptingcondition}{\text{Acc}}

\newcommand{\LTLformula}{\psi}
\newcommand{\quantificationPrefix}{\qntElm}

\newcommand{\PathToPathmorph}{\mathtt{f}}
\newcommand{\PathToPath}{f}
\newcommand{\PlayToAsg}{g}
\newcommand{\QgameToAgame}{\PathToPath^{-1}}

\newcommand{\wordAsg}{w}

\newcommand{\wrdInvFun}{\wrdFun[][-1]}

\newcommand{\APosSet}{\PosSet(\Game[\varphi])}
\newcommand{\QPosSet}{\PosSet(\Game[\quantificationPrefix][\LTLformula])}

\newcommand{\Apos}[1][]{(\Astate_{#1},\valElm_{#1})}
\newcommand{\Aposq}[1]{(\Astate,\valElm_{#1})}
\newcommand{\AposA}[1]{(\Astate#1,\emptyfun)}
\newcommand{\Qpos}[1][]{\valElm_{#1}}

\newcommand{\iQPaths}{\PthSet_\text{init}(\Game[\quantificationPrefix][\LTLformula])}
\newcommand{\iAPaths}{\PthSet_\text{init}(\Game[\varphi])}

\newcommand{\Apath}{\pthElm}

\newcommand{\aQhistory}{{\hstElm}}

\newcommand{\APlays}{\PlaySet(\Game[\varphi])}

\newcommand{\Aplay}{\playElm}

\newcommand{\snd}{\text{snd}}
\newcommand{\last}{\lst}

\usepackage{tikz}
\usetikzlibrary{arrows,shapes,calc,patterns}

\usepackage{float}
\usepackage{wrapfig}

\usepackage{pgf}
\usepackage{pgfplots}
\usepackage{pgfplotstable}

\usepackage{multirow}
\usepackage{multicol}

\usepackage{rotating}

\tikzstyle{every node} =
  [draw = none, fill = white, thin]
\tikzstyle{every edge} +=
  [black, thick]

\tikzstyle{noall} =
  [draw = none, fill = none]
\tikzstyle{nodraw} =
  [draw = none, fill = white]
\tikzstyle{nofill} =
  [draw = black, fill = none]

\tikzstyle{nonode} =
  [draw = white]
\tikzstyle{cnode} =
  [circle, draw = black, draw = blue!50!black]
\tikzstyle{snode} =
  [regular polygon, regular polygon sides = 4, draw = red!50!black, inner
  sep = 0.125em]
\tikzstyle{lnode} =
  [diamond, draw = gray!75]
\tikzstyle{pnode} =
  [regular polygon, regular polygon sides = 5, draw = gray]

\tikzstyle{eloisenode} =
[align=center, draw=black, inner sep=0.5em, circle, text centered,
font=\sffamily, text width=1.5em] 

\tikzstyle{abelardnode} =
[align=center, draw=black, inner sep=-0.25em, rectangle, text centered,
font=\sffamily, text width=1.5em] 

\tikzstyle{dots} =
[align=center, circle, text centered,
font=\sffamily]

\tikzstyle{level} = [sibling distance = 8em/(#1), level distance = 5em]

\AfterEndPreamble
  {

  \newcommand{\fighypasgord}
    {
    \begin{tikzpicture}[>=stealth']

      \draw[very thick] (0, 0) rectangle (5, 5.75);

      \draw (-0.25, 6) node [below right] {$\AsgFam[1]$};

      \draw[thick, color = red!50!black] (0.5, 4) rectangle (4.5, 5.25);
      \draw[thick, color = red!50!black] (0.5, 2.25) rectangle (4.5, 3.5);
      \draw[thick, color = red!50!black] (0.5, 0.5) rectangle (4.5, 1.75);

      \draw (0.25, 5.5) node [below right] {$\XOne{1}$};
      \draw (0.25, 3.75)  node [below right] {$\XOne{2}$};
      \draw (0.25, 2) node [below right] {$\XOne{3}$};

      \draw (1.75, 4.625) node {$\asgElm[1]\;\asgElm[3]\;\asgElm[5]$};
      \draw (1.75, 2.875) node {$\asgElm[6]\;\asgElm[7]\;\asgElm[8]$};
      \draw (1.75, 1.125) node {$\asgElm[7]\;\asgElm[8]\;\asgElm[10]$};

      \draw[color = blue!50!black, dashed] (3, 4.25) rectangle (4.25, 5);
      \draw[color = blue!50!black, dashed] (3, 2.5) rectangle (4.25, 3.25);
      \draw[color = blue!50!black, dashed] (3, 0.75) rectangle (4.25, 1.5);

      \draw (3.625, 4.625) node {\color{blue!50!black}$\asgElm[2]\;\asgElm[4]$};
      \draw (3.625, 2.875) node {\color{blue!50!black}$\asgElm[2]\;\asgElm[4]$};
      \draw (3.625, 1.125) node {\color{blue!50!black}$\asgElm[6]\;\asgElm[9]$};

      \draw (6, 2.875) node {\LARGE$\sqsubseteq$};

      \draw[very thick] (7, 0) rectangle (12, 5.75);

      \draw (6.75, 6) node [below right] {$\AsgFam[2]$};

      \draw[thick, color = blue!50!black] (7.5, 4) rectangle (11.5, 5.25);
      \draw[thick, color = blue!50!black] (7.5, 2.25) rectangle (11.5, 3.5);
      \draw[thick] (7.5, 0.5) rectangle (11.5, 1.75);

      \draw (7.25, 5.5) node [below right] {$\XTwo{1}$};
      \draw (7.25, 3.75) node [below right] {$\XTwo{2}$};
      \draw (7.25, 2) node [below right] {$\XSet[23]$};

      \draw (9.5, 4.625) node {\color{blue!50!black}$\asgElm[2]\;\asgElm[4]$};

      \draw (9.5, 2.875) node {\color{blue!50!black}$\asgElm[6]\;\asgElm[9]$};

      \draw (9.5, 1.125) node {$\asgElm[1]\;\asgElm[9]\;\asgElm[11]$};

      \draw[thick, ->] (4.30, 4.625) -- (7.45, 4.625);
      \draw[thick, ->] (4.30, 2.875) -- (7.45, 4.625);
      \draw[thick, ->] (4.30, 1.125) -- (7.45, 2.875);

    \end{tikzpicture}
    }

  \newcommand{\figbhvfun}
    {
    \newcommand{\ET}{{\color{blue!50!black}\Tt}}
    \newcommand{\EF}{{\color{blue!50!black}\Ff}}
    \newcommand{\DT}{{\color{red!50!black}\Tt}}
    \newcommand{\DF}{{\color{red!50!black}\Ff}}
    \newcommand{\XT}{{\color{green!35!black}\Tt}}
    \newcommand{\XF}{{\color{green!35!black}\Ff}}
    \(
      \begin{matrix}
                                        & 0   & 1   & 2   & 3   & 4   & 5   & \\
        \asgElm[1] = \{\; \apElm \colon & \ET & \EF & \EF & \ET & \DF & \DT &
          \cdots \;\} \\
        \asgElm[2] = \{\; \apElm \colon & \ET & \EF & \EF & \ET & \XT & \XF &
          \cdots \;\} \\[0.10em]
        \hline\\[-0.80em]
        \FFun[\ASym](\asgElm[1]) =      & \EF & \EF & \ET & \DF & \DT & \DF &
          \cdots \\
        \FFun[\ASym](\asgElm[2]) =      & \EF & \EF & \ET & \XT & \XF & \XT &
          \cdots \\[0.10em]
        \hline\\[-0.80em]
        \FFun[\BSym](\asgElm[1]) =      & \EF & \ET & \ET & \EF & \DT & \DF &
          \cdots \\
        \FFun[\BSym](\asgElm[2]) =      & \EF & \ET & \ET & \EF & \XF & \XT &
          \cdots \\[0.10em]
        \hline\\[-0.80em]
        \FFun[\SSym](\asgElm[1]) =      & \ET & \ET & \EF & \EF & \ET & \XF &
          \cdots \\
        \FFun[\SSym](\asgElm[2]) =      & \ET & \ET & \EF & \EF & \ET & \XT &
          \cdots
      \end{matrix}
    \)
    }

  \newcommand{\figeqvasg}
    {
    \newcommand{\ET}{{\color{blue!50!black}\Tt}}
    \newcommand{\EF}{{\color{blue!50!black}\Ff}}
    \newcommand{\DT}{{\color{red!50!black}\Tt}}
    \newcommand{\DF}{{\color{red!50!black}\Ff}}
    \newcommand{\XT}{{\color{green!35!black}\Tt}}
    \newcommand{\XF}{{\color{green!35!black}\Ff}}
    \(
      \begin{matrix}
                      &                 & 0   & 1   & 2   & 3   & 4   & 5   & \\
        \multirow{2}{*}{\asgElm[1] = \Big\{}
                      & \apElm \colon   & \ET & \EF & \EF & \ET & \DF & \DT &
                      \cdots  & \multirow{2}{*}{\Big\}} \\
                      & \qapElm \colon  & \EF & \EF & \ET & \DT & \DF & \DT &
                      \cdots \\
        \multirow{2}{*}{\asgElm[2] = \Big\{}
                      & \apElm \colon   & \ET & \EF & \EF & \ET & \XT & \XF &
                      \cdots  & \multirow{2}{*}{\Big\}} \\
                      & \qapElm \colon  & \EF & \EF & \ET & \DT & \DF & \DT &
                      \cdots \\
        \multirow{2}{*}{\asgElm[3] = \Big\{}
                      & \apElm \colon   & \ET & \EF & \EF & \ET & \XT & \XF &
                      \cdots  & \multirow{2}{*}{\Big\}} \\
                      & \qapElm \colon  & \EF & \EF & \ET & \XF & \XT & \XT &
                      \cdots
      \end{matrix}
    \)
    }

  \newcommand{\figqntgam}
    {
    \begin{tikzpicture}[->, >=stealth', every node/.style={scale=0.85}, grow =
      right]

      \node [cnode] (E) {$\emptyfun$}
        child{
          node [snode] {$\stackrel{\Ff}{\apElm[1]}$}
          child{
            node [cnode] {\footnotesize
              $\stackrel{\Ff}{\apElm[1]}\stackrel{\Ff}{\apElm[2]}$}
            child{
              node [dots] {$\cdots$}
              child{
                node [snode] (P1) {$\valElm[1]$}
              }
            }
            child{
              node [dots] {$\cdots$}
              node [xshift=5.875em, yshift=1.125em] {$\vdots$}
            }
            edge from parent node [below] {$\Ff$}
          }
          child{
            node [cnode] {\footnotesize
              $\stackrel{\Ff}{\apElm[1]}\stackrel{\Tt}{\apElm[2]}$}
            child{
              node [dots] {$\cdots$}
            }
            child{
              node [dots] {$\cdots$}
              node [snode, xshift=5.875em, yshift=0.8175em] (Pj) {$\valElm[j]$}
            }
            edge from parent node [above] {$\Tt$}
          }
          edge from parent node [below] {$\Ff$}
        }
        child{
          node [snode] {$\stackrel{\Tt}{\apElm[1]}$}
          child{
            node [cnode] {\footnotesize
              $\stackrel{\Tt}{\apElm[1]}\stackrel{\Ff}{\apElm[2]}$}
            child{
              node [dots] {$\cdots$}
            }
            child{
              node [dots] {$\cdots$}
              node [xshift=5.875em, yshift=1.125em] {$\vdots$}
            }
            edge from parent node [below] {$\Ff$}
          }
          child{
            node [cnode] {\footnotesize
              $\stackrel{\Tt}{\apElm[1]}\stackrel{\Tt}{\apElm[2]}$}
            child{
              node [dots] {$\cdots$}
            }
            child{
              node [dots] {$\cdots$}
              child{
                node [snode] (Pn) {$\valElm[n]$}
              }
            }
            edge from parent node [above] {$\Tt$}
          }
          edge from parent node [above] {$\Tt$}
        }
      ;

      \node [above right of = Pn, node distance = 3em] (Px) {};

      \draw[-] (P1) -| (Px.center);
      \draw[-] (Pj) -| (Px.center);
      \draw[-] (Pn) -| (Px.center);
      \draw[->] (Px.center) -| (E);

    \end{tikzpicture}
    }

  }

\hypersetup
  {
  pdftitle  = {Good-for-Game QPTL: An Alternating Hodges Semantics},
  pdfauthor = {D. Bellier, M. Benerecetti, D. Della Monica, F. Mogavero}
  }

\begin{document}

  \title{\LARGE Good-for-Game QPTL: An Alternating Hodges Semantics}

  \author{
    \IEEEauthorblockN{Dylan Bellier}
    \IEEEauthorblockA{Universit\'e Rennes 1}
    \and
    \IEEEauthorblockN{Massimo Benerecetti}
    \IEEEauthorblockA{Universit\`a di Napoli Federico II}
    \and
    \IEEEauthorblockN{Dario Della Monica}
    \IEEEauthorblockA{Universit\`a di Udine}
    \and
    \IEEEauthorblockN{Fabio Mogavero}
    \IEEEauthorblockA{Universit\`a di Napoli Federico II}}

  \maketitle



\begin{abstract}
  An extension of \QPTL is considered where \emph{functional dependencies} among
  the quantified variables can be restricted in such a way that their current
  values are \emph{independent of the future} values of the other variables.
  This restriction is tightly connected to the notion of \emph{behavioral
  strategies} in game-theory and allows the resulting logic to naturally express
  game-theoretic concepts.
  The fragment where only restricted quantifications are considered, called
  \emph{behavioral quantifications}, can be decided, for both \emph{model
  checking} and \emph{satisfiability}, in 2\ExpTime and is \emph{expressively
  equivalent} to \QPTL, though significantly \emph{less succinct}.
\end{abstract}



  \vspace{-0.5em}
  


\begin{section}{Introduction}


  The tight connection between \emph{logic} and \emph{games} has been
  acknowledged since the sixties, when first Lorenzen~\cite{Lor61} and later
  Lorenz~\cite{Lor68} and Hintikka~\cite{Hin73} proposed \emph{game-theoretic
  semantics} for first-order logic~\cite{HS97,Hin97}.
  In this approach, the meaning of a sentence is given in terms of a zero-sum
  game played by two agents: the \emph{verifier}, whose objective is to show the
  sentence true, and the \emph{falsifier}, with the dual objective of showing
  the sentence false.
  Satisfiability of a sentence, then, becomes a game between these two players
  and the sentence is satisfiable (\resp, unsatisfiable) \iff verifier (\resp,
  falsifier) has a strategy to win it.
  This tight connection can clearly be viewed in the other direction as well:
  logic can be used to reason about games, \ie, we can encode the problem of
  solving a game into a \emph{decision problem}, such as \emph{satisfiability}
  or \emph{model-checking}, of some logic.
  The idea is to describe the game and the \emph{winning condition} with a
  formula of the logic and exploit the \emph{game-theoretic interpretation} to
  reduce the solution of the game to a specific decision problem for that
  logic.
  Essentially, the winning strategy for the game can be extracted from the
  winning strategy for the \emph{decision game}.
  \\ \indent
  Suppose we have a formula $\psi(\xElm, \yElm)$, expressing a required relation
  between the choice $\yElm$ made by a player, from now on called \emph{Eloise},
  and a choice $\xElm$ made by the adversary, namely \emph{Abelard}, \ie,
  $\psi(\xElm, \yElm)$ encodes the \emph{objective} of a two-player game.
  We say that the game is won by Eloise if there exists a \emph{strategy} for
  her such that, for each choice $\xElm$ made by Abelard, the corresponding
  response $\yElm$ of Eloise using that strategy guarantees that the resulting
  play satisfies the requirement $\psi(\xElm, \yElm)$.
  This condition can clearly be expressed by a sentence of the form $\forall
  \xElm \ldotp \exists \yElm \ldotp \psi(\xElm, \yElm)$.
  We could then solve the game by solving the satisfiability problem for
  this sentence.
  In other words, solving the game reduces to checking whether there exists a
  \emph{Skolem function} $\fFun$ such that $\forall \xElm \ldotp \psi(\xElm,
  \fFun(\xElm))$ is satisfied.
  This function basically dictates the response of Eloise to the choice of
  Abelard, thereby encoding her strategy.
  \\ \indent
  The above approach works pretty well when we consider \emph{single-round
  games}, \aka, \emph{normal-form games}~\cite{NM44}, and can easily be extended
  to \emph{finite-rounds games}, \aka, \emph{extensive-form
  games}~\cite{Neu28,Kuh50,Kuh53}, by extending the quantification prefix to a
  sequence of alternations of quantifiers, one for each round.
  Things, however, get much more complicated when \emph{infinite-rounds games}
  come into play~\cite{GS53,Wol55}.
  For such a class of extensive-form games, indeed, plays are induced by
  infinite sequences of choices made by the players over time and a strategy
  dictates how a player at a given stage of a play responds to the choices
  made by the adversary up to that stage.
  %
  %
  Extending the quantification prefix to match the rounds would immediately lead
  to \emph{infinitary logics}, such as the one proposed in~\cite{Kol85} and
  further studied in~\cite{HV94} (see also~\cite{HR76}).
  This technique has some interesting applications in logic~\cite{Hel89},
  computer science~\cite{Kai11}, and even philosophy~\cite{FG17}.
  Besides its infinitary nature, however, this approach has also the drawback
  of heavily departing from the standard Tarskian viewpoint, as only
  non-compositional game-theoretic semantics has been provided.
  \\ \indent
  A more viable route, instead, is to make the quantified variables $\xElm$ and
  $\yElm$ range over sequences of choices.
  For example, if the choices are simply \emph{Boolean values}, \ie, when
  \emph{iterated Boolean games} are considered~\cite{GHW13,GHW15}, first-order
  extension of \emph{temporal logics}, such as \emph{Quantified Propositional
  Temporal Logic} (\QPTL)~\cite{Sis83}, seem like a good place to start, as they
  predicate over infinite sequences of temporal points, the stages of the game.
  In this setting, however, the Skolem function $\fFun$ cannot be interpreted as
  a strategy in the game-theoretic sense, as its value at a given stage depends
  on the entire evaluation of its argument $\xElm$, namely the entire sequence
  of choices made by the adversary, including all the future ones.
  By contrast, a strategy for a player can only dictate, step by step, what its
  responses should be, depending on the choices made so far by its opponent.
  What that means is that, in principle, the satisfiability and the game
  solution problems do not coincide anymore.
  A classic example of this problem already appears in~\cite{PR89}.
  Assume $\psi(\xElm, \yElm)$ is the \LTL~\cite{Pnu77,Pnu81} formula $\G (\yElm
  \leftrightarrow \X \xElm)$ (or just $\yElm \leftrightarrow \X \xElm$).
  Clearly, the sentence $\forall \xElm \ldotp \exists \yElm \ldotp \psi(\xElm,
  \yElm)$ is satisfiable.
  However, there is no \emph{``feasible''} (\ie, \emph{implementable}) strategy
  that can enforce $\psi(\xElm, \yElm)$, without Eloise knowing in advance the
  future values that Abelard is going to choose for $\xElm$ in the rest of the
  play.
  The problem is that the standard interpretation of quantification treats the
  quantified objects as \emph{atomic entities}, regardless of their \emph{inner
  structure}, like their being sequences in the above example.
  This is by no means the only exemplification of the problem, which was already
  recognized in the theory of extensive-form games since its dawn~\cite{Kuh50},
  where the notion of feasible strategy, called \emph{behavioral}, has been
  introduced (see also~\cite{Kuh53,Sel75,KMRW82,Mye91}).
  Another important source of unfeasible strategic behaviors is hidden in the
  semantics of \emph{Strategy Logic} (\SL)~\cite{CHP07,MMV10a,CHP10,MMPV14}, an
  extension of \LTL that allows for explicit \emph{quantifications} over
  strategies and \emph{binding} of strategies with players.
  In this logic, formulae can be written that can be satisfied only by allowing
  players to look at what other strategies dictate in the future or
  counterfactual situations~\cite{MMS13}, admitting infeasible behaviors.
  Once again, the problem lies in the intrinsic dependence among the variables
  quantified in the formula.
  \\ \indent
  One way to reconcile quantifications and strategies in a temporal setting
  would be to extend the game-theoretic interpretation of the quantifiers, and
  of the logic in general, to account for the underlying temporal dynamics.
  This would imply allowing the players in the satisfiability game to play with
  partial information on the choices of the adversary, namely the players have
  no information about the future and can only choose based on the moves played
  so far in the game.
  Previous attempts to address the issue typically involve resorting to \adhoc
  \emph{Skolem semantics}~\cite{Hod97} for the specific logic.
  In the case of \SL, for instance, the notion of \emph{behavioral semantics}
  has been introduced~\cite{MMPV14}, which prevents the players from looking at
  future choices when selecting their strategy, effectively limiting the player
  observation ability to the current history in the game.
  A more liberal semantics based on \emph{timeline dependencies} has been also
  proposed~\cite{GBM18,GBM20}.
  While these approaches do solve the problem in the specific case, they lead to
  non-compositional semantics~\cite{SH01}, in that the interpretation of a
  formula is not defined in terms of the interpretation of its component
  subformulae.
  To obtain a compositional version of the game-theoretic semantics, a finer
  grained technical setting is required, compared to the classic Tarskian
  semantics, specifically, one that can accommodate some form of \emph{partial
  independence} among the quantified variables.
  \\ \indent
  Following Tarski approach, each choice for a quantified variable in a sentence
  is made with complete information about, hence it is (potentially) completely
  dependent on, the values of variables quantified before it in the
  sentence.
  This idiosyncrasy of the classic interpretation of quantifiers is well known
  and attempts have been made to overcome the \emph{linear dependence} of
  quantifiers dictated by their relative  position in a
  sentence~\cite{Hen61,Hin73a,BG86,SV92}.
  Most notably, Hintikka and Sandu~\cite{HS89} proposed
  \emph{Independence-Friendly Logic} (\IF), as a first-order logic where
  independence between quantified variables can be explicitly asserted in the
  formulae together with a game-theoretic, non-compositional,
  semantics~\cite{San93} for the logic.
  A compositional semantics for \IF was later proposed by
  Hodges~\cite{Hod97a,Hod97b}, whose idea was to replace the standard notion of
  \emph{assignment} of the Tarskian semantics with that of a set of assignments
  (called \emph{trump}~\cite{Hod97a} or \emph{team}~\cite{Vaa07}), as the basic
  semantic element with respect to which the truth of a formula is evaluated.
  This multiplicity of assignments effectively allows one to express the notion of
  \emph{dependence/independence} among variables, a distinction that makes very
  little sense, in particular from a formal point of view, when only a single
  assignment is considered.
  \\ \indent
  Taking inspiration from Hodges' work, the goal of this work is to devise a
  \emph{compositional semantic framework} that can account for a game-theoretic
  interpretation of quantification over (possibly infinite) sequences of
  choices.
  The framework is specifically tailored to deal with quantifications in a
  linear-time settings and applied to the logic \QPTL, which was introduced
  in~\cite{Sis83} as a unifying \emph{$\omega$-regular language} allowing for
  both temporal operators and propositional quantifiers.
  Despite its expressiveness and theoretical interest, \QPTL has not gained much
  traction in practical contexts, mainly due to the high complexity of its
  decision problems.
  Indeed, both the satisfiability and the model checking problems are
  \emph{non-elementary} in the number of \emph{alternations of the
  quantifiers}~\cite{SVW87}.
  \\ \indent
  In this article we propose a novel semantics for \QPTL, inspired by the body
  of work on (in)dependence logics~\cite{Vaa07,MSS11,AKVV16}.
  Similarly to those works, the semantics provides a compositional
  formulation~\cite{SH01} for a game-theoretic interpretation of the
  quantifiers.
  In contrast to them, however, we require a \emph{symmetric} treatment of the
  two quantifiers in order to preserve  \emph{closure under negation} and
  avoid \emph{undetermined formulae}~\cite{Hod97a,Hod97b}.
  The most significant feature of the new approach is the ability to encode
  various forms of \emph{independence constraints} among the quantified
  variables and provide a powerful tool to fine-tune the semantics of the
  propositional quantifiers.  In particular, we discuss a specific instantiation
  of the semantics that allows one to recover a game-theoretic interpretation of
  the quantifiers and reconcile the satisfiability and the game solution
  problems.
  This result is achieved by first generalizing classic \emph{temporal
  assignments}, which give values to propositional variables at each time
  instant, to sets of sets of assignments, called \emph{hyperassignments}.
  This also generalizes teams, defined as sets of assignments, used by Hodges.
  The second step is to introduce new \emph{classes of functors} that maps
  temporal assignments to valuations of a given variable over time and,
  intuitively, correspond to the semantic counterparts of the Skolem functions.
  The dependence of functors on assignments allows us to impose various forms of
  independence constraints among the variables.
  In particular, we investigate two specific forms, called \emph{behavioral} and
  \emph{strongly-behavioral}, that require functors to choose the value of the
  variable at any given time instant based only upon the values dictated by the
  input assignment to the other variables up to that instant (possibly
  excluded).
  These are forms of independence constraints that make the choice of the value
  of a variable at a given time totally independent of the values that other
  variables assume in the future.
  The behavioral restrictions are precisely what allows us to recover the
  correspondence between Skolem functions and strategies and to reconcile the
  satisfiability and game solution problems, thus making the resulting version
  of \QPTL, called \emph{Good-for-Games \QPTL} (\GFG-\QPTL), well suited to
  express game-theoretic concepts and a logical analogue of \emph{Good-for-Games
  Automata}~\cite{HP06,BK19}.
  \\ \indent
  On the technical side, the novel semantics under the behavioral interpretation
  of the quantifiers leads to 2\ExpTime decision procedures for both the
  satisfiability and model-checking problems.
  On the other hand, it does not give up expressiveness, as we show that the
  vanilla and behavioral semantics turn out to be expressively equivalent.
  These results also show that the high complexity of the decision problems for
  vanilla \QPTL stems from the fact that unrestricted dependencies among the
  quantified variables are allowed.
  The properties expressible by exploiting such unrestricted dependencies can,
  however, still be expressed under the behavioral semantics via encoding of
   $\omega$-regular automata, though with a non-elementary blowup.

\end{section}






\section{Alternating Hodges Semantics}
\label{sec:althodsem}

\QPTL~\cite{Sis83} extends \LTL~\cite{Pnu77,Pnu81} with quantifications over
\emph{atomic propositions} from a given set $\APSet$, with the intuition that
the Boolean values of the same proposition in different time instants are
independent of each other.

\subsection{Quantified Propositional Temporal Logic}

For convenience, we provide a syntax for \QPTL where quantifications do not
occur within temporal operators.
This is equivalent to the original logic, thanks to the \emph{prenex normal
form} (\pnf, for short) property enjoyed by \QPTL~\cite{Sis83}, which allows to
move quantifiers outside temporal operators.

\begin{definition}[\QPTL Syntax]
  \label{def:syn}
  The \emph{Quantified Propositional Temporal Logic} is the set of formulae
  built accordingly to the following context-free grammar, where $\psi \in \LTL$
  and $\apElm \in \APSet$:
  \(
    \varphi \seteq \psi \mid \neg \varphi \mid \varphi \wedge \varphi \mid
    \varphi \vee \varphi \mid \exists \apElm \ldotp \varphi \mid \forall \apElm
    \ldotp \varphi.
  \)
\end{definition}

The classic semantics is given in terms of \emph{temporal assignments} (simply
\emph{assignments}, from now on), which are functions associating each
proposition with a \emph{temporal valuation} mapping each time instant to a
Boolean value, \ie, infinite sequences of truth assignments.
Let $\AsgSet \defeq \allowbreak \APSet \pto (\SetN \to \SetB)$ denote the set of
assignments over arbitrary subsets of $\APSet$.
For convenience, we also introduce the set of assignments defined exactly over
the propositions in $\PSet \!\subseteq\! \APSet$, \ie, $\AsgSet(\PSet)
\!\defeq\! \set{\! \asgElm \!\in\! \AsgSet }{ \dom{\asgElm} \!=\! \PSet \!}$ and
the set $\AsgSet[\subseteq](\PSet) \!\defeq\! \set{\! \asgElm \!\in\! \AsgSet }{
\PSet \!\subseteq\! \dom{\asgElm} \!}$ of assignments defined at least over
$\PSet$.
The \emph{satisfaction relation} $\models$ between an assignment $\asgElm$ and a
\QPTL formula $\varphi$ is defined below, where $\cmodels[\LTL]$ is the standard
\LTL satisfiability and ${\asgElm}[\apElm \mapsto \fFun]$ denotes the assignment
that extends $\asgElm$ and maps proposition $\apElm$ to temporal valuation
$\fFun$.
As usual, by $\free{\varphi}$ we denote the set of propositions free in
$\varphi$.

\begin{definition}[Tarski Semantics]
  \label{def:tarsem}
  The \emph{Tarski-semantics relation} $\asgElm \models \varphi$ is inductively
  defined as follows, for all \QPTL formulae $\varphi$ and assignments $\asgElm
  \in \AsgSet[\subseteq](\free{\varphi})$.
  \begin{enumerate}
    \item\label{def:tarsem(ltl)}
      $\asgElm \models \psi$, if $\asgElm \cmodels[\LTL] \psi$, whenever $\psi$
      is an \LTL formula;
    \item\label{def:tarsem(bln)}
      the semantics of Boolean connectives is defined as usual;
    \item\label{def:tarsem(qnt)}
      for all atomic propositions $\apElm \in \APSet$:
      \begin{enumerate}
        \item\label{def:tarsem(qnt:exs)}
          $\asgElm \models \exists \apElm \ldotp \phi$ if ${\asgElm}[\apElm
          \mapsto \fFun] \models \phi$, for some $\fFun \in \SetN \to \SetB$;
        \item\label{def:tarsem(qnt:all)}
          $\asgElm \models \forall \apElm \ldotp \phi$ if ${\asgElm}[\apElm
          \mapsto \fFun] \models \phi$, for all $\fFun \in \SetN \to \SetB$.
      \end{enumerate}
  \end{enumerate}
\end{definition}

\subsection{A New Semantics for \QPTL}

We now introduce a novel compositional semantics for \QPTL that, unlike Tarski's
one, will allow us to specify, later on, independence constraints among
quantified propositions.
The new semantics follows an approach similar to~\cite{Hod97a}, where a
compositional semantics for \IF was first proposed.
Hodges' idea was to expand an assignment for the free variables to a set of
assignments, a trump in his terminology, with the intuition of capturing all
possible choices made by one of the two players for its own variables in the
satisfiability game underlying the game-theoretic semantics of the
logic~\cite{HS89}.
Hodges' semantics, though able to correctly capture \IF, is, however, not
adequate for our purposes.
Indeed, by design, it is intrinsically asymmetric, treating the two players
differently.
More specifically, a single set of assignments only provides complete
information about the choices of one of the two players and only allows to
restrict the choices of the adversary.
This, in turn, limits the class of games expressible in the logic to asymmetric
games, where only the observation power of one player can be restricted.
To also capture symmetric games, we need to get rid of this asymmetry, which
requires a non-trivial generalization of Hodges' approach.

To give semantics to a \QPTL formula $\varphi$, we proceed as follows.
Similarly to Hodges, the idea is that the interpretations of the free atomic
propositions correspond to the choices that the two players could make prior to
the current stage of the game, \ie, the stage where the formula $\varphi$ has
still to be evaluated.
These possible choices can be organized on a two-level structure, \ie, a set of
sets of assignments, each level summarizing the information about the choices a
player can make in its turns.
In order to evaluate the formula $\varphi$, then, a player chooses a set of
assignment, while its opponent chooses one assignment in that set where
$\varphi$ must hold.
We shall use a flag $\alpha \in \{ \QEA, \QAE \}$, called \emph{alternation
flag}, to keep track of which player is assigned to which level of choice.
If $\alpha = \QEA$, Eloise chooses the set of assignments, while Abelard chooses
one of those assignments; if $\alpha = \QAE$, the dual reasoning applies.
Given a flag $\alpha \in \{ \QEA, \QAE \}$, we denote by $\dual{\alpha}$ the
dual flag, \ie, $\dual{\alpha} \in \{ \QEA, \QAE \}$ with $\dual{\alpha} \neq
\alpha$.
The idea above is captured by the following notion of \emph{hyperassignment},
namely a non-empty set of non-trivial, \ie, non-empty, sets of assignments
defined over an arbitrary set $\PSet \subseteq \APSet$:
\[
  \HAsgSet \defeq \set{ \AsgFam \subseteq \pow{\AsgSet(\PSet)} }{ \emptyset
  \!\not\in\! \AsgFam \!\neq\! \emptyset \land \PSet \!\subseteq\! \APSet }.
\]
Note that we require all the assignments contained in a hyperassignment to be
defined on the same atomic propositions, though the domains of assignments in
different hyperassignment may differ.
By $\ap{\AsgFam} \subseteq \APSet$ we denote the set of atomic propositions over
which the hyperassignment $\AsgFam$ is defined.
$\HAsgSet(\PSet) \defeq \allowbreak \set{\! \AsgFam \in \HAsgSet }{ \AsgFam
\subseteq \pow{\AsgSet(\PSet)} \!}$ is the set of hyperassignments over the same
set of atomic propositions $\PSet$, while the set whose hyperassignments have
domains that include $\PSet$ is $\HAsgSet[\subseteq](\PSet) \defeq \allowbreak
\set{\! \AsgFam \in \HAsgSet }{ \AsgFam \subseteq
\pow{\AsgSet[\subseteq](\PSet)} \!}$.

\begin{wrapfigure}{r}{0.45\linewidth}
  \vspace{-1.5em}
  \newcommand{\XOne}[1]{{\color{red!50!black}\XSet[1#1]}}
  \newcommand{\XTwo}[1]{{\color{blue!50!black}\XSet[2#1]}}
  \begin{center}
    \scalebox{0.60}[0.50]{\fighypasgord}
  \end{center}
  \vspace{-1.0em}
  \caption{\label{fig:hypasgord} The preorder between hyperassignments: for
  every $\XOne{i} \!\in\! \AsgFam[1]$, there is a $\XTwo{j} \!\in\! \AsgFam[2]$
  with $\XTwo{j} \!\subseteq\! \XOne{i}$. More than one set in $\AsgFam[1]$
  could choose the same set in $\AsgFam[2]$; there may be sets in $\AsgFam[2]$
  not used by any in $\AsgFam[1]$.}
  \vspace{-1.5em}
\end{wrapfigure}

For any pair of hyperassignments $\AsgFam[1], \AsgFam[2] \in \HAsgSet$, we write
$\AsgFam[1] \sqsubseteq \AsgFam[2]$ to express the fact that, for all sets of
assignments $\XSet[1] \in \AsgFam[1]$, there is a set of assignments $\XSet[2]
\in \AsgFam[2]$ with $\XSet[2] \subseteq \XSet[1]$.
Obviously, $\AsgFam[1] \subseteq \AsgFam[2]$ implies $\AsgFam[1] \sqsubseteq
\AsgFam[2]$, which in its turn implies $\ap{\AsgFam[1]} = \ap{\AsgFam[2]}$.
Figure~\ref{fig:hypasgord} reports a graphical representation of the relation
$\sqsubseteq$.
As usual, we write $\AsgFam[1] \equiv \AsgFam[2]$ if both $\AsgFam[1]
\sqsubseteq \AsgFam[2]$ and $\AsgFam[2] \sqsubseteq \AsgFam[1]$ hold true.
It is clear that the relation $\sqsubseteq$ is both reflexive and transitive,
hence it is a preorder.
Consequently, $\equiv$ is an equivalence relation.
In particular, we shall show (see Corollary~\ref{cor:eqv}) that $\equiv$
captures the intuitive notion of equivalence between hyperassignments, in the
sense that two equivalent hyperassignments \wrt $\equiv$ do satisfy the same
formulae.


Our goal is to define a semantics for \QPTL by providing a satisfaction relation
between a hyperassignment $\AsgFam$ and a \QPTL formula $\varphi$, \wrt a given
interpretation of the players of $\AsgFam$, \ie, \wrt an alternation flag
$\alpha \in \{ \QEA, \QAE \}$.
Therefore, we shall have two satisfaction relations, namely $\cmodels[][\QEA]$
and $\cmodels[][\QAE]$, depending on how we interpret the levels of the
hyperassignment.
The idea is to capture the following intuition that relates, in a natural way,
to the classic Tarskian semantics.
When the alternation flag $\alpha$ is $\QEA$, then a set of assignments is
chosen existentially by Eloise and all its assignments, chosen universally by
Abelard, must satisfy $\varphi$.
Conversely, when $\alpha$ is $\QAE$, then a set of assignments is chosen
universally by Abelard and at least one assignment, chosen existentially by
Eloise, must satisfy $\varphi$.
Semantically, hyperassignments are a notion similar to
\emph{quasi-strategies}~\cite{Kol85}.

We break down the presentation of the semantics by introducing three operations:
the \emph{dualization} swaps the role of the two players in a hyperassignment,
allowing for connecting the two satisfaction relations and a symmetric treatment
of quantifiers later on; the \emph{partitioning} deals with disjunction and
conjunction; finally, the \emph{extension} directly handles quantifications.

Let us consider the dualization operator first.
The idea is that, given a hyperassignment $\AsgFam$, the dual hyperassignment
$\dual{\AsgFam}$ exchanges the role of the two players \wrt $\AsgFam$.
This means that, if Eloise is the first to choose in $\AsgFam$, then her choice
will be postponed in $\dual{\AsgFam}$ after that of Abelard.
To ensure that, in exchanging the order of choice for the two players, we do not
alter the semantics of the underlying game, we need to reshuffle the assignments
in $\AsgFam$ so as to simulate the original dependencies between the choices of
the players.
To this end, we introduce the set of choice functions for $\AsgFam$ as follows,
whose definition implicitly assumes the axiom of choice:
\[
  \ChcSet{\AsgFam} \defeq \set{ \chcFun \colon \AsgFam \to \AsgSet }{ \forall
  \XSet \in \AsgFam \ldotp \chcFun(\XSet) \in \XSet }.
 \]
$\ChcSet{\AsgFam}$ contains all the functions $\chcFun$ that, for every set of
assignments $\XSet$ in $\AsgFam$, pick a specific assignment $\chcFun(\XSet)$ in
that set.
Each such function simulates a possible choice of the second player of $\AsgFam$
depending on the choice of (the set of assignments chosen by) its first player.
The dual hyperassignment $\dual{\AsgFam}$, then, collects the images of
the choice functions in $\ChcSet{\AsgFam}$.
We, thus, obtain a hyperassignment in which the 
choice order of the two players is inverted:
\[
  \dual{\AsgFam} \defeq \set{ \img{\chcFun} }{ \chcFun \in \ChcSet{\AsgFam} }.
\]

\newcommand{\XOne}{{\color{red!50!black}\XSet[1]}}
\newcommand{\XTwo}{{\color{blue!50!black}\XSet[2]}}
\newcommand{\XThree}{{\color{green!25!black}\XSet[3]}}
\newcommand{\asgone}[1]{{\color{red!50!black}\asgElm[#1]}}
\newcommand{\asgtwo}[1]{{\color{blue!50!black}\asgElm[#1]}}
\newcommand{\asgthree}[1]{{\color{green!25!black}\asgElm[#1]}}
\begin{wrapfigure}[4]{r}{0.5\textwidth}
\vspace{-2em}
{\small \[
  \AsgFam \!=\!
  \begin{Bmatrix}
    \XOne = \{ \asgone{11}, \asgone{12} \}, \\
    \XTwo = \{ \asgtwo{21}, \asgtwo{22} \}, \\
    \XThree = \{ \asgthree{3} \}
  \end{Bmatrix}
  \dual{\AsgFam} \!=\!
  \begin{Bmatrix}
    \img{\chcFun[1]} = \{ \asgone{11}, \asgtwo{21}, \asgthree{3} \}, \\
    \img{\chcFun[2]} = \{ \asgone{11}, \asgtwo{22}, \asgthree{3} \}, \\
    \img{\chcFun[3]} = \{ \asgone{12}, \asgtwo{21}, \asgthree{3} \}, \\
    \img{\chcFun[4]} = \{ \asgone{12}, \asgtwo{22}, \asgthree{3} \}\phantom{,}
  \end{Bmatrix}
\]}
\end{wrapfigure}
\noindent
\textbf{Example 1.}\setcounter{example}{1}
\emph{
Consider the hyperassignment $\AsgFam$ on the right, where $\XOne \!=\! \{
\asgone{11}, \asgone{12} \}$, $\XTwo \!=\! \{ \asgtwo{21}, \asgtwo{22} \}$,
and $\XThree \!=\! \{ \asgthree{3} \}$.
Every set of assignments in $\dual{\AsgFam}$ is obtained as the image of one of
the four choice functions $\chcFun[i] \!\in\! \ChcSet{\AsgFam}$, each choosing
exactly one assignment from $\XOne$, one from $\XTwo$, and one from $\XThree$.
Intuitively, in $\AsgFam$ the strategy of the first player, say Eloise, can only
choose the color of the final assignments (either red for \XOne, blue for \XTwo,
or green for \XThree), while the one for Abelard decides which assignment of
each color will be picked.
After dualization, the two players exchange the order in which they choose.
Therefore, Abelard, starting first in $\dual{\AsgFam}$, will select one of the
four choice functions, which picks an assignment for each color.
Eloise, choosing second, by using her strategy that selects the color will give
the final assignment.
In other words, the original strategies of the players encoded in the
hyperassignment, as well as their dependencies, are preserved, regardless of the
swap of their role in the dual hyperassignment.
}

The following proposition ensures that the dualization operator enjoys an
\emph{involution property}, similarly to the Boolean negation: by applying the
dualization twice, we obtain a hyperassignment equivalent to the original one.

\defenv{proposition}[][Prp][DltInv]
  {
  $\AsgFam \subseteq \dual{\dual{\AsgFam}}$ and $\AsgFam \equiv
  \dual{\dual{\AsgFam}}$, for all $\AsgFam \in \HAsgSet$.
  }

Note that there is a clear analogy between the structure of hyperassignments
with alternation flag $\QEA$ (\resp, $\QAE$) and the structure of DNF (\resp,
CNF) formulae, where the dualization swaps between two equivalent forms.
The following lemma formally states that the dualization swaps the role of the
two players while still preserving the original dependencies among their
choices.

\defenv{lemma}[Dualization][Lmm][Flp]
  {
  The following hold true, for all \QPTL formulae $\varphi$ and hyperassignments
  $\AsgFam \in \HAsgSet[\subseteq](\free{\varphi})$.
  \begin{minipage}{0.5\textwidth}
  \noindent
  \begin{enumerate}
    \item\labelx{ea}
      Statements~\ref{lmm:flp(ea:org)} and~\ref{lmm:flp(ea:dlt)} are equivalent:
      \begin{enumerate}
        \item\labelx{ea:org}
          there exists $\XSet \in \AsgFam$ such that $\asgElm \models \varphi$,
          for all $\asgElm \in \XSet$;
        \item\labelx{ea:dlt}
          for all $\XSet \in \dual{\AsgFam}$, $\asgElm \models \varphi$, for
          some $\asgElm \in \XSet$.
      \end{enumerate}
  \end{enumerate}
  \end{minipage}
  \hspace{-1em}
  \begin{minipage}{0.5\textwidth}
  \begin{enumerate}
    \setcounter{enumi}{1}
    \item\labelx{ae}
      Statements~\ref{lmm:flp(ae:org)} and~\ref{lmm:flp(ae:dlt)} are equivalent:
      \begin{enumerate}
        \item\labelx{ae:org}
          for all $\XSet \in \AsgFam$, $\asgElm \models \varphi$, for some
          $\asgElm \in \XSet$;
        \item\labelx{ae:dlt}
          there exists $\XSet \in \dual{\AsgFam}$ such that $\asgElm \models
          \varphi$, for all $\asgElm \in \XSet$.
      \end{enumerate}
    \end{enumerate}
  \end{minipage}
  }

The \emph{partition operator} decomposes hyperassignments and is instrumental in
capturing the semantics of Boolean connectives.
Given a hyperassignment $\AsgFam$, the following set
\[
  \parFun[]{\AsgFam} \defeq \set{ (\AsgFam[1], \AsgFam[2]) \in \pow{\AsgFam}
  \times \pow{\AsgFam} }{ \AsgFam[1] \uplus \AsgFam[2] = \AsgFam }
\]
collects all the possible partitions of $\AsgFam$ into two parts.
Assume that the two players of $\AsgFam$ are interpreted according to the
alternation flag $\QAE$: Abelard chooses first and Eloise chooses second.
The game-theoretic interpretation of the disjunction requires Eloise to choose
one of two disjuncts to be proven true.
In our setting, then, in order to satisfy $\varphi_{1} \vee \varphi_{2}$, Eloise
has to show that, for each set of assignments chosen by Abelard, she has a way
to
select one of the disjuncts $\varphi_{i}$ in such a way that $\varphi_{i}$ is
satisfied by some assignment in that set.
This selection is summarized by one of the pairs $(\AsgFam[1], \AsgFam[2])$ in
$\parFun[]{\AsgFam}$, where $\AsgFam[i]$ collects the sets of assignments for
which the $i$-th disjunct is selected, with $i \in \{ 1, 2 \}$.
A similar argument, with the role of the two players reversed and switching the
quantifications throughout, leads to a dual interpretation for conjunction,
where it is Abelard who chooses one of the two conjuncts to be proven false.
This intuition is made precise by the following lemma.

\defenv{lemma}[Boolean Connectives][Lmm][BlnCon]
  {
  The following hold true, for all \QPTL formulae $\varphi_{1}$ and
  $\varphi_{2}$ and hyperassignments $\AsgFam \in \HAsgSet[\subseteq](\PSet)$,
  with $\PSet \defeq \allowbreak \free{\varphi_{1}} \cup \free{\varphi_{2}}$.\\
  \begin{minipage}{0.5\textwidth}
  \noindent
  \begin{enumerate}
    \item\labelx{ea}
      Statements~\ref{lmm:blncon(ea:org)} and~\ref{lmm:blncon(ea:par)} are
      equivalent:
      \begin{enumerate}
        \item\labelx{ea:org}
          there is $\XSet \!\in\! \AsgFam$ such that $\asgElm \!\models\!
          \varphi_{1} \!\wedge\! \varphi_{2}$, for all $\asgElm \!\in\! \XSet$;
        \item\labelx{ea:par}
          for each $(\AsgFam[1], \AsgFam[2]) \in \parFun[]{\AsgFam}$, there are
          $i \in \{ 1, 2 \}$ and $\XSet \in \AsgFam[i]$ such that $\asgElm
          \models \varphi_{i}$, for all $\asgElm \in \XSet$.
      \end{enumerate}
  \end{enumerate}
  \end{minipage}
  \hspace{-1em}
  \begin{minipage}{0.5\textwidth}
  \begin{enumerate}
    \setcounter{enumi}{1}
    \item\labelx{ae}
      Statements~\ref{lmm:blncon(ae:org)} and~\ref{lmm:blncon(ae:par)} are
      equivalent:
      \begin{enumerate}
        \item\labelx{ae:org}
          for all $\XSet \!\in\! \AsgFam$, $\asgElm \!\models\! \varphi_{1}
          \!\vee\! \varphi_{2}$, for some $\asgElm \!\in\! \XSet$;\!
        \item\labelx{ae:par}
          there is $(\AsgFam[1], \AsgFam[2]) \!\in\! \parFun[]{\AsgFam}$ such
          that, for all $i \in \{ 1, 2 \}$ and $\XSet \!\in\! \AsgFam[i]$,
          $\asgElm \!\models\! \varphi_{i}$, for some $\asgElm \!\in\! \XSet$.
      \end{enumerate}
    \end{enumerate}
  \end{minipage}
  }

Quantifications are taken care of by the \emph{extension operator}.
Let $\FncSet(\PSet) \defeq \AsgSet(\PSet) \to (\SetN \to \SetB)$ be the set of
\emph{functors} that maps every assignment over $\PSet$ to a temporal valuation.
Essentially, these objects play the role of Skolem functions in the
non-compositional semantics.
The \emph{extension of an assignment} $\asgElm \in \AsgSet(\PSet)$ \wrt a
functor $\FFun \in \FncSet(\PSet)$ for an atomic proposition $\apElm \in \APSet$
is defined as $\extFun{\asgElm, \FFun, \apElm} \defeq {\asgElm}[\apElm \mapsto
\FFun(\asgElm)]$.
Intuitively, it extends $\asgElm$ with $\apElm$, by assigning to it the value
$\FFun(\asgElm)$ prescribed by the functor $\FFun$.
The \emph{extension operation} can then be lifted to sets of assignments $\XSet
\!\subseteq\! \AsgSet(\PSet)$ in the obvious way, \ie, we set $\extFun{\XSet,
\FFun, \apElm} \defeq \set{ \extFun{\asgElm, \FFun, \apElm} }{ \asgElm \!\in\!
\XSet }$.
This operation embeds into $\XSet$ the entire player strategy encoded by
$\FFun$.
Finally, the \emph{extension of a hyperassignment} $\AsgFam \in \HAsgSet(\PSet)$
with $\apElm$ is simply the set of extensions with $\apElm$ of all its sets of
assignments \wrt all possible functors over the atomic propositions of
$\AsgFam$:
\[
  \extFun{\AsgFam, \apElm} \defeq \set{ \extFun{\XSet, \FFun, \apElm} }{ \XSet
  \in \AsgFam, \FFun \in \FncSet(\ap{\AsgFam}) }.
\]
Intuitively, this operation embeds into $\AsgFam$ all possible strategies,
encoded by the functors $\FFun$, for choosing the value of $\apElm$.
The following lemma states that the extension operator provides an adequate
semantics for quantifications, where statement~\ref{lmm:hypext(ea)} considers
Eloise's choices, when the player interpretation of the hyperassignment is
$\QEA$, and statement~\ref{lmm:hypext(ae)} takes care of Abelard's choices, when
the player interpretation is $\QAE$.

\defenv{lemma}[Hyperassignment Extensions][Lmm][HypExt]
  {
  The following hold true, for all \QPTL formulae $\varphi$, atomic propositions
  $\apElm \in \APSet$, and hyperassignments $\AsgFam \in
  \HAsgSet[\subseteq](\free{\varphi} \setminus \{ \apElm \})$.\\
  \begin{minipage}{0.525\textwidth}
  \noindent
  \begin{enumerate}
    \item\labelx{ea}
      Statements~\ref{lmm:hypext(ea:org)} and~\ref{lmm:hypext(ea:ext)} are
      equivalent:
      \begin{enumerate}
        \item\labelx{ea:org}
          there is $\XSet \in \AsgFam$ such that $\asgElm \models \exists \apElm
          \ldotp \varphi$, for all $\asgElm \in \XSet$;
        \item\labelx{ea:ext}
          there is $\XSet \!\in\! \extFun[]{\AsgFam, \apElm}$ such that $\asgElm
          \!\!\models\!\! \varphi$, for all $\asgElm \!\in\! \XSet$.
      \end{enumerate}
  \end{enumerate}
  \end{minipage}
  \hspace{-1em}
  \begin{minipage}{0.475\textwidth}
  \begin{enumerate}
    \setcounter{enumi}{1}
    \item\labelx{ae}
      Statements~\ref{lmm:hypext(ae:org)} and~\ref{lmm:hypext(ae:ext)} are
      equivalent:
      \begin{enumerate}
        \item\labelx{ae:org}
          for all $\XSet \in \AsgFam$, $\asgElm \models \forall \apElm \ldotp
          \varphi$, for some $\asgElm \in \XSet$;
        \item\labelx{ae:ext}
          for all $\XSet \in \extFun[]{\AsgFam, \apElm}$, $\asgElm \models
          \varphi$, for some $\asgElm \in \XSet$.
      \end{enumerate}
    \end{enumerate}
  \end{minipage}
  }

We can finally introduce the new semantics for \QPTL based on the novel notion
of hyperassignment.

\begin{definition}[Alternating Hodges Semantics]
  \label{def:althodsem}
  The \emph{alternating-Hodges-semantics relation} $\AsgFam \cmodels[][\alpha]\:
  \varphi$ is
  defined as follows, for all \QPTL formulae $\varphi$, hyperassignments
  $\AsgFam \in \HAsgSet[\subseteq](\free{\varphi})$, and alternation flags
  $\alpha \in \{ \QEA, \QAE \}$.
  \begin{enumerate}
    \item\label{def:althodsem(ltl)}
      whenever $\psi$ is an \LTL formula:
      \begin{enumerate}
        \item\label{def:althodsem(ltl:ea)}
          $\AsgFam \cmodels[][\QEA] \psi$ if there is a set of assignments
          $\XSet \in \AsgFam$ such that, for each assignment $\asgElm \in
          \XSet$, it holds that $\asgElm \cmodels[\LTL]\: \psi$;
        \item\label{def:althodsem(ltl:ae)}
          $\AsgFam \cmodels[][\QAE] \psi$ if, for each set of assignments $\XSet
          \in \AsgFam$, there is an assignment $\asgElm \in \XSet$ such that
          $\asgElm \cmodels[\LTL]\: \psi$;
      \end{enumerate}
    \item\label{def:althodsem(neg)}
      $\AsgFam \cmodels[][\alpha] \neg \phi$ if $\AsgFam
      \notcmodels[][\dual{\alpha}]\: \phi$, \ie, it is not the case that
      $\AsgFam \cmodels[][\dual{\alpha}]\: \phi$;
  \end{enumerate}
  \begin{minipage}{0.5\textwidth}
  \begin{enumerate}
    \setcounter{enumi}{2}
    \item\label{def:althodsem(con)}
      \begin{enumerate}
        \item\label{def:althodsem(con:ea)}
          $\AsgFam \cmodels[][\QEA] \phi_{1} \wedge \phi_{2}$ if, for each
          $(\AsgFam[1], \AsgFam[2]) \in \parFun[]{\AsgFam}$, it holds that
          $\AsgFam[1] \neq \emptyset$ and $\AsgFam[1] \cmodels[][\QEA] \phi_{1}$
          or $\AsgFam[2] \neq \emptyset$ and $\AsgFam[2] \cmodels[][\QEA]
          \phi_{2}$;
        \item\label{def:althodsem(con:ae)}
          $\AsgFam \cmodels[][\QAE] \phi_{1} \wedge \phi_{2}$ if $\dual{\AsgFam}
          \cmodels[][\QEA] \phi_{1} \wedge \phi_{2}$;
      \end{enumerate}
    \setcounter{enumi}{4}
    \item\label{def:althodsem(exs)}
      for all atomic propositions $\apElm \in \APSet$:
      \begin{enumerate}
        \item\label{def:althodsem(exs:ea)}
          $\AsgFam \cmodels[][\QEA]\: \exists \apElm \ldotp \phi$ if
          $\extFun{\AsgFam, \apElm} \cmodels[][\QEA] \phi$;
        \item\label{def:althodsem(exs:ae)}
          $\AsgFam \cmodels[][\QAE]\: \exists \apElm \ldotp \phi$ if
          $\dual{\AsgFam} \cmodels[][\QEA]\: \exists \apElm \ldotp \phi$;
      \end{enumerate}
  \end{enumerate}
  \end{minipage}
  \begin{minipage}{0.5\textwidth}
  \begin{enumerate}
    \setcounter{enumi}{3}
    \item\label{def:althodsem(dis)}
      \begin{enumerate}
        \item\label{def:althodsem(dis:ae)}
          $\AsgFam \cmodels[][\QAE] \phi_{1} \vee \phi_{2}$ if there is
          $(\AsgFam[1], \AsgFam[2]) \in \parFun[]{\AsgFam}$ such that if
          $\AsgFam[1] \!\neq\! \emptyset$ then $\AsgFam[1] \cmodels[][\QAE]
          \phi_{1}$ and if $\AsgFam[2] \!\neq\! \emptyset$ then $\AsgFam[2]
          \cmodels[][\QAE] \phi_{2}$;
        \item\label{def:althodsem(dis:ea)}
          $\AsgFam \cmodels[][\QEA] \phi_{1} \vee \phi_{2}$ if $\dual{\AsgFam}
          \cmodels[][\QAE] \phi_{1} \vee \phi_{2}$;
      \end{enumerate}
    \setcounter{enumi}{5}
    \item\label{def:althodsem(all)}
      for all atomic propositions $\apElm \in \APSet$:
      \begin{enumerate}
        \item\label{def:althodsem(all:ea)}
          $\AsgFam \cmodels[][\QEA]\: \forall \apElm \ldotp \phi$ if
          $\dual{\AsgFam} \cmodels[][\QAE]\: \forall \apElm \ldotp \phi$;
        \item\label{def:althodsem(all:ae)}
          $\AsgFam \cmodels[][\QAE]\: \forall \apElm \ldotp \phi$ if
          $\extFun{\AsgFam, \apElm} \cmodels[][\QAE] \phi$.
      \end{enumerate}
  \end{enumerate}
  \end{minipage}

\end{definition}

The base case (Item~\ref{def:althodsem(ltl)}) for \LTL formulae $\psi$ simply
formalizes the intuition about satisfaction relative to the alternation flag: if
$\alpha \!=\! \QEA$, there exists a set of assignments whose elements satisfy
$\psi$ in the Tarski sense; the dual applies when $\alpha \!=\! \QAE$.
Negation, in accordance with the classic game-theoretic interpretation, is dealt
with by simply exchanging the player interpretation of the hyperassignment
(Item~\ref{def:althodsem(neg)}).
Observe that, from this semantic condition, it immediately follows that either
$\AsgFam \cmodels[][\alpha] \varphi$ or $\AsgFam \cmodels[][\dual{\alpha}] \neg
\varphi$.
In other words, the semantics does not allow formulae with an undetermined truth
value.
The semantics of the remaining Boolean connectives
(Items~\ref{def:althodsem(con:ea)} and~\ref{def:althodsem(dis:ae)}) and
quantifiers (Items~\ref{def:althodsem(exs:ea)} and~\ref{def:althodsem(all:ae)})
is a direct application of Lemmata~\ref{lmm:blncon} and~\ref{lmm:hypext}.
Observe that swapping between $\cmodels[][\QEA]$ and $\cmodels[][\QAE]$
(Items~\ref{def:althodsem(con:ae)}, \ref{def:althodsem(dis:ea)},
\ref{def:althodsem(exs:ae)} and~\ref{def:althodsem(all:ea)}) is done according
to Lemma~\ref{lmm:flp} and it represents the fundamental point where our
approach departs from Hodges' semantics~\cite{Hod97a,Hod97b}.
The above three lemmata also imply the following theorem, which formalizes an
\emph{adequacy principle} that reduces the two satisfiability relations of the
new semantics to the classic Tarskian satisfaction in a natural way.

\defenv{theorem}[Semantics Adequacy I][Thm][SemAdqI]
  {
  For all \QPTL formulae $\varphi$ and hyperassignments $\!\AsgFam \!\in\!
  \HAsgSet[\subseteq](\free{\varphi})$:
  \begin{enumerate}
    \item\labelx{ea}
      $\AsgFam \cmodels[][\QEA]\: \varphi$ \iff there exists a set of
      assignments $\XSet \in \AsgFam$ such that $\asgElm \models \varphi$, for
      all $\asgElm \in \XSet$;
    \item\labelx{ae}
      $\AsgFam \cmodels[][\QAE]\: \varphi$ \iff, for all sets of assignments
      $\XSet \in \AsgFam$, it holds that $\asgElm \models \varphi$, for some
      $\asgElm \in \XSet$.
  \end{enumerate}
  }

\newcommand{\asgrElm}[1][]{{\color{red!50!black}\asgElm[#1]}}
\newcommand{\asgbElm}[1][]{{\color{blue!50!black}\asgElm[#1]}}
\begin{example}
  \label{exm:sem}
  Let us consider the \QPTL sentence $\varphi \defeq \forall \apElm \ldotp
  (\psi_{\apElm} \rightarrow \exists \qapElm \ldotp (\psi_{\qapElm} \wedge
  (\qapElm \leftrightarrow \X \apElm)))$, with $\psi_{\apElm} \defeq \neg \apElm
  \wedge \X(\G \apElm \vee \G \neg \apElm)$ and $\psi_{\qapElm} \defeq \G
  \qapElm \vee \G \neg \qapElm$.
  It can be viewed as describing a very simple game with two players, Abelard
  and Eloise in this order.
  Abelard can only choose a truth value for $\apElm$ that will hold constant at
  any time instant except for time $0$, where it is false regardless of his
  choice, in accordance with $\psi_{\apElm}$.
  Eloise, instead, chooses a truth value for $\qapElm$ that will hold constant
  from time $0$ onward, as dictated by $\psi_{\qapElm}$.
  The \LTL formula $\qapElm \leftrightarrow \X \apElm$ encodes the game
  objective, requiring that the truth value of $\apElm$ at time $1$ matches that
  of $\qapElm$ at time $0$.
  Sentence $\varphi$, then, asks whether Eloise can respond with one of her
  legal moves to every legal move by Abelard so that the objective is always
  met.
  Let us now apply the semantic rules given in Definition~\ref{def:althodsem},
  from which we derive the following deduction steps.
  \begin{enumerate}
    \item\label{exm:sem(sat)}
      $\{\{ \emptyfun \}\} \cmodels[][\QAE]\: \varphi$;
    \item\label{exm:sem(all:ae)}
      $\{ \{ \asgElm[\apElm] \}, \{ \asgElm[\dual{\apElm}] \}, \ldots \}
      \cmodels[][\QAE]\: \psi_{\apElm} \rightarrow \exists \qapElm \ldotp
      (\psi_{\qapElm} \wedge (\qapElm \leftrightarrow \X \apElm))$;
    \item\label{exm:sem(dis:ae)}
      $\{ \ldots \} \cmodels[][\QAE]\: \neg \psi_{\apElm}$ and
      $\{ \{ \asgElm[\apElm] \}, \{ \asgElm[\dual{\apElm}] \} \}
      \cmodels[][\QAE]\: \exists \qapElm \ldotp (\psi_{\qapElm} \wedge (\qapElm
      \leftrightarrow \X \apElm))$;
    \item\label{exm:sem(exs:ae)}
      $\{ \{ \asgElm[\apElm], \asgElm[\dual{\apElm}] \} \} \cmodels[][\QEA]\:
      \exists \qapElm \ldotp (\psi_{\qapElm} \wedge (\qapElm \leftrightarrow \X
      \apElm))$;
    \item\label{exm:sem(exs:ea)}
      $\{ \{ \asgElm[\apElm\qapElm], \asgElm[\dual{\apElm}\qapElm] \}, \{
      \asgElm[\apElm\qapElm], \asgElm[\dual{\apElm}\dual{\qapElm}] \}, \{
      \asgElm[\apElm\dual{\qapElm}], \asgElm[\dual{\apElm}\qapElm] \}, \{
      \asgElm[\apElm\dual{\qapElm}], \asgElm[\dual{\apElm}\dual{\qapElm}] \},
      \ldots \} \cmodels[][\QEA]\: \psi_{\qapElm} \wedge (\qapElm
      \leftrightarrow \X \apElm)$.
  \end{enumerate}
  Being $\varphi$ a sentence, it is satisfiable \iff Step~\ref{exm:sem(sat)}
  holds true.
  By Rule~\ref{def:althodsem(all:ae)} of Definition~\ref{def:althodsem} on
  universal quantifications, we derive Step~\ref{exm:sem(all:ae)}, where
  $\asgElm[\apElm] \defeq \{ \apElm \mapsto \bot\top^{\omega} \}$ and
  $\asgElm[\dual{\apElm}] \defeq \{ \apElm \mapsto \bot^{\omega} \}$ are the
  only two assignments satisfying the precondition $\psi_{\apElm}$.
  The first assignment is obtained by extending $\emptyfun$ by means of the
  constant functor $\FFun[\!\bot\!\top]$ which returns false at time $0$ and true at
  every future instant, \ie, $\asgElm[\apElm] = \extFun{\emptyfun,
  \FFun[\!\bot\!\top], \apElm}$.
  Similarly, the second one is obtained by the constant functor
  $\FFun[\bot]$ returning false at any time.
  The assignments obtained by the uncountably many remaining functors are
  summarized by the ellipsis.
  Now, due to the semantics of disjunction, see
  Rule~\ref{def:althodsem(dis:ae)}, we can split the hyperassignment into two
  parts: $\{ \ldots \}$ containing all the singleton sets of those assignments
  violating $\psi_{\apElm}$ and its complement $\{ \{ \asgElm[\apElm] \}, \{
  \asgElm[\dual{\apElm}] \} \}$.
  On the first hyperassignment we need to check $\neg \psi_{\apElm}$, while on
  the second one the remaining part of the formula, as stated in
  Step~\ref{exm:sem(dis:ae)}.
  Since $\{ \ldots \} \cmodels[][\QAE]\: \neg \psi_{\apElm}$ holds by
  construction, Rule~\ref{def:althodsem(exs:ae)} applied to the second part
  leads to Step~\ref{exm:sem(exs:ae)}, where we use the equality $\{ \{
  \asgElm[\apElm], \asgElm[\dual{\apElm}] \} \} = \dual{\{ \{ \asgElm[\apElm]
  \}, \{ \asgElm[\dual{\apElm}] \} \}}$.
  Rule~\ref{def:althodsem(exs:ea)} on existential quantifications allows, then,
  to derive Step~\ref{exm:sem(exs:ea)}, where $\asgElm[\flat\qapElm] \defeq
  {\asgElm[\flat]}[\qapElm \mapsto \top^{\omega}]$ and
  $\asgElm[\flat\dual{\qapElm}] \defeq {\asgElm[\flat]}[\qapElm \mapsto
  \bot^{\omega}]$, with $\flat \in \{ \apElm, \dual{\apElm}\}$.
  The relevant sets of assignments in the hyperassignment at
  Step~\ref{exm:sem(exs:ea)} are obtained as follows:
  \begin{enumerate}[(a)]
    \item
      $\{ \asgElm[\apElm\qapElm], \asgElm[\dual{\apElm}\qapElm] \} =
      \extFun{\{ \asgElm[\apElm], \asgElm[\dual{\apElm}] \}, \FFun[\top],
      \qapElm}$, where $\FFun[\top]$ is the constant functor returning true at
      every time;
    \item
      $\{ \asgElm[\apElm\qapElm], \asgElm[\dual{\apElm}\dual{\qapElm}] \} =
      \extFun{\{ \asgElm[\apElm], \asgElm[\dual{\apElm}] \}, \FFun[\apElm],
      \qapElm}$, where $\FFun[\apElm](\asgElm)$ returns at time $i$ the value of
      $\apElm$ in $\asgElm$ at $i + 1$;
    \item
      $\{ \asgElm[\apElm\dual{\qapElm}], \asgElm[\dual{\apElm}\qapElm] \} =
      \extFun{\{ \asgElm[\apElm], \asgElm[\dual{\apElm}] \},
      \FFun[\dual{\apElm}], \qapElm}$, where $\FFun[\dual{\apElm}](\asgElm)$
      returns at time $i$ the dual value of $\apElm$ in $\asgElm$ at $i + 1$;
    \item
      $\{ \asgElm[\apElm\dual{\qapElm}], \asgElm[\dual{\apElm}\dual{\qapElm}] \}
      = \extFun{\{ \asgElm[\apElm], \asgElm[\dual{\apElm}] \}, \FFun[\bot],
      \qapElm}$, where $\FFun[\bot]$ is, as above, the constant functor returning false at
      every time.
  \end{enumerate}
  %
%
  At this point, since $\psi_{\qapElm} \wedge (\qapElm \leftrightarrow \X
  \apElm)$ is an \LTL formula, Rule~\ref{def:althodsem(ltl)} of
  Definition~\ref{def:althodsem} can be applied, thus asking for a set of
  assignments containing only assignments that make $\psi_{\qapElm} \wedge
  (\qapElm \leftrightarrow \X \apElm)$ true.
  Both assignments in the doubleton $\{ \asgElm[\apElm\qapElm],
  \asgElm[\dual{\apElm}\dual{\qapElm}] \}$ satisfy the \LTL formula
  $\psi_{\qapElm} \wedge (\qapElm \leftrightarrow \X \apElm)$, which implies
  that $\varphi$ is satisfiable, witnessing Eloise's win.
%
\end{example}





\section{Good-for-Game \QPTL}
\label{sec:behdep}

The semantic framework introduced in the previous section allows us to encode
behavioral independence constraints among the quantified variables of \QPTL.
We thus obtain the logic \GFG-\QPTL, an extension of \QPTL able to express the
behavioralness of quantifications over temporal valuations.

\subsection{Adding Behavioral Dependencies to \QPTL}

Given a set of assignments $\AsgSet(\PSet)$ over some $\PSet \subseteq \APSet$,
a \emph{behavioral quantification} \wrt a proposition $\apElm \in \PSet$ should
choose, for each assignment $\asgElm \in \AsgSet(\PSet)$, a temporal valuation
$\fFun \colon \SetN \to \SetB$ in such a way that, intuitively, at each instant
of time $k \in \SetN$, the value $\fFun(k)$ of $\fFun$ at $k$ only depends on
the values $\asgElm(\apElm)(t)$ of the temporal valuation $\asgElm(\apElm)$ at
the instants of time $t \leq k$; this means that $\fFun(k)$ is independent of
the values $\asgElm(\apElm)(t)$ at any future instant $t > k$.
To be more precise, consider two assignments $\asgElm[1], \asgElm[2] \in
\AsgSet(\PSet)$ that may differ only on $\apElm$ strictly after $k$.
Then, the functor $\FFun \in \FncSet(\PSet)$ interpreting a quantification
behavioral \wrt $\apElm$ must return the same value at $k$ as a reply to both
$\asgElm[1]$ and $\asgElm[2]$, \ie, $\FFun(\asgElm[1])(k) =
\FFun(\asgElm[2])(k)$; in other words, $\FFun(\asgElm)(k)$ cannot exploit the
knowledge of the values $\asgElm(\apElm)(t)$, with $t > k$.
An analogous concept has been introduced in \SL~\cite{MMPV14}.
A stronger notion of behavioralness, similar to one reported in~\cite{GBM18},
requires the functor $\FFun$ to satisfy the above equality when $\asgElm[1]$ and
$\asgElm[2]$ only (possibly) differ on $\apElm$ for $t \!\geq\! k$ and leads to the
concept of \emph{strongly behavioral quantification}.
In game-theoretic terms, the interpretation of a \emph{behavioral quantifier}
\wrt $\apElm$ requires the corresponding player to
choose the value of a proposition at each round only based on the choices for
$\apElm$ made by the adversary up to that round.
For a \emph{strongly behavioral quantifier}, instead, the adversary keeps its
choice for $\apElm$ at the current round hidden and the player can only access
the choices made for $\apElm$ at previous rounds.
Definitions~\ref{def:asgind} and~\ref{def:behfun} formalize these fundamental
concepts.

\begin{definition}[Assignment Distinguishability]
  \label{def:asgind}
  Let $\asgElm[1], \asgElm[2] \in \AsgSet(\PSet)$ be two assignments over some
  set $\PSet \subseteq \APSet$ of propositions, $\apElm \in \PSet$ one of these
  propositions, and $k \in \SetN$ a number.
  Then, $\asgElm[1]$ and $\asgElm[2]$ are \emph{$(\apElm, k)$-strict
  distinguishable} (\resp, \emph{$(\apElm, k)$-distinguishable}), in symbols
  $\asgElm[1] \!\approx_{\apElm}^{> k}\! \asgElm[2]$ (\resp, $\asgElm[1]
  \!\approx_{\apElm}^{\geq k}\! \asgElm[2]$), if the following holds: \\
  \begin{minipage}{0.5\textwidth}
  \begin{enumerate}
    \item\label{def:asgind(var)}
      $\asgElm[1](\qapElm) = \asgElm[2](\qapElm)$, for all $\qapElm \in \PSet$
      with $\qapElm \neq \apElm$;
  \end{enumerate}
  \end{minipage}
  \begin{minipage}{0.5\textwidth}
  \begin{enumerate}
    \setcounter{enumi}{1}
    \item\label{def:asgind(piv)}
      $\asgElm[1](\apElm)(t) = \asgElm[2](\apElm)(t)$, for all $t \leq k$
      (\resp, $t < k$).
  \end{enumerate}
  \end{minipage}
\end{definition}

The notion of $(\apElm, k)$-strict distinguishability (\resp, $(\apElm,
k)$-distinguishability) allows us to identify all the assignments that can only
differ on the proposition $\apElm$ at some instant $t \!>\! k$ (\resp, $t
\!\geq\! k$).
Indeed, $\approx_{\apElm}^{> k}$ (\resp, $\approx_{\apElm}^{\geq k}$) is an
equivalence relation on $\AsgSet(\PSet)$, whose equivalence classes identify
those assignments precisely.
A \emph{behavioral} (\resp, \emph{strongly-behavioral}) functor must reply at
time $k$ uniformly to all $\approx_{\apElm}^{> k}$-equivalent (\resp,
$\approx_{\apElm}^{\geq k}$-equivalent) assignments.


\begin{wrapfigure}{r}{0.38\linewidth}
  \vspace{-2em}
  \begin{center}
   \hspace{-3em} \scalebox{1.00}[0.80]{\figbhvfun}
  \end{center}
  \vspace{-0.5em}
  \caption{\label{fig:bhvfun} Two $\approx_{\apElm}^{> 3}$ (\resp,
  $\approx_{\apElm}^{\geq 4}$)-equivalent assignments and some functors.}
  \vspace{-0.5em}
\end{wrapfigure}

\begin{definition}[Behavioral Functor]
  \label{def:behfun}
  Let $\FFun \!\in\! \FncSet(\PSet)$ be a functor over some set $\PSet \subseteq
  \APSet$ of propositions and $\apElm \in \PSet$ one of these propositions.
  Then, $\FFun$ is \emph{behavioral} (\resp, \emph{strongly behavioral}) \wrt
  $\apElm$ if $\FFun(\asgElm[1])(k) = \FFun(\asgElm[2])(k)$, for all numbers $k
  \in \SetN$ and pairs of $\approx_{\apElm}^{> k}$-equivalent (\resp,
  $\approx_{\apElm}^{\geq k}$-equivalent) assignments $\asgElm[1], \asgElm[2]
  \in \AsgSet(\PSet)$.
\end{definition}

\begin{example}
  \label{exm:bhvfun}
  Let $\asgElm[1]$ and $\asgElm[2]$ be two assignments over the singleton $\{
  \apElm \}$ defined as in Figure~\ref{fig:bhvfun}.
  It is clear that $\asgElm[1] \!\approx_{\apElm}^{> 3}\! \asgElm[2]$, but
  $\asgElm[1] \!\not\approx_{\apElm}^{> 4}\! \asgElm[2]$, and so $\asgElm[1]
  \!\approx_{\apElm}^{\geq 4}\! \asgElm[2]$, but $\asgElm[1]
  \!\not\approx_{\apElm}^{\geq 5}\! \asgElm[2]$.
  Also, consider the three functors $\FFun[\ASym], \FFun[\BSym], \FFun[\SSym]
  \in \FncSet(\{ \apElm \})$ defined as follows, for all $\AsgFam \in
  \AsgSet(\{ \apElm \})$ and $t \!\in\! \SetN$:
  $\FFun[\ASym](\asgElm)(t) \defeq \asgElm(\apElm)(t + 1)$;
  $\FFun[\BSym](\asgElm)(t) \defeq \dual{\asgElm(\apElm)(t)}$;
  $\FFun[\SSym](\asgElm)(t) \defeq \Tt$, if $t = 0$, and
  $\FFun[\SSym](\asgElm)(t) \defeq \asgElm(\apElm)(t - 1)$, otherwise.
  It is immediate to see that $\FFun[\BSym]$ is behavioral, while $\FFun[\SSym]$
  is strongly behavioral.
  However, $\FFun[\ASym]$ does not enjoy any behavioral property, being defined
  as a future-dependent functor.
  Indeed, $\FFun[\ASym](\asgElm[1])(3) \neq \FFun[\ASym](\asgElm[2])(3)$, even
  though $\asgElm[1] \!\approx_{\apElm}^{> 3}\! \asgElm[2]$.
\end{example}

To capture in the logic the behavioral constraints on the functors, we extend
\QPTL with additional decorations for the quantifiers that express behavioral
dependencies among the propositions involved.
The result is a new logic, called \emph{Good-for-Games \QPTL}, able to express
in a natural way game-theoretic concepts of Boolean games.

\begin{definition}[\GFG-\QPTL Syntax]
  \label{def:extsyn}
  \emph{Good-for-Games \QPTL} (\emph{\GFG-\QPTL}) is the set of formulae built
  accordingly to the
  following context-free grammar, where $\psi \in \LTL$, $\apElm \in \APSet$,
  and $\PSet[\BSym], \PSet[\SSym] \subseteq \APSet$:
  \begin{align*}
    \varphi
       \seteq \psi \mid \neg \varphi \mid \varphi \wedge \varphi \mid \varphi
        \vee \varphi \mid \exists \apElm \colon\! \spcElm \ldotp \varphi \mid
        \forall \apElm \colon\! \spcElm \ldotp \varphi; &
    \hspace{1em}\spcElm
       \seteq \dep{ \PSet[\BSym] }{ \PSet[\SSym] }.
  \end{align*}
\end{definition}

A propositional quantifier of the form $\Qnt \apElm \colon\! \dep{ \PSet[\BSym]
}{ \PSet[\SSym] }$ explicitly expresses a $\Qnt$-quantification over $\apElm$,
\ie, a choice of a functor to interpret $\apElm$, that is behavioral \wrt all
the propositions in $\PSet[\BSym]$ and strongly-behavioral \wrt those in
$\PSet[\SSym]$.

To ease the notation, we may write $\Qnt[][\spcElm] \apElm \ldotp \varphi$
instead of $\Qnt \apElm \colon\! \spcElm \ldotp \varphi$, write
$\bdep{\PSet[\BSym]}$ and $\sdep{\PSet[\SSym]}$ for $\dep{ \PSet[\BSym] }{
\emptyset }$ and $\dep{ \emptyset }{ \PSet[\SSym] }$, respectively, and $\BSym$
and $\SSym$ instead of $\bdep{\APSet}$ and $\sdep{\APSet}$.
We also omit the quantifier specification $\dep{ \emptyset } { \emptyset }$,
using $\Qnt \apElm \ldotp \varphi$ to denote $\Qnt \apElm \colon\! \dep{
\emptyset }{ \emptyset } \ldotp \varphi$.
Finally, we may drop the curly bracket for the sets $\PSet[\BSym]$ and
$\PSet[\SSym]$ and write $\bdep{ \apElm, \qapElm }$ instead of $\bdep{\{ \apElm,
\qapElm \}}$.

We say that a \GFG-\QPTL formula is \emph{behavioral} (\resp,
\emph{strongly-behavioral}) if it is in prenex form and all its quantifier
specifications are equal to $\BSym$ (\resp, $\SSym$).
We denote by $\QntSet$ (\resp, $\QntSet[\BSym]$) the set of (\resp,
behavioral) quantifier prefixes and by $\SpcSet$ the set of quantifier
specifications.

Given assignments $\asgElm[1], \asgElm[2] \!\in\! \AsgSet(\PSet)$, we write
$\asgElm[1] \!\sim_{\spcElm}^{k}\! \asgElm[2]$, for some $\spcElm = \dep{
\PSet[\BSym] }{ \PSet[\SSym] } \in \SpcSet$ and $k \in \SetN$, if one of the
following conditions holds:
\begin{inparaenum}
  \item
    $\asgElm[1] = \asgElm[2]$;
  \item
    $\asgElm[1] \approx_{\apElm}^{> k} \asgElm[2]$, for some $\apElm \in
    \PSet[\BSym]$;
  \item
    $\asgElm[1] \approx_{\apElm}^{\geq k} \asgElm[2]$, for some $\apElm \in
    \PSet[\SSym]$.
\end{inparaenum}
We denote by $\approx_{\spcElm}^{k}$ the transitive closure of the reflexive
and symmetric relation $\sim_{\spcElm}^{k}$.

\defenv{proposition}[][Prp][AsgInd]
  {
  Let $\PSet \subseteq \APSet$, $\asgElm[1], \asgElm[2] \in \AsgSet(\PSet)$,
  $\spcElm \in \SpcSet$, and $k \in \mathbb{N}$.
  Then, $\asgElm[1] \approx_{\spcElm}^{k} \asgElm[2]$ \iff the following hold
  true:
  \begin{enumerate}
    \item\labelx{var}
      $\asgElm[1](\qapElm) = \asgElm[2](\qapElm)$, for all $\qapElm \in \PSet
      \setminus (\PSet[\BSym] \cup \PSet[\SSym])$;
    \item\labelx{pivi}
      $\asgElm[1](\apElm)(t) = \asgElm[2](\apElm)(t)$, for all $t \leq k$ and
      $\apElm \in (\PSet[\BSym] \cap \PSet) \setminus \PSet[\SSym]$;
    \item\labelx{pivii}
      $\asgElm[1](\apElm)(t) = \asgElm[2](\apElm)(t)$, for all $t < k$ and
      $\apElm \in \PSet[\SSym] \cap \PSet$.
  \end{enumerate}
  }


\begin{example}
  \label{exm:eqvasg}
  Consider the three assignments $\asgElm[1]$, $\asgElm[2]$, and $\asgElm[3]$
  over the doubleton $\{ \apElm, \qapElm \}$ depicted in
  Figure~\ref{fig:eqvasg}.
  It is easy to see that $\asgElm[1] \!\approx_{\apElm}^{> 3}\! \asgElm[2]$ and
  $\asgElm[2] \!\approx_{\qapElm}^{\geq 3}\! \asgElm[3]$.
  Therefore, $\asgElm[1] \!\sim_{\spcElm}^{3}\! \asgElm[2]
  \!\sim_{\spcElm}^{3}\! \asgElm[3]$, where $\spcElm \defeq \dep{ \apElm }{
  \qapElm }$, which implies $\asgElm[1] \!\approx_{\spcElm}^{3}\! \asgElm[3]$.
\end{example}

\begin{wrapfigure}[7]{r}{0.43\linewidth}
  \vspace{-2em}
  \begin{center}
    \hspace{-3em}\scalebox{1.00}[0.80]{\figeqvasg}
  \end{center}
  \vspace{-0.5em}
  \caption{\label{fig:eqvasg} Three $\approx_{\spcElm}^{3}$-equivalent
  assignments, with $\spcElm \defeq \dep{ \apElm }{ \qapElm }$.}
  \vspace{-0.5em}
\end{wrapfigure}

Given a set of propositions $\PSet \subseteq \APSet$ and a quantifier
specification $\spcElm \defeq \dep{ \PSet[\BSym] }{ \PSet[\SSym] } \in \SpcSet$,
we introduce the set of \emph{$\spcElm$-functors} $\FncSet[\spcElm](\PSet)
\subseteq \FncSet(\PSet)$ containing exactly those $\FFun \in \FncSet(\PSet)$
that are behavioral \wrt all the propositions in $\PSet[\BSym] \cap \PSet$ and
strongly behavioral \wrt those in $\PSet[\SSym] \cap \PSet$.

\begin{example}
  \label{exm:eqvfun}
  Any $\dep{ \apElm }{ \qapElm }$-functor $\FFun$ replies to all assignments of
  Figure~\ref{fig:eqvasg} uniformly, for all time instants between $0$ and $3$
  included.
  Indeed, $\FFun(\asgElm[1])(3) = \FFun(\asgElm[2])(3)$, since $\asgElm[1]
  \!\approx_{\apElm}^{> 3}\! \asgElm[2]$, being $\FFun$ behavioral \wrt
  $\apElm$.
  Similarly, $\FFun(\asgElm[2])(3) = \FFun(\asgElm[3])(3)$, since $\asgElm[2]
  \!\approx_{\qapElm}^{\geq 3}\! \asgElm[3]$, being $\FFun$ strongly-behavioral
  \wrt $\qapElm$.
  Hence, $\FFun(\asgElm[1])(3) = \FFun(\asgElm[3])(3)$.
\end{example}

The following proposition ensures that the above observation highlights a
general phenomenon.

\defenv{proposition}[][Prp][BhvFnc]
  {
  If $\asgElm[1] \!\approx_{\spcElm}^{k}\! \asgElm[2]$ then
  $\FFun(\asgElm[1])(k) \!=\! \FFun(\asgElm[2])(k)$, for all $\asgElm[1],
  \asgElm[2] \!\in\! \AsgSet(\PSet)$, $\spcElm \in \SpcSet$, $k \!\in\! \SetN$,
  and $\FFun \!\in\! \FncSet[\spcElm](\PSet)$.
  }

We can now extend to \GFG-\QPTL the alternating Hodges semantics of \QPTL
reported in Definition~\ref{def:althodsem}.
We simply need to parameterize the extension operation for hyperassignments with
the corresponding specification of the behavioral dependencies:
\[
  \extFun[\spcElm]{\AsgFam, \apElm} \defeq \set{ \extFun{\XSet, \FFun, \apElm}
  }{ \XSet \in \AsgFam, \FFun \in \FncSet[\spcElm](\ap{\AsgFam}) }.
\]

\begin{definition}[Alternating Hodges Semantics Revisited]
  \label{def:extalthodsem}
  The alternating-Hodges-semantics relation $\AsgFam \cmodels[][\alpha] \varphi$
  is inductively defined as in Definition~\ref{def:althodsem}, for all but
  Items~\ref{def:althodsem(exs:ea)} and~\ref{def:althodsem(all:ae)} that are
  modified, respectively, as follows, for all propositions $\apElm \in \APSet$
  and quantifier specifications $\spcElm \in \SpcSet$: \\
  \begin{minipage}{0.5\textwidth}
  \begin{enumerate}
    \item[\ref*{def:althodsem(exs:ea)}')]\label{def:extalthodsem(exs:ea)}
      $\AsgFam \cmodels[][\QEA]\: \exists \apElm \colon\! \spcElm \ldotp \phi$
      if $\extFun[\spcElm]{\AsgFam, \apElm} \cmodels[][\QEA] \phi$;
  \end{enumerate}
  \end{minipage}
  \begin{minipage}{0.5\textwidth}
  \begin{enumerate}
    \item[\ref*{def:althodsem(all:ae)}')]\label{def:extalthodsem(all:ae)}
      $\AsgFam \cmodels[][\QAE]\: \forall \apElm \colon\! \spcElm \ldotp \phi$
      if $\extFun[\spcElm]{\AsgFam, \apElm} \cmodels[][\QAE] \phi$.
  \end{enumerate}
  \end{minipage}
\end{definition}

Notice that one could easily extend both the syntax and semantics of the
quantifier specification $\dep{ \PSet[\BSym] }{ \PSet[\SSym] }$ of \GFG-\QPTL in
order to accommodate other types of (in)dependence constraints, like the ones
already studied in first-order logic of incomplete
information~\cite{HS89,Hod97a,Vaa07,MSS11,GV13a}.
It would suffice to introduce suitable class of functors and corresponding
construct, such as the \emph{dependence atoms} of dependence logic, whose
semantics can be easily defined via hyperassignments.

For every \GFG-\QPTL formula $\varphi$ and alternation flag $\alpha \in \{ \QEA,
\QAE \}$, we say that $\varphi$ is \emph{$\alpha$-satisfiable} if there exists a
hyperassignment $\AsgFam \!\in\! \HAsgSet(\free{\varphi})$ such that $\AsgFam
\cmodels[][\alpha] \varphi$.
Also, $\varphi$ \emph{$\alpha$-implies} (\resp, is \emph{$\alpha$-equivalent}
to) a \GFG-\QPTL formula $\phi$, in symbols $\!\varphi \cimplies[][\alpha]
\phi\!$ (\resp, $\!\varphi \cequiv[][\alpha] \phi$), whenever $\free{\varphi}
\!=\! \free{\phi}$ and if $\AsgFam \cmodels[][\alpha] \varphi$ then $\AsgFam
\cmodels[][\alpha] \phi$ (\resp, $\AsgFam \cmodels[][\alpha] \varphi$ \iff
$\AsgFam \cmodels[][\alpha] \phi$), for all $\AsgFam \in
\HAsgSet[\subseteq](\free{\varphi})$.
Finally, we say that $\varphi$ is \emph{satisfiable} if it is both $\QEA$- and
$\QAE$-satisfiable, and write $\varphi \implies \phi$ (\resp, $\varphi \equiv
\phi$) if both $\varphi \cimplies[][\QEA] \phi$ and $\varphi \cimplies[][\QAE]
\phi$ (\resp, $\varphi \cequiv[][\QEA] \phi$ and $\varphi \cequiv[][\QAE] \phi$)
hold.

At this point, let us consider some examples to provide some insight on the
expressive power of the new logic.

\begin{example}
  \label{exm:unrprp}
  Let us consider again the \QPTL formula $\varphi$ of Example~\ref{exm:sem}.
  Obviously, $\varphi$ is not realizable, as the functor $\FFun[\apElm]$
  required to obtain the two satisfying assignments $\asgElm[\apElm\qapElm]$
  and $\asgElm[\dual{\apElm}\dual{\qapElm}]$ is non-behavioral, thus, not
  implementable by any real transducer.
  By replacing the two quantifiers with their behavioral counterparts, the
  resulting \GFG-\QPTL formula $\varphi_{\BSym} \defeq \forall^{\BSym} \apElm
  \ldotp (\psi_{\apElm} \rightarrow \exists^{\BSym} \qapElm \ldotp
  (\psi_{\qapElm} \wedge (\qapElm \leftrightarrow \X \apElm))))$, is not
  satisfiable anymore.
  Indeed, the only behavioral functors that allow to satisfy $\psi_{\qapElm}$
  are $\FFun[\top]$ and $\FFun[\bot]$ and, therefore, deduction steps analogous
  to the ones of Example~\ref{exm:sem} applied to $\varphi_{\BSym}$ would lead
  to $\{ \{ \asgElm[\apElm\qapElm], \asgElm[\dual{\apElm}\qapElm] \}, \{
  \asgElm[\apElm\dual{\qapElm}], \asgElm[\dual{\apElm}\dual{\qapElm}] \} \}
  \cmodels[][\QEA]\: \psi_{\qapElm} \wedge (\qapElm \leftrightarrow \X \apElm)$.
  However, each of the two sets of assignments contains a non satisfying
  assignment for the \LTL formula $\qapElm \leftrightarrow \X \apElm$.
\end{example}

The previous example shows a satisfiable \QPTL sentence whose behavioral
counterpart becomes unsatisfiable. The opposite may also occur, as the following
example illustrates.

\begin{example}
  \label{exm:bhcsat}
  Consider the \QPTL sentence $\exists \qapElm \ldotp \forall \apElm \ldotp
  \psi$, with $\psi \defeq \apElm \leftrightarrow \X \qapElm$.
  The sentence is unsatisfiable: Abelard can falsify $\psi$ by looking at the
  value of $\qapElm$ one instant in the future and choosing the opposite value
  as the present value for $\apElm$.
  However, the two \GFG-\QPTL sentences $\forall^{\BSym} \apElm \ldotp
  \exists^{\SSym} \qapElm \ldotp \psi$ and $\exists^{\BSym} \qapElm \ldotp
  \forall^{\BSym} \apElm \ldotp \psi$ are both satisfiable.
  For the first one, it is enough to observe that the strongly-behavioral
  functor $\FFun[\SSym]$ of Example~\ref{exm:bhvfun} allows to mimic any
  temporal valuation assigned to the proposition $\apElm$ one-instant in the
  past, as required by the \LTL property $\psi$.
  For the second one, we need to show that, $\extFun[\BSym]{\{ \YSet \}, \apElm}
  \cmodels[][\QAE] \psi$, with $\YSet = \AsgSet(\{ \qapElm \})$.
  Now, let $\XSet \in \extFun[\BSym]{\{ \YSet \}, \apElm}$ be an arbitrary set
  of assignments obtained by extending those in $\YSet$ as prescribed by the
  specification $\BSym$.
  Also, consider $\asgElm[1], \asgElm[2] \in \XSet$ as two of those assignments
  that differs on $\qapElm$ at time $1$, but are equal at time $0$, \ie,
  $\asgElm[1](\qapElm)(0) = \asgElm[2](\qapElm)(0)$, but $\asgElm[1](\qapElm)(1)
  \neq \asgElm[2](\qapElm)(1)$.
  Due to the required behavioralness \wrt $\qapElm$ of the functors used in the
  extension of $\YSet$, we necessarily have that $\asgElm[1](\apElm)(0) =
  \asgElm[2](\apElm)(0)$.
  As a consequence, either one between $\asgElm[1]$ and $\asgElm[2]$ satisfies
  $\psi$, as required by Item~\ref{def:althodsem(ltl:ae)} of the semantics.
  In other words, Abelard is no longer allowed to look at the value of $\qapElm$
  in the future.
  Note that $\exists^{\BSym} \qapElm \ldotp \forall^{\BSym} \apElm \ldotp \psi$
  could not be expressed with an asymmetric Hodges-like semantics, as it cannot
  restrict the universal quantifiers.
\end{example}

\begin{example}
  \label{exm:infleak}
  Information leaks via quantification of unused variables is a well-known
  phenomenon in \IF~\cite{Vaa07}.
  The same occurs in \GFG-\QPTL, as the (in)equivalences below show:
  \[
    \forall \apElm \ldotp \exists \uapElm \ldotp \exists^{\BSym} \qapElm \ldotp
    \phi \equiv \forall \apElm \ldotp \exists \qapElm \ldotp \phi \not\equiv
    \forall \apElm \ldotp \exists^{\BSym} \qapElm \ldotp \phi \equiv \forall
    \apElm \ldotp \exists^{\BSym} \uapElm \ldotp \exists^{\BSym} \qapElm
    \ldotp \phi;
  \]
  where $\apElm, \qapElm \!\in\! \free{\phi}$, but $\uapElm \!\not\in\!
  \free{\phi}$.
  Indeed, an arbitrary functor $\GFun[\qapElm]$ for $\qapElm$ in $\forall \apElm
  \ldotp \exists \qapElm \ldotp \phi$ can be simulated in $\forall \apElm \ldotp
  \exists \uapElm \ldotp \exists^{\BSym} \qapElm \ldotp \phi$ by the functors
  $\FFun[\uapElm] = \GFun[\qapElm]$, for $\uapElm$, and $\FFun[\qapElm](\asgElm)
  = \asgElm(\uapElm)$, for $\qapElm$.
  Clearly, $\FFun[\qapElm]$, being the identity on $\uapElm$, is behavioral.
  Intuitively, the unused non-behaviorally-quantified proposition
  $\uapElm$ leaks information about the future of $\apElm$ to $\qapElm$ even if
  the latter is behaviorally quantified, as it can see the future of
  $\apElm$ through the present of $\uapElm$.
  This does not happen, however, if $\uapElm$ is forced to be behavioral (\resp,
  strongly-behavioral).
  Indeed, the behavioral (\resp, strongly-behavioral) fragments of \GFG-\QPTL
  enjoys the classic property of elimination of unused propositions.
\end{example}

The following example expands on the connection between \GFG-\QPTL and
\GFG-Automata briefly mentioned in the introduction and shows that \GFG-\QPTL
can express the property of being
good-for-game for an automaton.

\begin{example}
  \label{exm:gfgaut}
  It is well known that \QPTL is able to express any $\omega$-regular
  language~\cite{Sis83}.
  This can be proved by encoding the existence of an accepting run of an
  arbitrary nondeterministic B\"uchi word automaton $\NAutName$ via a formula
  $\varphi \defeq \exists \sapElm[1] \ldots \exists \sapElm[k] \ldotp \psi$,
  where $\free{\varphi} = \{ \apElm[1], \ldots, \apElm[n] \}$ is the set of
  propositions of the recognized language, $\sapElm[1], \ldots, \sapElm[k]$ are
  mutually exclusive fresh propositions representing the $k$ states of
  $\NAutName$, and $\psi$ is the \LTL formula encoding the transition function
  and the acceptance condition.
  Via the behavioral \GFG-\QPTL formula $\varphi_{\BSym} \defeq \exists^{\BSym}
  \sapElm[1] \ldots \exists^{\BSym} \sapElm[k] \ldotp \psi$ we can identify
  precisely the sublanguages recognized by $\NAutName$ when the nondeterminism
  is resolved in a \emph{good-for-game} manner~\cite{HP06}, \ie, when the choice
  of a successor state is based on the prefix of the input words read up to that
  moment.
  Thus, the \GFG-\QPTL sentence $\forall \apElm[1] \ldots \forall \apElm[n]
  \ldotp \allowbreak (\varphi \leftrightarrow \varphi_{\BSym})$ is satisfiable
  \iff $\NAutName$ is a good-for-game automaton.
\end{example}


\subsection{Model-Theoretic Analysis}

Let us proceed with an elementary \emph{model-theoretic analysis} of \GFG-\QPTL,
showing that it enjoys several basic properties, like \emph{De Morgan laws}, one
would expect from a classic logic.

We start by observing the \emph{monotonicity} of both the dualization and
extension operators \wrt the preorder $\sqsubseteq$, a simple property that is a
key tool in all subsequent statements.

\defenv{proposition}[][Prp][OprMon]
  {
  Let $\apElm \in \APSet$, $\spcElm \in \SpcSet$, and $\AsgFam[1], \AsgFam[2]
  \in \HAsgSet$ with $\AsgFam[1] \sqsubseteq \AsgFam[2]$:
  \begin{inparaenum}
    \item
      $\dual{\AsgFam[2]} \sqsubseteq \dual{\AsgFam[1]}$;
    \item
      $\extFun[\spcElm]{\AsgFam[1], \apElm} \sqsubseteq
      \extFun[\spcElm]{\AsgFam[2], \apElm}$.
  \end{inparaenum}
  }

%

The preorder $\sqsubseteq$ between hyperassignments captures the intuitive
notion of satisfaction strength \wrt \GFG-\QPTL formulae.
Indeed, thanks to Item~\ref{def:althodsem(ltl)} of
Definition~\ref{def:althodsem}, it holds that, if $\AsgFam[1] \sqsubseteq
\AsgFam[2]$, the hyperassignment $\AsgFam[1]$ satisfies \wrt the $\QEA$ (\resp,
$\QAE$) semantics less (\resp, more) \LTL formulae than the hyperassignment
$\AsgFam[2]$, \ie, if $\AsgFam[1]$ (\resp, $\AsgFam[2]$) satisfies $\psi$, then
$\AsgFam[2]$ (\resp, $\AsgFam[1]$) does as well.
This property can easily be lifted to arbitrary \GFG-\QPTL formulae, by a
standard structural induction using the monotonicity of the dualization and
extension operators.

\defenv{theorem}[Hyperassignment Refinement][Thm][HypRef]
  {
  Let $\varphi$ be a \GFG-\QPTL formula and $\AsgFam[1], \AsgFam[2] \in
  \HAsgSet[\subseteq](\free{\varphi})$ with $\AsgFam[1] \sqsubseteq \AsgFam[2]$.
  Then, $\AsgFam[1] \!\cmodels[][\QEA]\! \varphi$ implies $\AsgFam[2]
  \!\cmodels[][\QEA]\! \varphi$ and $\AsgFam[2] \!\cmodels[][\QAE]\! \varphi$
  implies $\AsgFam[1] \!\cmodels[][\QAE]\! \varphi$.
  }

As an immediate consequence, we obtain the following result.\!

\defenv{corollary}[Hyperassignment Equivalence][Cor][Eqv]
  {
  Let $\varphi$ be a \GFG-\QPTL formula and $\AsgFam[1], \AsgFam[2] \in
  \HAsgSet[\subseteq](\free{\varphi})$ with $\AsgFam[1] \equiv \AsgFam[2]$.
  Then, $\AsgFam[1] \cmodels[][\alpha] \varphi$ \iff $\AsgFam[2]
  \cmodels[][\alpha] \varphi$.
  }

A fundamental feature of the proposed alternating semantics is the
\emph{duality} between swapping the players of a hyperassignment $\AsgFam$, \ie,
swapping the alternation flag, and swapping the choices of the players, \ie,
dualizing $\AsgFam$.
Indeed, the following results states that dualizing both the alternation flag
$\alpha$ and the hyperassignment preserves the truth of any formula.
This also implies, as one might expect, that double dualization has no effect
either.
The latter fact is also a consequence of the previous corollary, since $\AsgFam
\equiv \dual{\dual{\AsgFam}}$, due to Proposition~\ref{prp:dltinv}.

\defenv{theorem}[Double Dualization][Thm][Dlt]
  {
  Let $\varphi$ be a \GFG-\QPTL formula and $\AsgFam \in
  \HAsgSet[\subseteq](\free{\varphi})$.
  Then, $\dual{\AsgFam} \cmodels[][\dual{\alpha}] \varphi$ \iff
  $\dual{\dual{\AsgFam}} \cmodels[][\alpha] \varphi$ \iff $\AsgFam
  \cmodels[][\alpha] \varphi$.
  }

%

The duality property also grants that formulae satisfaction and equivalence do
not depend on the specific interpretation $\alpha$ of hyperassignments: a
positive answer for $\alpha$ implies the same for $\dual{\alpha}$.
This \emph{invariance} corresponds to the intuition that Eloise and Abelard both
agree on the true and false formulae.
Similarly, if $\varphi$ is considered to be equivalent to, or to imply, some other
property $\phi$ by Eloise, the same equivalence, or implication, holds for
Abelard as well, and \viceversa.

\defenv{corollary}[Interpretation Invariance][Cor][IntInv]
  {
  Let $\varphi$ and $\phi$ be \GFG-\QPTL formulae.
  $\varphi$ is $\QEA$-satisfiable \iff $\varphi$ is $\QAE$-satisfiable.
  Also, $\varphi \cimplies[][\QEA] \phi$ \iff $\varphi \cimplies[][\QAE] \phi$
  and $\varphi \cequiv[][\QEA] \phi$ \iff $\varphi \cequiv[][\QAE] \phi$.
  }

Thanks to this invariance, the following Boolean laws hold.

\defenv{lemma}[Boolean Laws][Lmm][DML]
  {
  Let $\varphi$, $\varphi_{1}$, $\varphi_{2}$ be \GFG-\QPTL formulae:
  \begin{inparaenum}
    \item\labelx{neg}
      $\varphi \equiv \neg \neg \varphi$;
    \item\labelx{crem}
      $\varphi_{1} \wedge \varphi_{2} \implies \varphi_{1}$;
    \item\labelx{dadd}
      $\varphi_{1} \implies \varphi_{1} \vee \varphi_{2}$;
    \item\labelx{cass}
      $\varphi_{1} \!\wedge\! (\varphi \!\wedge\! \varphi_{2}) \equiv
      (\varphi_{1} \!\wedge\! \varphi) \!\wedge\! \varphi_{2}$;
    \item\labelx{dass}
      $\varphi_{1} \!\vee\! (\varphi \!\vee\! \varphi_{2}) \equiv (\varphi_{1}
      \!\vee\! \varphi) \!\vee\! \varphi_{2}$;
    \item\labelx{con}
      $\varphi_{1} \wedge \varphi_{2} \equiv \neg (\neg \varphi_{1} \vee \neg
      \varphi_{2})$;
    \item\labelx{dis}
      $\varphi_{1} \vee \varphi_{2} \equiv \neg (\neg \varphi_{1} \wedge \neg
      \varphi_{2})$;
    \item\labelx{exs}
      $\exists^{\spcElm} \apElm \ldotp \varphi \equiv \neg (\forall^{\spcElm}
      \apElm \ldotp \neg \varphi)$;
    \item\labelx{all}
      $\forall^{\spcElm} \apElm \ldotp \varphi \equiv \neg (\exists^{\spcElm}
      \apElm \ldotp \neg \varphi)$.
  \end{inparaenum}
  }

Unlike \QPTL, \GFG-\QPTL in its full generality does not enjoy the \pnf
property.
This is a consequence of the information-leak phenomenon reported in
Example~\ref{exm:infleak}.
Indeed, $(\exists \apElm \ldotp \phi) \wedge \varphi \!\not\equiv\! \exists
\apElm \ldotp (\phi \wedge \varphi)$ and $(\forall \apElm \ldotp \phi) \vee
\varphi \!\not\equiv\! \forall \apElm \ldotp (\phi \vee \varphi)$, when $\apElm
\!\not\in\! \free{\varphi}$.
A similar problem arises in \IF due to \emph{signaling}, if one let
quantifications depend on non-free variables~\cite{MSS11}.
Fortunately, for the purposes of this work, we can focus on \pnf
formulae, since, as we shall show, behavioral \GFG-\QPTL is powerful enough to
express all $\omega$-regular languages.

We now introduce an operator on quantifier prefixes, called \emph{evolution},
that, given an arbitrary hyperassignment $\AsgFam$ and one of its two
interpretations $\alpha$, computes the result $\evlFun[\alpha](\AsgFam,
\qntElm)$ of the application to $\AsgFam$ of all quantifiers $\Qnt^{\spcElm}
\apElm$ occurring in a prefix $\qntElm$ in that specific order.
To this aim, we need to introduce the notion of \emph{coherence} of a quantifier
symbol $\Qnt \!\in\! \{ \exists, \forall \}$ \wrt an alternation flag $\alpha
\!\in\! \{ \QEA, \allowbreak \QAE \}$ as follows: $\Qnt$ is
\emph{$\alpha$-coherent} if either $\alpha = \QEA$ and $\Qnt = \exists$ or
$\alpha = \QAE$ and $\Qnt = \forall$.
Essentially, the evolution operator iteratively applies the semantics of
quantifiers, as defined by Items~\ref{def:althodsem(exs:ea)}'
and~\ref{def:althodsem(all:ae)}' of Definition~\ref{def:extalthodsem} and
Items~\ref{def:althodsem(exs:ae)} and~\ref{def:althodsem(all:ea)}
of Definition~\ref{def:althodsem}, for all the quantifiers $\Qnt^{\spcElm}
\apElm$ in the input prefix $\qntElm$.
For a single quantifier, $\evlFun[\alpha](\AsgFam, \Qnt^{\spcElm} \apElm)$ just
corresponds to the $\spcElm$-extension of $\AsgFam$ with $\apElm$, when $\Qnt$
is $\alpha$-coherent.
On the contrary, when $\Qnt$ is $\dual{\alpha}$-coherent, we need to dualize the
$\spcElm$-extension with $\apElm$ of the dual of $\AsgFam$.
\[
  \evlFun[\alpha](\AsgFam, \Qnt^{\spcElm} \apElm) \defeq
  \begin{cases}
    \extFun[\spcElm]{\AsgFam, \apElm},
      & \text{ if } \Qnt \text{ is }  \alpha\text{-coherent}; \\
    \dual{\extFun[\spcElm]{\dual{\AsgFam}, \apElm}},
      & \text{ otherwise}.
  \end{cases}
\]
The operator lifts naturally to an arbitrary quantification prefix $\qntElm
\!\in\!
\QntSet$ as follows:
\begin{inparaenum}
  \item
    $\evlFun[\alpha](\AsgFam, \epsilon) \defeq \AsgFam$;
  \item
    $\evlFun[\alpha](\AsgFam, \Qnt^{\spcElm} \apElm \ldotp \qntElm) \defeq
    \allowbreak \evlFun[\alpha](\evlFun[\alpha](\AsgFam, \Qnt^{\spcElm} \apElm),
    \qntElm)$.
\end{inparaenum}
We also set $\evlFun[\alpha](\qntElm) \defeq \evlFun[\alpha](\{ \{ \emptyfun \}
\}, \qntElm)$.

It is easy to show that the evolution operator is monotone \wrt $\sqsubseteq$,
by simply exploiting the monotonicity of the dualization and extension operators
given in Proposition~\ref{prp:oprmon}.

\defenv{proposition}[][Prp][EvlMon]
  {
  Let $\AsgFam[1], \AsgFam[2] \!\in\! \HAsgSet$ with $\AsgFam[1] \!\sqsubseteq\!
  \AsgFam[2]$ and $\qntElm \!\in\! \QntSet$: $\evlFun[\alpha](\AsgFam[1],
  \qntElm) \sqsubseteq \evlFun[\alpha](\AsgFam[2], \qntElm)$.
  }

By simple structural induction on a quantifier prefix $\qntElm \in \QntSet$, we
can show that a hyperassignment $\AsgFam$ $\alpha$-satisfies a formula $\qntElm
\phi$ \iff its $\alpha$-evolution \wrt $\qntElm$ $\alpha$-satisfies $\phi$.

\defenv{lemma}[Prefix Evolution][Lmm][PrfEvl]
  {
  Let $\qntElm \phi$ be a \GFG-\QPTL formula with $\qntElm \!\in\! \QntSet$.
  Then, $\AsgFam \cmodels[][\alpha] \qntElm \phi$ \iff $\evlFun[\alpha](\AsgFam,
  \qntElm) \cmodels[][\alpha] \phi$, for all $\AsgFam \in \HAsgSet(\free{\qntElm
  \phi})$.
  }

\section{Quantification Games}
\label{sec:canqnt}

The solution of the satisfiability problem for the behavioral fragment of
\GFG-\QPTL relies on the existence of a game, played by Eloise and Abelard, with
the property that Eloise wins the game \iff the corresponding formula is indeed
satisfiable.
We provide here a general result, showing that, for any quantifier prefix
$\qntElm$ and Borelian property $\Psi$, there exists a game, called
\emph{quantification game}, such that Eloise wins the game \iff the
hyperassignment obtained by evaluating the prefix, namely
$\evlFun[\exists\forall](\qntElm)$, contains a set of assignments completely
included in $\Psi$.
The correctness of this result depends, in turn, on the existence of canonical
forms for the quantifier prefixes that allow one to reduce the alternations to at most one.

\subsection{Quantification Game for Sentences}

To define the quantification game, we first need a few notions.
A two-player turn-based \emph{arena} $\Arena = \tuple {\PPosSet} {\OPosSet}
{\iposElm} {\MovRel}$ is a tuple where
\begin{inparaenum}
  \item
    $\PPosSet$ and $\OPosSet$ are the sets of \emph{positions} of \emph{Eloise}
    and \emph{Abelard}, \aka \emph{Player} and \emph{Opponent}, respectively,
    with $\PPosSet \cap \OPosSet = \emptyset$,
  \item
    $\iposElm \in \PosSet \defeq \PPosSet \cup \OPosSet$ is the \emph{initial
    position}, and
  \item
    $\MovRel \subseteq \PosSet \times \PosSet$ is the binary relation describing
    all possible \emph{moves} such that $\tuple{\PosSet}{\MovRel}$ is a sinkless
    directed graph.
\end{inparaenum}
A \emph{game} $\Game = \tuple {\Arena} {\ObsSet} {\WinSet}$ is a tuple where
$\Arena$ is an arena, $\ObsSet \subseteq \PosSet$ is the set of \emph{observable
positions}, and $\WinSet \subseteq \ObsSet^{\omega}$ is the set of
\emph{observable sequences} that are \emph{winning} for Eloise; the complement
$\dual{\WinSet} \defeq \ObsSet^{\omega} \setminus \WinSet$ is \emph{winning} for
Abelard.
Eloise (\resp, Abelard) \emph{wins} the game if she (\resp, he) has a strategy
such that, for all adversary strategies, the corresponding play induces an
observation sequence belonging (\resp, not belonging) to $\WinSet$.
All notions necessary to formalize this intuition, as \emph{history},
\emph{strategy}, and \emph{play}, are given in Appendix~\ref{app:canqnt}.\!

\begin{wrapfigure}{r}{0.45\linewidth}
  \vspace{-1.5em}
  \begin{center}
    \scalebox{0.85}[0.65]{\figqntgam}
  \end{center}
  \vspace{-1em}
  \caption{\label{fig:gntgam} Quantification game for $\qntElm = \exists^{\BSym}
  \apElm[1] \ldotp \forall^{\BSym} \apElm[2]\ldotp \exists^{\BSym} \apElm[3]
  \cdots$. Eloise owns the circled positions, while Abelard the squared ones.
  From the total-valuation positions $\valElm[1], \ldots, \valElm[n]$, with $n =
  \pow{\card{\ap{\qntElm}}}$, Abelard moves to the initial position.}
  \vspace{-1em}
\end{wrapfigure}

Martin's determinacy theorem~\cite{Mar75,Mar85} states that all games whose
winning condition is a Borel set in the Cantor topological space of infinite
words~\cite{PP04} are determined, \ie, either one of the two players necessarily
wins the game.
To ensures that the quantification game we are about to define is indeed
determined, we require a form of Borelian condition that can be applied to sets
of assignments.
This determinacy requirement is crucial here, since it is tightly connected with
the fact that \GFG-\QPTL does not allow for undetermined formulae.
To this end, let $\ValSet \defeq \APSet \pto \SetB$ denote the set of Boolean
valuations for sets of propositions and $\ValSet(\PSet) \!\defeq\! \set{\!
\valElm \in \ValSet \!}{\! \dom{\valElm} = \PSet \!}$ the set of valuations for
propositions in $\PSet \!\subseteq\! \APSet$.
Also, $\#(\valElm) \defeq \card{\dom{\valElm}}$.
We can now define a bijection between sets of assignments over $\PSet$ and languages of
infinite words over the alphabet $\ValSet(\PSet)$.
Let $\wrdFun \colon \AsgSet(\PSet) \to \ValSet(\PSet)^{\omega}$ be the
\emph{word function} mapping each assignment $\asgElm \in \AsgSet(\PSet)$ to the
word $\wrdElm \defeq \wrdFun(\asgElm) \in \ValSet(\PSet)^{\omega}$ satisfying
the equality $\asgElm(\apElm)(t) = (\wrdElm)_{t}(\apElm)$, for all $\apElm \in
\PSet$ and $t \in \SetN$.
Clearly $\wrdFun$ is a bijection.
Now, every property $\Psi \subseteq \AsgSet(\PSet)$, \ie, every set of
assignments, uniquely induces the language of infinite words $\wrdFun(\Psi)
\defeq \set{ \wrdFun(\asgElm) }{ \asgElm \in \Psi } \subseteq
\ValSet(\PSet)^{\omega}$ over the alphabet $\ValSet(\PSet)$.
Thus, $\Psi$ is said to be \emph{Borelian} (\resp, regular) if
$\wrdFun(\Psi)$ is a Borel (\resp, regular) set.


Given a behavioral sentence $\qntElm \psi$, let $\LangFun(\psi) \subseteq \AsgSet(\ap{\qntElm})$ denote the set
of assignments satisfying the \LTL formula $\psi$.
The quantification game $\Game[\qntElm][\psi] \defeq
\Game[\qntElm][\LangFun(\psi)]$ is defined in Construction~\ref{cns:qntgami} and
exemplified in Figure~\ref{fig:gntgam}.
\Wlogx, we assume that the prefix $\qntElm$ does not contain duplicates.
The positions of the game are (partial) valuations of the propositions in
$\qntElm$ and each position belongs to the player corresponding to the first
quantifier in the prefix whose proposition is not defined at that position.
The initial position of the game contains the empty valuation and in the example
of Figure~\ref{fig:gntgam} belongs to Eloise, she being the first to play in
$\qntElm$.
Obviously, the game features an infinite number of rounds.
Each round begins with the empty valuation and ends in a total valuation, after
players have chosen (jointly) a value for all the propositions.
A move in the round corresponds to a player choosing a value for the next
proposition in the prefix.
Take, for instance, position $\valElm \defeq \{ \apElm[1] \mapsto \Ff \}$ in the
figure, where the first proposition $\apElm[1]$ has been already assigned value
$\Ff$ by Eloise.
From that position, Abelard first chooses a Boolean value, say $\Tt$, for the
next proposition $\apElm[2]$ in the prefix.
Then he moves to the position $\valElm' \defeq \{ \apElm[1] \mapsto \Ff,
\apElm[2] \mapsto \Tt \}$, corresponding to the valuation ${\valElm}[\apElm[2]
\mapsto \Tt]$, obtained by extending $\valElm$ with the value chosen for
$\apElm[2]$.
Position $\valElm'$ belongs to Eloise, since the next quantifier $\exists
\apElm[3]$ in the prefix is existential.
The last positions belong to Abelard and, from there, he can only move back to
the starting position for the next turn.
By sampling any infinite sequence of rounds of the games at the positions with
total valuations, namely the observable positions, we obtain an infinite word
$\wrdElm$ corresponding to some assignment $\asgElm \defeq
\wrdFun[][-1](\wrdElm)$.
Then, $\wrdElm$ is winning for Eloise \iff $\asgElm$ belongs to $\Psi \! \defeq
\! \LangFun(\psi)$ (\ie, $\asgElm \! \models \! \psi$), while it is winning for
Abelard otherwise.
This intuition is formalized by the following construction.

\begin{construction}[Quantification Game I]
  \label{cns:qntgami}
  For every quantifier prefix $\qntElm \in \QntSet[\BSym]$ and property $\Psi
  \subseteq \AsgSet(\ap{\qntElm})$, the game $\Game[\qntElm][\Psi] \defeq
  \allowbreak \tuple {\Arena[\qntElm]} {\ObsSet} {\WinSet}$ with arena
  $\Arena[\qntElm] \defeq \tuple {\PPosSet} {\OPosSet} {\iposElm} {\MovRel}$ is
  defined as prescribed in the following:
  \begin{itemize}
    \item
      the set of positions $\PosSet \subset \ValSet$ contains exactly those
      valuations $\valElm \in \ValSet$ of the propositions in $\ap{\qntElm}$
      that are quantified in the prefix $(\qntElm)_{< \#(\valElm)}$ of $\qntElm$
      having length $\#(\valElm)$ \ie, $\dom{\valElm} = \ap{(\qntElm)_{<
      \#(\valElm)}}$;
    \item
      the set of Eloise's positions $\PPosSet \subseteq \PosSet$ only contains
      the valuations $\valElm \in \PosSet$ for which the proposition quantified
      in $\qntElm$ at index $\#(\valElm)$ is existentially quantified, \ie,
      $(\qntElm)_{\#(\valElm)} = \exists^{\BSym} \apElm$, for some $\apElm \in
      \ap{\qntElm}$;
    \item
      the initial position $\iposElm \defeq \emptyfun$ is just the empty
      valuation;
    \item
      the move relation $\MovRel \subseteq \PosSet \times \PosSet$ contains
      exactly those pairs of valuations $(\valElm[1], \valElm[2]) \in \PosSet
      \times \PosSet$ such that:
      \begin{itemize}
        \item
          $\valElm[1] \subseteq \valElm[2]$ and $\#(\valElm[2]) = \#(\valElm[1])
          + 1$, or
        \item
          $\valElm[1] \in \ValSet(\ap{\qntElm})$ and $\valElm[2] = \emptyfun$;
      \end{itemize}
    \item
      the set of observable positions $\ObsSet \defeq \ValSet(\ap{\qntElm})$
      precisely contains the valuations of all the propositions in $\qntElm$;
    \item
      the winning condition induced by the property $\Psi$ is the language of
      infinite words $\WinSet \defeq \wrdFun(\Psi)$ over
      $\ValSet(\ap{\qntElm})$.
  \end{itemize}
\end{construction}

The game $\Game[\qntElm][\psi]$ above essentially provides a game-theoretic
version of the semantics of behavioral quantifications.
The correctness of the game is established by the following theorem.

\defenv{theorem}[Game-Theoretic Semantics I][Thm][GamThrSemI]
  {
    A behavioral \GFG-\QPTL sentence $\qntElm \psi$, with $\psi \in \LTL$,
    is satisfiable (\resp, unsatisfiable) \iff Eloise (\resp,
    Abelard) wins $\Game[\qntElm][\psi]$.
  }

The proof of this result is split into the following three steps.
First, for an arbitrary behavioral quantifier prefix $\qntElm$, we provide two
transformations, $\CSym[\QEA](\qntElm)$ and $\CSym[\QAE](\qntElm)$, called
\emph{canonizations}, which allow one to reduce a behavioral \GFG-\QPTL sentence
$\varphi = \qntElm \psi$ to the sentences $\CSym[\QEA](\qntElm) \psi$ and
$\CSym[\QAE](\qntElm) \psi$ featuring at most one quantifier alternation.
Second, in Theorem~\ref{thm:qntgami}, we connect the winner of the game $\Game[\qntElm][\psi]$ with the
satisfiability of one of the normal forms $\CSym[\QEA](\qntElm) \psi$ and
$\CSym[\QAE](\qntElm) \psi$, showing also that $\CSym[\QEA](\qntElm) \psi$ implies $\CSym[\QAE](\qntElm) \psi$.
Finally, in Theorem~\ref{thm:sntcanfor}, we prove that the original sentence
$\varphi$ is equisatisfiable with the two normal forms.

Let us start with the definition of the two prefix canonizations based on the
following syntactic quantifier-swap operations.
Consider, \eg, the formula $\forall^{\BSym} \apElm \ldotp \exists^{\BSym}
\qapElm \ldotp \phi$.
A na\"ive quantifier-swap operator would simply swap the two quantifiers that, in game-theoretic
terms, corresponds to a swap in the choices of the two players, which allows
Abelard to see Eloise's move at the current round.
To balance this additional power, we simply restrict the universal quantifier
to be strictly behavioral, thus preventing Abelard from reading Eloise's choice.
This leads to the formula $\exists^{\BSym} \qapElm \ldotp \allowbreak
\forall^{\dep{\APSet}{\qapElm}} \apElm \ldotp \psi$.
A symmetric swap operation would transform the formula $\exists^{\BSym} \qapElm
\ldotp \forall^{\BSym} \apElm \ldotp \phi$ into $\forall^{\BSym} \apElm \ldotp
\exists^{\dep{\APSet}{\apElm}} \qapElm \ldotp \phi$.
Essentially, the swap operation exchange the positions of two adjacent dual
behavioral quantifiers and restrict one of the two to be strongly behavioral
\wrt the proposition of the other one.
By iteratively swapping adjacent quantifiers and adjusting the quantifier
specification accordingly, we can reduce the quantifier alternation to at most
one.

For technical convenience we use a vector notation for the quantifier prefixes:
$\Qnt^{\vec{\spcElm}} \vec{\apElm} \ldotp \phi \defeq \Qnt^{(\vec{\spcElm})_{0}}
(\vec{\apElm})_{0} \ldotp \cdots \Qnt^{(\vec{\spcElm})_{k}} (\vec{\apElm})_{k}
\ldotp \phi$, where $\card{\vec{\apElm}} = \card{\vec{\spcElm}} = k + 1$.
We omit the vector symbol in $\vec{\spcElm}$ if this is just a sequence of
$\BSym$ or $\SSym$ specifications and consider $\vec{\apElm}$ as sets of
propositions when convenient.
We also define in a natural way the union of two quantifier specifications as
follows:
\begin{center}
  $\dep{ \PSet[\BSym 1] }{ \PSet[\SSym 1] } \cup \dep{ \PSet[\BSym 2] }{
  \PSet[\SSym 2] } \defeq \dep{ \PSet[\BSym 1] \cup \PSet[\BSym 2] }{
  \PSet[\SSym 1] \cup \PSet[\SSym 2] }$.
\end{center}

We now introduce the two syntactic transformations $\CSym[\QEA](\cdot)$ and
$\CSym[\QAE](\cdot)$ that, given a behavioral quantifier prefix
$\qntElm \in \QntSet[\BSym]$, return the single-alternation ones
$\CSym[\QEA](\qntElm)$ and $\CSym[\QAE](\qntElm)$, by applying all the
quantifier swap operations at once.
More specifically, the function $\CSym[\QEA](\cdot)$ provides an
\emph{$\QEA$-prefix}, \ie, all existential quantifier precede the universal
ones, while $\CSym[\QAE](\cdot)$ gives us the the dual \emph{$\QAE$-prefix}.

For the definition of $\CSym[\QEA](\cdot)$, we observe that every behavioral
quantifier prefix $\qntElm$ can be uniquely rewritten as $\exists^{\BSym}
\vec{\qapElm}[0] \ldotp (\forall^{\BSym} \vec{\apElm}[i] \ldotp \allowbreak
\exists^{\BSym} \vec{\qapElm}[i])_{i = 1}^{k} \ldotp \forall^{\BSym}
\vec{\apElm}[k + 1]$, for some $k \!\in\! \SetN$ and vectors $\vec{\qapElm}[i]$,
with $i \!\in\! \numcc{0}{k}$, and $\vec{\apElm}[i]$, with $i \in \numcc{1}{k +
1}$, where $\card{\vec{\qapElm}[i]}, \card{\vec{\apElm}[i]} \geq 1$, for all $i
\in \numcc{1}{k}$.
For a quantifier prefix $\qntElm$ in such a form, we then define
\vspace{-0.5em}
\[
  \CSym[\exists\forall](\qntElm) \defeq (\exists^{\BSym} \vec{\qapElm}[i])_{i =
  0}^{k} \ldotp (\forall^{\vec{\spcElm}[i]} \vec{\apElm}[i])_{i = 1}^{k + 1},
  \vspace{-0.5em}
\]
where $(\vec{\spcElm}[i])_{j} \defeq \BSym \cup \sdep{ \vec{\qapElm}[i] \cdots
\vec{\qapElm}[k] }$, for all $j \in \numco{0}{\card{\vec{\apElm}[i]}}$.

The definition of $\CSym[\QAE](\cdot)$ is analogous.
First, we rewrite a
prefix $\qntElm$ as $\forall^{\BSym} \vec{\apElm}[0] \ldotp (\exists^{\BSym}
\vec{\qapElm}[i] \ldotp \forall^{\BSym} \vec{\apElm}[i])_{i = 1}^{k} \ldotp
\exists^{\BSym} \vec{\qapElm}[k + 1]$, for some $k \!\in\! \SetN$ and vectors
$\vec{\apElm}[i]$, with $i \!\in\! \numcc{0}{k}$, and $\vec{\qapElm}[i]$, with
$i \in \allowbreak \numcc{1}{k + 1}$, where $\card{\vec{\apElm}[i]},
\card{\vec{\qapElm}[i]} \!\geq\! 1$, for all $i \!\in\! \numcc{1}{k}$.
Then, we define
\[
  \CSym[\forall\exists](\qntElm) \defeq (\forall^{\BSym} \vec{\apElm}[i])_{i =
  0}^{k} \ldotp (\exists^{\vec{\spcElm}[i]} \vec{\qapElm}[i])_{i = 1}^{k + 1},
\]
where $(\vec{\spcElm}[i])_{j} \defeq \BSym \cup \sdep{ \vec{\apElm}[i] \cdots
\vec{\apElm}[k] }$, for all $j \in \numco{0}{\card{\vec{\qapElm}[i]}}$.

\begin{example}
  \label{exm:ceacae}
  Consider the behavioral quantifier prefix $\qntElm = \forall^{\BSym} \apElm
  \ldotp  \allowbreak \exists^{\BSym} \qapElm\, \rapElm \ldotp \forall^{\BSym}
  \sapElm \ldotp \exists^{\BSym} \tapElm$.
  The corresponding $\QEA$ canonical-form is $\CSym[\QEA](\qntElm) \allowbreak
  \!=\! \exists^{\BSym} \qapElm\, \rapElm\, \tapElm \ldotp\!
  \forall^{\spcElm[][\apElm]}\! \apElm \ldotp\! \forall^{\spcElm[][\sapElm]}\!
  \sapElm$, where $\spcElm[][\apElm] \defeq \dep{\APSet}{\qapElm\, \rapElm\,
  \tapElm}$ and $\spcElm[][\sapElm] \defeq \dep{\APSet}{\tapElm}$.
  The $\QAE$ canonical-form prefix is, instead, $\CSym[\QAE](\qntElm) \!=\!
  \forall^{\BSym} \apElm\, \sapElm \ldotp\! \exists^{\spcElm} \qapElm\, \rapElm
  \ldotp\! \exists^{\BSym} \tapElm$, where $\spcElm \defeq
  \dep{\APSet}{\sapElm}$.
\end{example}

For the second part of the proof of Theorem~\ref{thm:gamthrsemi}, we need to
connect the winner of $\Game[\qntElm][\psi]$ with the satisfiability of (one among)
$\CSym[\QEA](\qntElm) \psi$ and $\CSym[\QAE](\qntElm) \psi$.
This also corresponds to showing that $\CSym[\QEA](\qntElm) \psi \implies
\CSym[\QAE](\qntElm) \psi$.
To this end, we exploit the $\omega$-regularity of \LTL languages, which ensures
that the game is Borelian.

\defenv{theorem}[Quantification Game I][Thm][QntGamI]
  {
  For each behavioral quantification prefix $\qntElm \!\in\! \QntSet[\BSym]$ and
  Borelian property $\Psi \!\subseteq\! \AsgSet(\ap{\qntElm})$, the game
  $\Game[\qntElm][\Psi]$\! satisfies the following two properties:
  \begin{enumerate}[1)]
    \item\labelx{elo}
      if Eloise wins then $\ESet \subseteq \Psi$, for some $\ESet \in
      \evlFun[\QEA](\CSym[\QAE](\qntElm))$;
    \item\labelx{abe}
      if Abelard wins then $\ESet \not\subseteq \Psi$, for all $\ESet \in
      \evlFun[\QEA](\CSym[\QEA](\qntElm))$.
  \end{enumerate}
  }

The idea of the proof is to extract, from a winning strategy of Eloise (\resp,
Abelard), a vector $\vec{\FFun}$ of functors, one for each proposition
associated with that player, witnessing the existence (\resp, non-existence) of
a set $\ESet$ of assignments satisfying the property $\Psi$.
More precisely, assume Eloise has a strategy $\strElm$ to win the game and let
$\forall^{\BSym} \vec{\apElm} \ldotp \exists^{\vec{\spcElm}} \vec{\qapElm} =
\CSym[\QAE](\qntElm)$ be the $\QAE$ canonical-form of $\qntElm$.
Then, thanks to the bijection between plays $\playElm$ and assignments
$\asgElm$, we can operate as follows, for every round $k$ and existential
proposition $\qapElm[i]$ in $\vec{\qapElm}$: given Abelard's choices up to round
$k$ in $\playElm$, we can extract, from Eloise's response for $\qapElm[i]$ in
$\strElm$, the response to $\asgElm$ at time $k$ of the functor $\FFun[i]$ in
$\vec{\FFun}$.
As a consequence, for all $\asgElm \in \AsgSet(\vec{\apElm})$ chosen by Abelard,
Eloise's response corresponding to the extension of $\asgElm$ with $\vec{\FFun}$
on $\vec{\qapElm}$ satisfies, \ie, belongs to, the property $\Psi$.
The witness $\ESet$ is precisely the set of all those extensions.
An analogous argument applies to Abelard for the $\QEA$ canonical-form.
Notice that $\vec{\FFun}$ meets the specification $\vec{\spcElm}$ thanks to the
alternation of the players prescribed by $\qntElm$ in each round of
$\Game[\qntElm][\Psi]$.

The final step establishes the equisatisfiability of a behavioral \GFG-\QPTL
sentence with its canonical forms.

\defenv{theorem}[Sentence Canonical Forms][Thm][SntCanFor]
  {
  For every behavioral \GFG-\QPTL sentence $\qntElm \psi$, with $\psi \in \LTL$,
  it holds that $\qntElm \psi$, $\CSym[\QEA](\qntElm) \psi$, and
  $\CSym[\QAE](\qntElm) \psi$ are equisatisfiable.
  }

Towards the proof, we can derive the chain of implications $\CSym[\QAE](\qntElm)
\psi \implies \qntElm \psi \implies \CSym[\QEA](\qntElm) \psi$ by exploiting
the following property of the evolution function.
Specifically, this asserts a total ordering \wrt the preorder $\sqsubseteq$
between a behavioral quantifier prefix $\qntElm$ and its two canonical forms
$\CSym[\QEA](\qntElm)$ and $\CSym[\QAE](\qntElm)$ that can be proved by
induction on the structure of $\qntElm$.


\defenv{proposition}[][Prp][QntSwpInd]
  {
  $\!\evlFun[\alpha](\AsgFam, \CSym[\dual{\alpha}](\qntElm)) \!\sqsubseteq\!
  \evlFun[\alpha](\AsgFam, \qntElm) \!\sqsubseteq\! \evlFun[\alpha](\AsgFam,
  \CSym[\alpha](\qntElm))$, for $\AsgFam \in \HAsgSet$ and $\qntElm \in
  \QntSet[\BSym]$, with $\ap{\qntElm} \cap \ap{\AsgFam} = \emptyset$.
  }

From this result, Lemma~\ref{lmm:prfevl}, and Theorem~\ref{thm:hypref}, the above
implications immediately follow.
To complete the proof, we need to show that $\CSym[\QEA](\qntElm) \psi \implies
\CSym[\QAE](\qntElm) \psi$ holds.
Thanks to Lemma~\ref{lmm:prfevl}, if $\{\!\{ \emptyfun \}\!\} \cmodels[][\QEA]
\CSym[\QEA](\qntElm) \psi$, then $\ESet \subseteq \Psi \defeq \LangFun(\psi)$, for some $\ESet \in
\evlFun[\QEA](\CSym[\QEA](\qntElm))$.
Thus, by Item~\ref{thm:qntgami(abe)} of Theorem~\ref{thm:qntgami}, it follows
that Abelard loses the game $\Game[\qntElm][\Psi]$, which means, by determinacy,
that Eloise wins.
As a consequence of Item~\ref{thm:qntgami(elo)} of the same theorem, there
exists $\ESet \in \evlFun[\QEA](\CSym[\QAE](\qntElm))$ such that $\ESet
\subseteq \Psi$.
Hence, $\{\!\{ \emptyfun \}\!\} \cmodels[][\QEA] \CSym[\QAE](\qntElm) \psi$,
again by Lemma~\ref{lmm:prfevl}.


\subsection{Quantification Game for Formulae}

The game defined in the previous section can easily be adapted to deal with
the satisfiability problem for behavioral \GFG-\QPTL as shown in the
next section.
Solving the model-checking problem requires, however, a generalization of
Theorem~\ref{thm:gamthrsemi}, connecting a suitable game with satisfiability of
arbitrary behavioral formulae \wrt a hyperassignment $\AsgFam$.
We can prove such a property under the assumption that $\AsgFam$ is
\emph{well-behaved}, \ie, if
\begin{inparaenum}
  \item $\AsgFam$ is the evolution of a set of assignments $\XSet$ \wrt some
    behavioral prefix $\trn{\qntElm}$ and
  \item
    $\AsgFam$ is Borelian.
\end{inparaenum}
The Borelian requirement is again connected to determinacy of the underlying
game.
The behavioral requirement, instead, allows for a simple proof that leverages
the quantification game for sentences directly.
At this stage, it is not clear whether the property actually holds for arbitrary
Borelian hyperassignments.

%
To formalize the above assumption, we introduce the notion of \emph{generator}
for a hyperassignment $\AsgFam \in \HAsgSet$ as a pair $\tuple {\trn{\qntElm}}
{\XSet}$ of
\begin{inparaenum}
  \item
    a behavioral quantification prefix $\trn{\qntElm} \in \QntSet$ and
  \item
    a Borelian set of assignments $\emptyset \neq \XSet \subseteq
    \AsgSet(\ap{\AsgFam} \setminus \ap{\trn{\qntElm}})$
\end{inparaenum}
such that $\AsgFam = \evlFun[\QEA](\{ \XSet \}, \trn{\qntElm})$.
A hyperassignment $\AsgFam \in \HAsgSet$ is \emph{Borelian behavioral} if there
is a generator for it.
%
%
%
A \emph{quantification-game schema} is a tuple $\QStr \defeq \tuple {\AsgFam}
{\qntElm} {\Psi}$ where
\begin{inparaenum}
  \item
    $\AsgFam {\in} \HAsgSet$ is Borelian behavioral,
  \item
    $\qntElm {\in} \QntSet$ is behavioral,
  \item
    $\Psi \!\subseteq\! \AsgSet(\ap{\qntElm} \cup \ap{\AsgFam})$ is Borelian,
    and
  \item
    $\ap{\qntElm} \cap \ap{\AsgFam} = \emptyset$.
\end{inparaenum}

The idea behind the game-theoretic construction reported below is quite simple.
Given a generator $\tuple {\trn{\qntElm}} {\XSet}$ for a behavioral
hyperassignment $\AsgFam$, we force the two players to simulate the given
$\AsgFam$ by playing according to the prefix $\trn{\qntElm}$, once Abelard has
arbitrarily chosen the values of the atomic propositions $\vec{\apElm}$ over
which the set of assignments $\XSet$ is defined.
Since $\evlFun[\QEA](\forall \vec{\apElm}\,) = \{ \AsgSet(\vec{\apElm}) \}$ and
$\XSet \subseteq \AsgSet(\vec{\apElm}\,)$, it is clear that $\AsgFam \sqsubseteq
\evlFun[\QEA](\forall \vec{\apElm} \ldotp \trn{\qntElm})$.
Thus, if Eloise wins the game, she can ensure a given temporal property.
Notice, however, that we gave Abelard the freedom to cheat and choose arbitrary
values for $\vec{\apElm}$.
Thus, in principle, Eloise could be able to satisfy the property while loosing
the game, since Abelard can choose assignments over $\vec{\apElm}$ that do not
belong to $\XSet$.
To remedy this, we add all those assignments to Eloise's winning set, thus
deterring Abelard from cheating.

\begin{construction}[Quantification Game II] \label{cns:qntgamii}
  For a quantification-game schema $\QStr \defeq \tuple {\AsgFam} {\qntElm}
  {\Psi}$, we say that $\Game$ is a $\QStr$-game if there is a generator $\tuple
  {\trn{\qntElm}} {\XSet}$ for $\AsgFam$ such that $\Game \defeq
  \Game[\der{\qntElm}][\der{\Psi}]$, where \\
  \begin{minipage}{0.5\textwidth}
  \begin{itemize}
    \item
      $\der{\qntElm} \defeq \forall \vec{\apElm} \ldotp \trn{\qntElm} \ldotp
      \qntElm$ and
  \end{itemize}
  \end{minipage}
  \begin{minipage}{0.5\textwidth}
  \begin{itemize}
    \item
      $\der{\Psi} \defeq \Psi \cup \set{ \asgElm \in \AsgSet(\PSet) }{
      \asgElm \rst {\vec{\apElm}} \,\not\in \XSet }$,
  \end{itemize}
  \end{minipage}
  with $\vec{\apElm} \defeq \ap{\AsgFam} \setminus \ap{\trn{\qntElm}}$ and
  $\PSet \defeq \ap{\qntElm} \cup \ap{\AsgFam}$.
\end{construction}

The quantification-game schema for
$\qntElm \psi$, with $\psi \in \LTL$, and a hyperassignment $\AsgFam$ is the
tuple $\QStr[\qntElm \psi][\AsgFam] \defeq \tuple {\AsgFam} {\qntElm}
{\LangFun(\psi)}$.

\defenv{theorem}[Game-Theoretic Semantics II][Thm][GamThrSemII]
  {
  $\AsgFam \cmodels[][\QEA] \qntElm \psi$ \iff Eloise wins every
  $\QStr[\qntElm \psi][\AsgFam]$-game, for all behavioral \GFG-\QPTL formulae
  $\qntElm \psi$, with $\psi \in \LTL$, and Borelian behavioral hyperassignments $\AsgFam \in
  \HAsgSet(\free{\qntElm \psi})$.
  }

The proof is similar to the one of
Theorem~\ref{thm:gamthrsemi} and
uses the following result, which generalizes
Theorem~\ref{thm:qntgami}.


\defenv{theorem}[Quantification Game II][Thm][QntGamII]
  {
  Every $\QStr$-game $\Game$, for some quantification-game
  schema $\QStr \defeq \tuple {\qntElm} {\AsgFam} {\Psi}$, satisfies the
  following two properties:
  \begin{enumerate}[1)]
    \item\labelx{elo}
      if Eloise wins then $\ESet \subseteq \Psi$, for some $\ESet \in
      \evlFun[\QEA](\AsgFam, \CSym[\QAE](\qntElm))$;
    \item\labelx{abe}
      if Abelard wins then $\ESet \not\subseteq \Psi$, for all $\ESet \in
      \evlFun[\QEA](\AsgFam, \CSym[\QEA](\qntElm))$.
  \end{enumerate}
  }

%
  Theorem~\ref{thm:qntgamii}, together with Proposition~\ref{prp:qntswpind},
  lifts Theorem~\ref{thm:sntcanfor} to formulae, and allows us to obtain
  Theorem~\ref{thm:gamthrsemii}.

\defenv{theorem}[Formula Canonical Forms][Thm][FrmCanFor]
  {
  For every behavioral \GFG-\QPTL formula $\qntElm \psi$, with $\psi \in \LTL$,
  it holds that $\AsgFam \cmodels[][\alpha] \qntElm \psi$ \iff $\AsgFam
  \cmodels[][\alpha] \CSym[\QEA](\qntElm) \psi$ \iff $\AsgFam \cmodels[][\alpha]
  \CSym[\QAE](\qntElm) \psi$, for all Borelian behavioral hyperassignments
  $\AsgFam \in \HAsgSet(\free{\qntElm \psi})$.
  }







\section{Decision Problems \& Expressiveness}
\label{sec:decprb}

The results of the previous section can be exploited to solve optimally the
decision problems for behavioral \GFG-\QPTL.
More specifically, we can use the game of Constructions~\ref{cns:qntgami} for
the satisfiability problem, and the game of Constructions~\ref{cns:qntgamii} for
the model-checking one.
We also discuss the expressiveness relationship between \QPTL and behavioral
\GFG-\QPTL, showing, by means of a classic encoding of automata into logic, that
they have the same expressive power, though \QPTL is non-elementary more
succinct then \GFG-\QPTL.

\subsection{Decision Procedures}

The first step in deciding the satisfiability problem is to derive from a
behavioral sentence $\varphi = \qntElm \psi$ a \emph{parity
game}~\cite{Mos91,EJ91} that is won by Eloise \iff $\varphi$ is satisfiable.
To do that, we first construct a deterministic parity automaton $\DName[\psi]$
for the \LTL formula $\psi$, by combining the Vardi-Wolper
construction~\cite{VW86a} with the Safra-like translation from B\"uchi to parity
acceptance condition~\cite{Pit06}.
We then compute the synchronous product of the arena $\Arena[\qntElm]$ of
Construction~\ref{cns:qntgami} with $\DName[\psi]$, where the automaton
component changes its state only when Abelard moves from observable positions
containing a full valuation of the propositions.
This valuation is, then, read by the transition function of $\DName[\psi]$ to
determine its successor state.
The resulting game simulates both the quantification game and the automaton, so
that Eloise wins \iff the play satisfies $\psi$.

\defenv{theorem}[Satisfiability Game][Thm][SatGam]
  {
  For every behavioral \GFG-\QPTL sentence $\varphi$ there is a parity game,
  with $\AOmicron{\pow{\pow{\card{\varphi}}}}$ positions and
  $\AOmicron{\pow{\card{\varphi}}}$ priorities, won by Eloise \iff $\varphi$ is
  satisfiable.
  }

We can then obtain an upper bound on the complexity of the problem from the fact
that parity games can be solved in time polynomial in the number of positions
and exponential in that of the priorities~\cite{EJ88,EJS93,Zie98}.
For the lower bound, instead, we observe that the \emph{reactive synthesis
problem}~\cite{PR89} of an \LTL formula $\psi$ can be reduced to the
satisfiability of a sentence of the form $\forall^{\BSym} \vec{\apElm} \ldotp
\exists^{\BSym} \vec{\qapElm} \ldotp \psi$, where $\vec{\apElm}$ and
$\vec{\qapElm}$ denote, respectively, the input and output signals of the
desired system.

\defenv{theorem}[Satisfiability][Thm][SatCmp]
  {
  The satisfiability problem for behavioral \GFG-\QPTL sentences is 2\ExpTimeC.
  }


For the universal (\resp, existential) model-checking problem, given a Kripke
structure $\KName$, we ask whether $\KName \models \varphi$, in the sense that
$\AsgFam[\KName] \cmodels[][\QEA] \varphi$ (\resp, $\AsgFam[\KName]
\cmodels[][\QAE] \varphi$) holds, where $\AsgFam[\KName] \defeq \allowbreak \{
\wrdFun[][-1](\LangFun(\KName)) \}$ is the hyperassignment obtained by collecting all
those assignments $\asgElm \!\in\! \AsgSet(\ap{\KName})$ over the propositions
of $\KName$ for which the infinite word $\wrdFun(\asgElm)$ belongs to the
$\omega$-regular language $\LangFun(\KName)$ generated by $\KName$.
Obviously, $\AsgFam[\KName]$ is a Borelian behavioral hyperassignment.
As a consequence, Construction~\ref{cns:qntgamii} applies.
Thus, we can adopt the same synchronous product described above between the
arena of the game and the union of the two automata $\DName[\psi]$ and
$\NName[\dual{\KName}]$, where $\DName[\psi]$ is obtained from the formula
$\psi$, while $\NName[\dual{\KName}]$ is a co-safety automaton of size linear in
$\card{\KName}$, recognizing the complement of $\LangFun(\KName)$.


\defenv{theorem}[Model-Checking Game][Thm][MCGam]
  {
  For every Kripke structure $\KName$ and behavioral \GFG-\QPTL formula
  $\varphi$, with $\free{\varphi} \!\subseteq\! \ap{\KName}$, there is a parity
  game, with $\AOmicron{\pow{\pow{\card{\varphi}}} \card{\KName}}$ positions and
  $\AOmicron{\pow{\card{\varphi}}}$ priorities, won by Eloise \iff $\KName
  \!\models\! \varphi$.
  }


Upper bounds \wrt both formula and model complexity, and the lower bound \wrt
formula complexity, are proved as in the case of the satisfiability problem.
The lower bound \wrt model complexity is proved by reducing from reachability
games~\cite{Imm81}.

\defenv{theorem}[Model-Checking][Thm][MCCmp]
  {
  The model-checking problem for behavioral \GFG-\QPTL has 2\ExpTimeC formula
  complexity and \PTimeC model complexity.
  }


\subsection{Expressive Power}

We conclude the work by discussing the expressive power of the behavioral
fragment of \GFG-\QPTL, showing that it precisely corresponds to the
$\omega$-regular languages.
Similarly to Example~\ref{exm:gfgaut}, consider an arbitrary deterministic
parity automaton $\DName$ with $k$ states over an alphabet $\pow{\PSet}$, with
$\PSet \subseteq \APSet$.
Via the standard  technique of encoding the existence of an accepting run, we
can construct an \LTL formula $\psi$, over the set of propositions $\PSet \cup
\{ \sapElm[1], \ldots, \sapElm[k] \}$, such that the existential projection on
$\PSet$ of the language $\LangFun(\psi)$ coincides with the language
$\LangFun(\DName)$ recognized by $\DName$.
Since $\DName$ is deterministic, this projection is clearly behavioral.
Hence, the behavioral \GFG-\QPTL formula $\exists^{\BSym} \sapElm[1], \ldots,
\exists^{\BSym} \sapElm[k] \ldotp \psi$ is satisfied by an hyperassignment $\{
\{ \asgElm \} \}$ \iff $\wrdFun(\asgElm) \!\in\! \LangFun(\DName)$.
Since every \QPTL formula can be translated into an equivalent nondeterministic
B\"uchi automaton~\cite{SVW87}, which in turn can be determinized into a parity
one~\cite{Pit06}, we obtain that for every \QPTL formula, there is an
equivalent behavioral \GFG-\QPTL one.
The converse holds as well.
Indeed, the satisfiability game $\Game[\varphi]$ can be transformed into a
isomorphic alternating parity word automaton $\AName[\varphi]$, in the usual
way, which can then be reduced to a nondeterministic parity automaton
$\NName[\varphi]$~\cite{MS95}.
The emptiness of $\NName[\varphi]$ can then be encoded into a \QPTL sentence.
A similar reasoning applies also to formulae.



\begin{theorem}[Expressiveness]
  \label{thm:eqv}
  \QPTL and behavioral \GFG-\QPTL are equi-expressive.
\end{theorem}

Clearly, \QPTL is also non-elementary more succinct than behavioral \GFG-\QPTL.
Indeed, the satisfiability problem for \QPTL sentences with alternation of
quantifiers $k$ is $(k - 1)$-\ExpSpaceC~\cite{SVW87}, while behavioral
\GFG-\QPTL is decidable in 2\ExpTime, so no elementary reduction exists.

\begin{theorem}[Succinctness]
  \label{thm:suc}
  \QPTL is non-elementary more succinct than behavioral \GFG-\QPTL.
\end{theorem}










\begin{section}{Discussion}

  We have introduced a novel semantics for \QPTL extending Hodges' team
  semantics for Hintikka and Sandu's logic of imperfect information \IF in a
  non-trivial way.
  On the one hand, the new semantic setting can express games with both
  symmetric and asymmetric restrictions on the players.
  On the other hand, it allows for encoding behavioral constraints on the
  quantified propositions, connecting the underlying logic with the
  game-theoretic notion of behavioral strategies.
  Based on this semantics, the extension of \QPTL with constraints on the
  functional dependencies among propositions, called \GFG-\QPTL, has
  surprisingly interesting properties.
  For one, its behavioral fragment enables reducing the solution of two-player
  zero-sum games to the decision problems for the logic.
  Indeed, the deep connection with behavioral strategies ensures that
  satisfiable formulae of the logic express linear time properties that can
  always be realized by means of actual strategies.
  This fragment also enjoys good computational properties, being 2\ExpTimeC for
  both satisfiability and model-checking.
  It is also very expressive, being equivalent to, though less succinct than,
  \QPTL, hence able to describe all $\omega$-regular properties.
  Second, the behavioral semantics also bears a connection to good-for-game
  automata, allowing to naturally express the property of being a GFG automata,
  the significance of which is probably worth investigating further.


  To the best of our knowledge, this is the first attempt to provide a
  compositional account of behavioral constraints.
  We believe the generality and flexibility of the semantic settings opens up
  the possibility of a systematic investigation of the impact of this type of
  constraints in quantified temporal logics, such as \QCTL~\cite{Fre11,LM14},
  Hyper\LTL/\CTLS~\cite{CFKMRS14,FRS15,FH16,FZ17,CFHH19}, Coordination
  Logic~\cite{FS10}, and Strategy Logic~\cite{CHP10,MMPV14}.

\end{section}





  \bibliographystyle{IEEEtranS}
  \bibliography{References}

  \clearpage
  \appendix



\subsection{Proofs for Section~\ref{sec:althodsem}}
\label{app:althodsem}


\recenv{PrpDltInv}
\begin{proof}
  To begin with, we show that $\AsgFam \subseteq
  \flipFun{\flipFun{\AsgFam}}$.
  By definition of $\dual{\AsgFam}$, for every $\dual{\XSet} \in
  \dual{\AsgFam}$ there is a function $\chcFun[\dual{\XSet}] \in
  \ChcSet{\AsgFam}$ such that $\dual{\XSet} = \set{
  \chcFun[\dual{\XSet}](\XSet)}{
  \XSet \in \AsgFam }$.
  Now, consider an arbitrary $\XSet \in \AsgFam$ and define $\chcFun$ as:
  $\chcFun(\dual{\XSet}) = \chcFun[\dual{\XSet}](\XSet)$ for every
  $\dual{\XSet} \in \dual{\AsgFam}$.
  Notice that $\chcFun(\dual{\XSet}) \in \dual{\XSet}$, for every
  $\dual{\XSet} \in \dual{\AsgFam}$, and thus $\chcFun \in
  \ChcSet{\dual{\AsgFam}}$.
  %
  %
  Therefore, we have that $\set{ \chcFun(\dual{\XSet}) }{ \dual{\XSet} \in
  \dual{\AsgFam} } \in \dual{\dual{\AsgFam}}$.
  To conclude the proof, we are left to show that $\set{ \chcFun(\dual{\XSet}) }{
  \dual{\XSet} \in \dual{\AsgFam} } = \XSet$ holds as well.
  First, observe that $\chcFun(\dual{\XSet}) = \chcFun[\dual{\XSet}](\XSet)
  \in \XSet$ holds for every $\dual{\XSet} \in \dual{\AsgFam}$, implying $\set{
  \chcFun(\dual{\XSet}) }{ \dual{\XSet} \in \dual{\AsgFam} } \subseteq
  \XSet$.
  In order to show the converse inclusion ($\set{ \chcFun(\dual{\XSet}) }{
  \dual{\XSet} \in \dual{\AsgFam} } \supseteq \XSet$), consider an
  arbitrary
  $\asgElm \in \XSet$ and a function $\chcFun[\XSet_{\asgElm}] \in
  \ChcSet{\AsgFam}$ such that $\chcFun[\XSet_{\asgElm}](\XSet) = \asgElm$.
  Let $\XSet_{\asgElm} \defeq \set{ \chcFun[\XSet_{\asgElm}](\XSet) }{ \XSet
  \in \AsgFam }$.
  It holds $\XSet_{\asgElm} \in \dual{\AsgFam}$.
  Since $\chcFun(\XSet_{\asgElm}) = \chcFun[\XSet_{\asgElm}](\XSet) =
  \asgElm$, we have that $\asgElm \in \set{ \chcFun(\dual{\XSet}) }{ \dual{\XSet}
  \in
  \dual{\AsgFam} }$ and, since $\asgElm$ was chosen arbitrarily, we
  conclude $\set{
  \chcFun(\dual{\XSet}) }{ \dual{\XSet} \in \dual{\AsgFam} } \supseteq
  \XSet$.

  Observe that, straightforwardly, $\AsgFam \subseteq
  \flipFun{\flipFun{\AsgFam}}$ implies $\AsgFam \sqsubseteq
  \flipFun{\flipFun{\AsgFam}}$.

  Let us turn now to proving $\flipFun{\flipFun{\AsgFam}} \sqsubseteq
  \AsgFam$.
  Let $\dual{\dual{\XSet}} \in \dual{\dual{\AsgFam}}$.
  By definition of $\dual{\dual{\AsgFam}}$, there is a function
  $\chcFun[\dual{\dual{\XSet}}] \in \ChcSet{\dual{\AsgFam}}$ such
  that
  $\dual{\dual{\XSet}} = \set{ \chcFun[\dual{\dual{\XSet}}](\dual{\XSet}) }{
  \dual{\XSet} \in \dual{\AsgFam} }$.
  Towards a contradiction, assume that for every $\XSet \in \AsgFam$ there
  is $\asgElm_{\XSet} \in \XSet \setminus \dual{\dual{\XSet}}$.
  Let us define $\chcFun$ as: $\chcFun(\XSet) = \asgElm_{\XSet}$ for every
  $\XSet \in \AsgFam$.
  Notice that $\chcFun \in \ChcSet{\AsgFam}$
  Thus, $\dual{\XSet} \defeq \set{ \asgElm_{\XSet} }{ \XSet \in \AsgFam } \in
  \dual{\AsgFam}$ and $\dual{\XSet} \cap \dual{\dual{\XSet}} = \emptyset$.
  However, $\chcFun[\dual{\dual{\XSet}}](\dual{\XSet}) \in
  \dual{\dual{\XSet}} \cap \dual{\XSet}$, thus rising a contradiction.
\end{proof}

\recenv{LmmFlp}
\begin{proof}
  \begin{itemize}
  \item[($\ref{lmm:flp(ea:org)} \Rightarrow \ref{lmm:flp(ea:dlt)}$)]
  By~$\ref{lmm:flp(ea:org)}$, there is $\XSet \in \AsgFam$ such that
  $\asgElm
  \models \varphi$ holds for every $\asgElm \in \XSet$.
  By definition of $\dual{\AsgFam}$, for every $\dual{\XSet} \in
  \dual{\AsgFam}$ there is $\chcFun[\dual{\XSet}]$ such that
  $\chcFun[\dual{\XSet}](\XSet) \in \XSet$ and $\dual{\XSet} = \{
  \chcFun[\dual{\XSet}](\XSet) : \XSet \in \AsgFam \}$;
  by~$\ref{lmm:flp(ea:org)}$, $\chcFun[\dual{\XSet}](\XSet) \models
  \varphi$;
  since, in addition, $\chcFun[\dual{\XSet}](\XSet) \in \dual{\XSet}$, the
  thesis
  holds.

  \item[($\ref{lmm:flp(ea:dlt)} \Rightarrow \ref{lmm:flp(ea:org)}$)]
  By~$\ref{lmm:flp(ea:dlt)}$, for every $\dual{\XSet} \in \dual{\AsgFam}$
  there is
  $\asgElm_{\dual{\XSet}} \in \dual{\XSet}$ such that
  $\asgElm_{\dual{\XSet}}
  \models \varphi$.
  Consider function $\chcFun \in \ChcSet{\dual{\AsgFam}}$ defined
  as: $\chcFun(\dual{\XSet}) = \asgElm_{\dual{\XSet}}$, for every
  $\dual{\XSet} \in
  \dual{\AsgFam}$.
  By definition of $\dual{\dual{\AsgFam}}$, we have
  that $\{ \chcFun(\dual{\XSet}) : \dual{\XSet} \in \dual{\AsgFam} \} \in
  \dual{\dual{\AsgFam}}$.
  By Proposition~\ref{lmm:flp}, it holds $\dual{\dual{\AsgFam}}
  \sqsubseteq \AsgFam$, which means that there is $\XSet \in \AsgFam$,
  with $\XSet
  \subseteq \{ \chcFun(\dual{\XSet}) : \dual{\XSet} \in \dual{\AsgFam} \}$.
  Since, by construction, $\chcFun(\dual{\XSet}) \models \varphi$ for every
  $\dual{\XSet} \in \dual{\AsgFam}$, the thesis holds.

  \item[($\ref{lmm:blncon(ae:org)} \Leftrightarrow \ref{lmm:blncon(ae:par)}$)]
By
  statement~\ref{lmm:flp(ea)} of this lemma, we have that
  $\ref{lmm:blncon(ea:org)}$ is false if and only if $\ref{lmm:blncon(ea:par)}$
  is false
  ($\mathit{not} \ \ref{lmm:blncon(ea:org)} \Leftrightarrow \mathit{not} \
  \ref{lmm:blncon(ea:par)}$, for short).
  By instantiating, in this last equivalence, $\varphi$ with
  $\neg\varphi$, we have
  $\ref{lmm:blncon(ea:org)}' \Leftrightarrow \ref{lmm:blncon(ea:par)}'$, where
  $\ref{lmm:blncon(ea:org)}'$ and $\ref{lmm:blncon(ea:par)}'$ are abbreviations
  for,
  respectively:

  \begin{itemize}
  \item for all sets of assignments $\XSet
  \in \AsgFam$, there exists an assignment $\asgElm \in \XSet$ such that
  $\asgElm \not\models \neg\varphi$;
  \item there exists a set of assignments $\XSet \in \dual{\AsgFam}$ such
  that, for all assignments $\asgElm \in \XSet$, it holds that $\asgElm
  \not\models
  \neg \varphi$.
  \end{itemize}

  By applying semantics of negation, it is straightforward to see that
  $\ref{lmm:blncon(ea:org)}'$ and $\ref{lmm:blncon(ea:par)}'$ correspond to
  $\ref{lmm:blncon(ae:org)}$ and $\ref{lmm:blncon(ae:par)}$, respectively, hence
  the
  thesis.
    \qedhere
  \end{itemize}

\end{proof}

\recenv{LmmBlnCon}
\begin{proof}

  \begin{itemize}
  \item[($\ref{lmm:blncon(ea:org)} \Rightarrow \ref{lmm:blncon(ea:par)}$)]
  %
  %
  Let $\XSet \in \AsgFam$ be such that $\asgElm \models \varphi_{1}
  \wedge \varphi_{2}$ holds for every $\asgElm \in \XSet$ and consider an
  arbitrary pair $(\AsgFam[1], \AsgFam[2]) \in \parFun[]{\AsgFam}$.
  Since $(\AsgFam[1], \AsgFam[2])$ is a partition of $\AsgFam$, either
  $\XSet \in \AsgFam[1]$ or $\XSet \in \AsgFam[2]$: in the former case,
  let $i=1$;
  in the latter, let $i = 2$.
  Since $\XSet \in \AsgFam[i]$ and $\asgElm \models \varphi_{i} $ holds
  for every $\asgElm \in \XSet$, the thesis holds.

  \item[($\ref{lmm:blncon(ea:par)} \Rightarrow \ref{lmm:blncon(ea:org)}$)]
  Consider the hyperassignment $\dual{\AsgFam[1]} = \{ \XSet
  \in
  \AsgFam : \forall \asgElm \in \XSet \ . \ \asgElm \models \varphi_1 \}$
  and the
  pair $(\AsgFam[1] \defeq \AsgFam \setminus \dual{\AsgFam[1]}, \AsgFam[2]
  \defeq
  \dual{\AsgFam[1]}) \in \parFun[]{\AsgFam}$.
  Observe that, by definition of $\AsgFam[1]$, there is no $\XSet \in
  \AsgFam[1]$ such that $\asgElm \models \varphi_{1}$ holds for every
  $\asgElm \in
  \XSet$.
  Thus, by~$\ref{lmm:blncon(ea:par)}$, there must exist $\XSet \in
  \AsgFam[2]$ such that $\asgElm \models \varphi_{2}$ holds for every
  $\asgElm \in
  \XSet$.
  By definition of $\AsgFam[2]$, it also holds that $\asgElm \models
  \varphi_{1}$ for every $\asgElm \in \XSet$, hence the thesis.

  \item[($\ref{lmm:blncon(ae:org)} \Leftrightarrow \ref{lmm:blncon(ae:par)}$)]
By
  statement~\ref{lmm:flp(ea)} of this lemma, we have that
  $\ref{lmm:blncon(ea:org)}$ is false if and only if $\ref{lmm:blncon(ea:par)}$
  is false
  ($\mathit{not} \ \ref{lmm:blncon(ea:org)} \Leftrightarrow \mathit{not} \
  \ref{lmm:blncon(ea:par)}$, for short).
  By instantiating, in this last equivalence, $\varphi_1$ with
  $\neg\varphi_1$ and $\varphi_2$ with $\neg\varphi_2$, we have
  $\ref{lmm:blncon(ea:org)}' \Leftrightarrow \ref{lmm:blncon(ea:par)}'$, where
  $\ref{lmm:blncon(ea:org)}'$ and $\ref{lmm:blncon(ea:par)}'$ are abbreviations
  for,
  respectively:

  \begin{itemize}
  \item for all sets of assignments $\XSet
    \in \AsgFam$, there exists an assignment $\asgElm \in \XSet$ such that
    $\asgElm \not\models \neg\varphi_{1} \wedge \neg\varphi_{2}$;
  \item there exists a pair of hyperassignments $(\AsgFam[1],
    \AsgFam[2]) \in \parFun[]{\AsgFam}$ such that, for all indexes $i \in
  \{ 1, 2
    \}$ and sets of assignments $\XSet \in \AsgFam[i]$, there exists an
  assignment
    $\asgElm \in \XSet$ for which it holds that $\asgElm \not\models
    \neg\varphi_{i}$.
  \end{itemize}

  By applying semantics of negation and De Morgan's laws, it is
  straightforward to see that $\ref{lmm:blncon(ea:org)}'$ and
  $\ref{lmm:blncon(ea:par)}'$ correspond to $\ref{lmm:blncon(ae:org)}$ and
  $\ref{lmm:blncon(ae:par)}$, respectively, hence the thesis. \qedhere
  \end{itemize}

\end{proof}

\recenv{LmmHypExt}
\begin{proof}

  \begin{itemize}
  \item[($\ref{lmm:blncon(ea:org)} \Rightarrow \ref{lmm:blncon(ea:par)}$)] Let
  $\XSet \in \AsgFam$ be such that $\asgElm \models \exists \apElm .
  \varphi$
  holds for every $\asgElm \in \XSet$. By semantics
  (Def.~\ref{def:tarsem}, item
  \ref{def:tarsem(qnt:exs)}), for every $\asgElm \in \XSet$, there is a
  temporal
  function $\fFun_{\asgElm} \in \SetN \to \SetB$ such that
  ${\asgElm}[\apElm
  \mapsto \fFun_{\asgElm}] \models \varphi$.
  Let $\FFun \in \FncSet(\ap{\AsgFam})$ be such that
  $\FFun(\asgElm) = \fFun_{\asgElm}$ for every $\asgElm \in \XSet$ and let
  $\XSet_{\FFun} = \{ {\asgElm}[\apElm \mapsto \FFun(\asgElm)] : \asgElm
  \in \XSet
  \}$.
  Since $\XSet_{\FFun} \in \extFun[]{\AsgFam, \apElm}$ and $\asgElm
  \models \varphi$ holds for every $\asgElm \in \XSet_{\FFun}$, the thesis
  holds.

  \item[($\ref{lmm:blncon(ea:par)} \Rightarrow \ref{lmm:blncon(ea:org)}$)] Let
  $\XSet_{\FFun} \in \extFun[]{\AsgFam, \apElm}$ be such that $\asgElm
  \models
  \varphi$ holds for every $\asgElm \in \XSet_{\FFun}$.
  By definition of $\extFun[]{\AsgFam, \apElm}$, there are $\XSet \in
  \AsgFam$ and $\FFun \in \FncSet(\ap{\AsgFam})$ such that $\XSet_{\FFun} = \{
  {\asgElm}[\apElm
  \mapsto \FFun(\asgElm)] : \asgElm \in \XSet \}$.
  Clearly, by semantics (Def.~\ref{def:tarsem}, item
  \ref{def:tarsem(qnt:exs)}), $\asgElm \models \exists \apElm . \varphi$
  holds for
  every $\asgElm \in \XSet$, hence the thesis.

  \item[($\ref{lmm:blncon(ae:org)} \Leftrightarrow \ref{lmm:blncon(ae:par)}$)]
By
  statement~\ref{lmm:flp(ea)} of this lemma, we have that
  $\ref{lmm:blncon(ea:org)}$ is false if and only if $\ref{lmm:blncon(ea:par)}$
  is false
  ($\mathit{not} \ \ref{lmm:blncon(ea:org)} \Leftrightarrow \mathit{not} \
  \ref{lmm:blncon(ea:par)}$, for short).
  By instantiating, in this last equivalence, $\varphi$ with
  $\neg\varphi$, we have $\ref{lmm:blncon(ea:org)}' \Leftrightarrow
  \ref{lmm:blncon(ea:par)}'$, where $\ref{lmm:blncon(ea:org)}'$ and
  $\ref{lmm:blncon(ea:par)}'$ are abbreviations for, respectively:

  \begin{itemize}
  \item for all sets of assignments $\XSet \in \AsgFam$, there exists an
    assignment $\asgElm \in \XSet$ such that $\asgElm \not\models \exists
  \apElm
    . \neg \varphi$;
  \item for all sets of assignments $\XSet \in \extFun[]{\AsgFam,
      \apElm}$, there exists an assignment $\asgElm \in \XSet$ such that
  $\asgElm
    \not\models \neg \varphi$.
  \end{itemize}

  By applying semantics of negation and duality of $\exists$ and
  $\forall$, it is straightforward to see that $\ref{lmm:blncon(ea:org)}'$
  and
  $\ref{lmm:blncon(ea:par)}'$ correspond to $\ref{lmm:blncon(ae:org)}$ and
  $\ref{lmm:blncon(ae:par)}$, respectively, hence the thesis. \qedhere
  \end{itemize}

\end{proof}

\recenv{ThmSemAdqI}
\begin{proof}

  Both claims 1 and 2 are proved together, by induction on the structure of
  the formula.
  \begin{itemize}
  \item[(base case)] If $\varphi \in \LTL$, then the
  claims immediately follows from the semantics
  (Definition~\ref{def:althodsem},
  item~\ref{def:althodsem(ltl)}).

  \item[(inductive step)] If $\varphi = \neg \psi$, then we have, by
  semantics,
  $\AsgFam \cmodels[][\alpha]\: \varphi$ if and only if $\AsgFam
  \notcmodels[][\dual{\alpha}]\: \psi$.
  If $\alpha = \exists\forall$, then, by inductive hypothesis, it is not
  the case that for every $\XSet \in \AsgFam$ there is $\asgElm \in \XSet$
  such
  that $\asgElm \cmodels[] \psi$, which amounts to say that there is
  $\XSet \in \AsgFam$ such that for every $\asgElm \in \XSet$ it holds
  $\asgElm
  \notcmodels[] \psi$, from which the thesis follows.
  If, instead, $\alpha = \forall\exists$, then, by inductive hypothesis,
  there is no $\XSet \in \AsgFam$ such that for every $\asgElm \in \XSet$
  it holds
  $\asgElm \cmodels[] \psi$, which amounts to say that for every
  $\XSet
  \in \AsgFam$ there is $\asgElm \in \XSet$ such that $\asgElm
  \notcmodels[]
  \psi$, from which the thesis follows.

  If $\varphi = \varphi_1 \wedge\varphi_2$ and $\alpha = \exists\forall$,
  then we have, by semantics, $\AsgFam \cmodels[][\alpha]\: \varphi$ if
  and only
  if for every $(\AsgFam[1], \AsgFam[2]) \in \parFun[]{\AsgFam}$
  %
  %
  it holds true that $\AsgFam[1] \neq \emptyset$ and $\AsgFam[1]
  \cmodels[][\alpha] \varphi_{1}$ or it holds true that $\AsgFam[2] \neq
  \emptyset$ and $\AsgFam[2] \cmodels[][\alpha] \varphi_{2}$.
  By inductive hypothesis, this amounts to say that for every
  $(\AsgFam[1], \AsgFam[2]) \in \parFun[]{\AsgFam}$ there is $i \in \{ 1,2
  \}$ and
  $\XSet \in \AsgFam[i]$ such that for every $\asgElm \in \XSet$ it holds
  $\asgElm
  \cmodels[]\varphi_{i}$.
  The thesis follows from Lemma~\ref{lmm:blncon},
  item~\ref{lmm:blncon(ea)}.

  If $\varphi = \varphi_1 \wedge\varphi_2$ and $\alpha = \forall\exists$,
  then we have, by semantics, $\AsgFam \cmodels[][\alpha]\: \varphi$ if
  and only
  if $\dual{\AsgFam} \cmodels[][\dual{\alpha}]\: \varphi$.
  By proceeding as before, i.e., by applying semantics, inductive
  hypothesis, and Lemma~\ref{lmm:blncon}, item~\ref{lmm:blncon(ea)}, we
  have that
  there is $\dual{\XSet} \in \dual{\AsgFam}$ such that for every
  $\dual{\asgElm}
  \in \dual{\XSet}$ it holds $\dual{\asgElm} \cmodels[]\varphi$.
  The thesis follows from Lemma~\ref{lmm:flp},
  item~\ref{lmm:flp(ae)}.

  If $\varphi = \varphi_1 \vee\varphi_2$ and $\alpha = \forall\exists$,
  then we have, by semantics, $\AsgFam \cmodels[][\alpha]\: \varphi$ if
  and only
  if there is $(\AsgFam[1], \AsgFam[2]) \in \parFun[]{\AsgFam}$ such that
  $\AsgFam[1] \neq \emptyset$ implies $\AsgFam[1] \cmodels[][\alpha]
  \varphi_{1}$
  and $\AsgFam[2] \neq \emptyset$ implies $\AsgFam[2] \cmodels[][\alpha]
  \varphi_{2}$.
  By inductive hypothesis, this amounts to say that there is $(\AsgFam[1],
  \AsgFam[2]) \in \parFun[]{\AsgFam}$ such that for every $i \in \{ 1,2
  \}$ and
  $\XSet \in \AsgFam[i]$ there is $\asgElm \in \XSet$ for which it holds
  $\asgElm
  \cmodels[]\varphi_{i}$.
  The thesis follows from Lemma~\ref{lmm:blncon},
  item~\ref{lmm:blncon(ae)}.

  If $\varphi = \varphi_1 \vee\varphi_2$ and $\alpha = \exists\forall$,
  then we have, by semantics, $\AsgFam \cmodels[][\alpha]\: \varphi$ if
  and only
  if $\dual{\AsgFam} \cmodels[][\dual{\alpha}]\: \varphi$.
  By proceeding as before, i.e., by applying semantics, inductive
  hypothesis, and Lemma~\ref{lmm:blncon}, item~\ref{lmm:blncon(ae)}, we
  have that
  for every $\dual{\XSet} \in \dual{\AsgFam}$ there is $\dual{\asgElm}
  \in \dual{\XSet}$ such that $\dual{\asgElm} \cmodels[]\varphi$.
  The thesis follows from Lemma~\ref{lmm:flp},
  item~\ref{lmm:flp(ea)}.

  If $\varphi = \exists \apElm . \psi$ and $\alpha = \exists\forall$, then
  we have, by semantics, $\AsgFam \cmodels[][\alpha]\: \varphi$ if and
  only if
  $\extFun[]{\AsgFam, \apElm} \cmodels[][\alpha] \psi$.
  By inductive hypothesis, this amounts to say that there is $\XSet \in
  \extFun[]{\AsgFam, \apElm}$ such that for every $\asgElm \in \XSet$ it
  holds
  $\asgElm \cmodels[]\psi$.
  The thesis follows from Lemma~\ref{lmm:hypext},
  item~\ref{lmm:hypext(ea)}.

  If $\varphi = \exists \apElm . \psi$ and $\alpha = \forall\exists$,
  then we have, by semantics, $\AsgFam \cmodels[][\alpha]\: \varphi$ if
  and only
  if $\dual{\AsgFam} \cmodels[][\dual{\alpha}]\: \varphi$.
  By proceeding as before, i.e., by applying semantics, inductive
  hypothesis, and Lemma~\ref{lmm:hypext}, item~\ref{lmm:hypext(ea)}, we
  have that
  there is $\dual{\XSet} \in \dual{\AsgFam}$ such that for every
  $\dual{\asgElm}
  \in \dual{\XSet}$ it holds $\dual{\asgElm} \cmodels[]\varphi$.
  The thesis follows from Lemma~\ref{lmm:flp},
  item~\ref{lmm:flp(ae)}.

  If $\varphi = \forall \apElm . \psi$ and $\alpha = \forall\exists$, then
  we have, by semantics, $\AsgFam \cmodels[][\alpha]\: \varphi$ if and
  only if
  $\extFun[]{\AsgFam, \apElm} \cmodels[][\alpha] \psi$.
  By inductive hypothesis, this amounts to say that for every $\XSet \in
  \extFun[]{\AsgFam, \apElm}$ there is $\asgElm \in \XSet$ such that
  $\asgElm \cmodels[]\psi$.
  The thesis follows from Lemma~\ref{lmm:hypext},
  item~\ref{lmm:hypext(ae)}.

  If $\varphi = \forall \apElm . \psi$ and $\alpha = \exists\forall$,
  then we have, by semantics, $\AsgFam \cmodels[][\alpha]\: \varphi$ if
  and only
  if $\dual{\AsgFam} \cmodels[][\dual{\alpha}]\: \varphi$.
  By proceeding as before, i.e., by applying semantics, inductive
  hypothesis, and Lemma~\ref{lmm:hypext}, item~\ref{lmm:hypext(ae)}, we
  have that for every $\dual{\XSet} \in \dual{\AsgFam}$ there is
  $\dual{\asgElm}
  \in \dual{\XSet}$ such that $\dual{\asgElm} \cmodels[]\varphi$.
  The thesis follows from Lemma~\ref{lmm:flp},
  item~\ref{lmm:flp(ea)}. \qedhere
  \end{itemize}
\end{proof}





\subsection{Proofs for Section~\ref{sec:behdep}}
\label{app:behdep}

\recenv{PrpAsgInd}
\begin{proof}
  Assume $\asgElm[1] \approx_{\spcElm}^{k} \asgElm[2]$, i.e., $\asgElm[1] =
  \asgElm[][(1)] \sim_{\spcElm}^{k} \asgElm[][(2)] \sim_{\spcElm}^{k} \ldots
  \sim_{\spcElm}^{k} \asgElm[][(r)] = \asgElm[2]$, for some $\asgElm[][(1)],
  \ldots, \asgElm[][(r)]$, with $r \in \mathbb{N} \setminus \{ 0 \}$ (observe that
  $\asgElm[1] = \asgElm[2]$ if $r=1$).

  We prove, by induction on $r$, that
  items~\ref{prp:asgind(var)}--\ref{prp:asgind(pivii)} hold.
  If $r=1$, then the claim follows trivially.
  Let $r>1$.
  Since $\asgElm[][(1)] \sim_{\spcElm}^{k} \asgElm[][(2)]$, we have that
  \ref{prp:asgind(var)}--\ref{prp:asgind(pivii)} hold when instantiated with
  $\asgElm[][(1)]$ and $\asgElm[][(2)]$, by Definition~\ref{def:asgind}.
  Moreover, by inductive hypothesis,
  \ref{prp:asgind(var)}--\ref{prp:asgind(pivii)} hold when instantiated with
  $\asgElm[][(2)]$ and $\asgElm[][(r)]$.
  The claim follows by transitivity of
  \ref{prp:asgind(var)}--\ref{prp:asgind(pivii)}.

  Now, in order to prove the converse direction, assume that
  items~\ref{prp:asgind(var)}--\ref{prp:asgind(pivii)} hold.
  Let $\set{}{\apElm[1], \ldots, \apElm[r]}$ be an enumeration of $\PSet[\BSym]
  \cup \PSet[\SSym]$ and define $\asgElm[][(1)] \defeq \asgElm[1]$ and
  $\asgElm[][(i+1)] \defeq \asgElm[][(i)][][\apElm[i] \mapsto
  \asgElm[2](\apElm[i])]$ for $i \in [1, \ldots, r]$.
  It is not difficult to convince oneself that $\asgElm[1] = \asgElm[][(1)]
  \sim_{\spcElm}^{k} \asgElm[][(2)] \sim_{\spcElm}^{k} \ldots \sim_{\spcElm}^{k}
  \asgElm[][(r+1)] = \asgElm[2]$ holds, hence $\asgElm[1] \approx_{\spcElm}^{k}
  \asgElm[2]$.
\end{proof}

\recenv{PrpBhvFnc}
\begin{proof}
  Assume $\asgElm[1] \approx_{\spcElm}^{k} \asgElm[2]$, i.e., $\asgElm[1] =
  \asgElm[][(1)] \sim_{\spcElm}^{k} \asgElm[][(2)] \sim_{\spcElm}^{k} \ldots
  \sim_{\spcElm}^{k} \asgElm[][(r)] = \asgElm[2]$, for some $\asgElm[][(1)],
  \ldots, \asgElm[][(r)]$, with $r \in \mathbb{N} \setminus \{ 0 \}$ (observe that
  $\asgElm[1] = \asgElm[2]$ if $r=1$).

  We prove, by induction on $r$, that $\FFun(\asgElm[1])(k) =
  \FFun(\asgElm[2])(k)$.
  If $r=1$, then the claim follows trivially.
  Let $r>1$.
  Since $\asgElm[][(1)] \sim_{\spcElm}^{k} \asgElm[][(2)]$ and $\FFun \in
  \FncSet[\spcElm](\PSet)$, we have that $\FFun(\asgElm[][(1)])(k) =
  \FFun(\asgElm[][(2)])(k)$.
  Moreover, by inductive hypothesis, $\FFun(\asgElm[][(2)])(k) =
  \FFun(\asgElm[][(r)])(k)$.
  The claim follows by transitivity.
\end{proof}

\recenv{PrpOprMon}
\begin{proof}
  Proof of point $1)$. Assume $\AsgFam[1] \!\sqsubseteq\! \AsgFam[2]\!$ and let $\dual{\XSet[2]} \in
  \dual{\AsgFam[2]}$.
  We have to show that there exists $\dual{\XSet[1]} \in \dual{\AsgFam[1]}$ such
  that $\dual{\XSet[1]} \subseteq \dual{\XSet[2]}$.
  By $\AsgFam[1] \sqsubseteq \AsgFam[2]$, there is a function $f :
  \AsgFam[1] \rightarrow \AsgFam[2]$, such that $f(\XSet[1]) \subseteq \XSet[1]$.
  By definition of $\dual{\AsgFam[2]}$, we have that $\dual{\XSet[2]} =
  \img{\chcFun[2]}$ for some $\chcFun[2] \in \ChcSet{\AsgFam[2]}$.

  Now, define $\chcFun[1]$ as $\chcFun[1](\XSet[1]) \defeq
  \chcFun[2](f(\XSet[1]))$ for every $\XSet[1] \in \AsgFam[1]$.
  Clearly, $\chcFun[1] \in \ChcSet{\AsgFam[1]}$, as $\chcFun[1](\XSet[1]) =
  \chcFun[2](f(\XSet[1])) \in f(\XSet[1]) \subseteq \XSet[1]$, and thus
  $\img{\chcFun[1]} \in \dual{\AsgFam[1]}$.
  The thesis follows from the fact that $\img{\chcFun[1]} \subseteq
  \img{\chcFun[2]} = \dual{\XSet[2]}$.

  Proof of point $2)$. Assume $\AsgFam[1] \!\sqsubseteq\! \AsgFam[2]\!$ and let $\XSet[1]['] \in
  \extFun[\spcElm]{\AsgFam[1], \apElm}$.
  We have to show that there exists $\XSet[2]['] \in
  \extFun[\spcElm]{\AsgFam[2], \apElm}$ such that $\XSet[2]['] \subseteq \XSet[1][']$.
  By definition of $\extFun[\spcElm]{\AsgFam[1], \apElm}$, we have that
  $\XSet[1]['] = \extFun{\XSet[1], \FFun, \apElm}$ for some $\XSet[1] \in
  \AsgFam[1]$ and $\FFun \in \FncSet[\spcElm](\ap{\AsgFam[1]})$.
  By $\AsgFam[1] \sqsubseteq \AsgFam[2]$ and $\XSet[1] \in \AsgFam[1]$, we have
  that there is $\XSet[2] \in \AsgFam[2]$ such that $\XSet[2] \subseteq \XSet[1]$.

  It clearly holds that $\extFun{\XSet[2], \FFun, \apElm} \subseteq
  \extFun{\XSet[1], \FFun, \apElm}$.
  The thesis follows, since $\extFun{\XSet[2], \FFun, \apElm} \in
  \extFun[\spcElm]{\AsgFam[2], \apElm}$.
\end{proof}

\recenv{ThmHypRef}
\begin{proof}
  Assume $\AsgFam[1] \!\sqsubseteq\! \AsgFam[2]$.
  Thus, there is a function $f : \AsgFam[1] \rightarrow \AsgFam[2]$, such that
  $f(\XSet[1]) \subseteq \XSet[1]$.
  The claim is proved by induction on the structure of the formula and the
  alternation flags.
  \begin{itemize}
  \item[(base case)] If $\varphi \in \LTL$, then the
  claims immediately follows from the semantics
  (Definition~\ref{def:althodsem},
  item~\ref{def:althodsem(ltl)}).

\item[(inductive step)]
  \begin{itemize}
  \item ($\varphi = \neg \psi$) We have, by semantics, $\AsgFam[1]
    \cmodels[][\exists\forall]\: \varphi$ if and only if $\AsgFam[1]
    \notcmodels[][\forall\exists]\: \psi$.
    By inductive hypothesis, this implies $\AsgFam[2]
    \notcmodels[][\forall\exists]\: \psi$, which amounts to $\AsgFam[2]
    \cmodels[][\exists\forall]\: \varphi$.

    On the other hand, we also have, by semantics, $\AsgFam[2]
    \cmodels[][\forall\exists]\: \varphi$ if and only if $\AsgFam[2]
    \notcmodels[][\exists\forall]\: \psi$.
    By inductive hypothesis, this implies $\AsgFam[1]
    \notcmodels[][\exists\forall]\: \psi$, which amounts to $\AsgFam[1]
    \cmodels[][\forall\exists]\: \varphi$.

  \item ($\varphi = \varphi_1 \wedge\varphi_2$)
    We have, by semantics, $\AsgFam[1] \cmodels[][\exists\forall]\:
    \varphi$ if and only if for every $(\AsgFam[1]['], \AsgFam[1]['']) \in
    \parFun[]{\AsgFam[1]}$
  %
  %
    it holds true that $\AsgFam[1]['] \neq \emptyset$ and $\AsgFam[1][']
    \cmodels[][\exists\forall] \varphi_{1}$ or it holds true that $\AsgFam[1]['']
    \neq \emptyset$ and $\AsgFam[1][''] \cmodels[][\exists\forall] \varphi_{2}$.
    Now, consider $(\AsgFam[2]['], \AsgFam[2]['']) \in \parFun[]{\AsgFam[2]}$
    and $(\AsgFam[1]['], \AsgFam[1]['']) \in \parFun[]{\AsgFam[1]}$, where
    $\AsgFam[1]['] \defeq \set{\XSet \in \AsgFam[1]}{f(\XSet) \in \AsgFam[2][']}$
    and $\AsgFam[1][''] \defeq \set{\XSet \in \AsgFam[1]}{f(\XSet) \in
      \AsgFam[2]['']}$.
    We have that both $\AsgFam[1]['] \sqsubseteq \AsgFam[2][']$ and
    $\AsgFam[1][''] \sqsubseteq \AsgFam[2]['']$ hold.
    Thus, by inductive hypothesis, for every $(\AsgFam[2]['], \AsgFam[2][''])
    \in \parFun[]{\AsgFam[2]}$ it holds true that $\AsgFam[2]['] \neq \emptyset$ and
    $\AsgFam[2]['] \cmodels[][\exists\forall] \varphi_{1}$ or it holds true that
    $\AsgFam[2][''] \neq \emptyset$ and $\AsgFam[2][''] \cmodels[][\exists\forall]
    \varphi_{2}$, which amounts to $\AsgFam[2] \cmodels[][\exists\forall]\:
    \varphi$.

    On the other hand, we also have, by semantics, $\AsgFam[2]
    \cmodels[][\forall\exists]\: \varphi$ if and only if $\dual{\AsgFam[2]}
    \cmodels[][\exists\forall]\: \varphi$.
    By inductive hypothesis and Proposition~\ref{prp:oprmon}, this implies
    $\dual{\AsgFam[1]} \cmodels[][\exists\forall]\: \varphi$, which amounts to
    $\AsgFam[1] \cmodels[][\forall\exists]\: \varphi$.

  \item ($\varphi = \varphi_1 \vee\varphi_2$)
    We have, by semantics, $\AsgFam[2] \cmodels[][\forall\exists]\: \varphi$ if
    and only if there is $(\AsgFam[2]['], \AsgFam[2]['']) \in \parFun[]{\AsgFam[2]}$
    such that $\AsgFam[2]['] \neq \emptyset$ implies $\AsgFam[2][']
    \cmodels[][\forall\exists] \varphi_{1}$ and $\AsgFam[2][''] \neq \emptyset$
    implies $\AsgFam[2][''] \cmodels[][\forall\exists] \varphi_{2}$.
    Now, let $(\AsgFam[1]['], \AsgFam[1]['']) \in \parFun[]{\AsgFam[1]}$, where
    $\AsgFam[1]['] \defeq \set{\XSet \in \AsgFam[1]}{f(\XSet) \in \AsgFam[2][']}$
    and $\AsgFam[1][''] \defeq \set{\XSet \in \AsgFam[1]}{f(\XSet) \in
      \AsgFam[2]['']}$.
    We have that both $\AsgFam[1]['] \sqsubseteq \AsgFam[2][']$ and
    $\AsgFam[1][''] \sqsubseteq \AsgFam[2]['']$ hold.
    Moreover, we have that $\AsgFam[2]['] = \emptyset$ implies $\AsgFam[1]['] =
    \emptyset$ and $\AsgFam[2][''] = \emptyset$ implies $\AsgFam[1][''] =
    \emptyset$.
    In addition, by inductive hypothesis, $\AsgFam[2][']
    \cmodels[][\forall\exists] \varphi_{1}$ implies $\AsgFam[1][']
    \cmodels[][\forall\exists] \varphi_{1}$ and $\AsgFam[2]['']
    \cmodels[][\forall\exists] \varphi_{2}$ implies $\AsgFam[1]['']
    \cmodels[][\forall\exists] \varphi_{2}$.
    Therefore, $\AsgFam[1] \cmodels[][\forall\exists]\: \varphi$ holds.

    On the other hand, we also have, by semantics, $\AsgFam[1]
    \cmodels[][\exists\forall]\: \varphi$ if and only if $\dual{\AsgFam[1]}
    \cmodels[][\forall\exists]\: \varphi$.
    By inductive hypothesis and Proposition~\ref{prp:oprmon}, this implies
    $\dual{\AsgFam[2]} \cmodels[][\forall\exists]\: \varphi$, which amounts to
    $\AsgFam[2] \cmodels[][\exists\forall]\: \varphi$.

  \item ($\varphi = \exists \apElm \colon\! \spcElm \ldotp \psi$)
    We have, by semantics, $\AsgFam[1] \cmodels[][\exists\forall]\: \varphi$ if
    and only if $\extFun[\spcElm]{\AsgFam[1], \apElm} \cmodels[][\exists\forall]
    \psi$.
    By inductive hypothesis and Proposition~\ref{prp:oprmon}, this implies
    $\extFun[\spcElm]{\AsgFam[2], \apElm} \cmodels[][\exists\forall] \psi$, which
    amounts to $\AsgFam[2] \cmodels[][\exists\forall]\: \varphi$.

    On the other hand, we also have, by semantics, $\AsgFam[2]
    \cmodels[][\forall\exists]\: \varphi$ if and only if $\dual{\AsgFam[2]}
    \cmodels[][\exists\forall]\: \varphi$.
    By inductive hypothesis and Proposition~\ref{prp:oprmon}, this implies
    $\dual{\AsgFam[1]} \cmodels[][\exists\forall]\: \varphi$, which amounts to
    $\AsgFam[1] \cmodels[][\forall\exists]\: \varphi$.

  \item ($\varphi = \forall \apElm \colon\! \spcElm \ldotp \psi$)
    We have, by semantics, $\AsgFam[2] \cmodels[][\forall\exists]\: \varphi$ if
    and only if $\extFun[\spcElm]{\AsgFam[2], \apElm} \cmodels[][\forall\exists] \psi$.
    By inductive hypothesis and Proposition~\ref{prp:oprmon}, this implies
    $\extFun[\spcElm]{\AsgFam[1], \apElm} \cmodels[][\forall\exists] \psi$, which
    amounts to $\AsgFam[1] \cmodels[][\forall\exists]\: \varphi$.

    On the other hand, we also have, by semantics, $\AsgFam[1]
    \cmodels[][\exists\forall]\: \varphi$ if and only if $\dual{\AsgFam[1]}
    \cmodels[][\forall\exists]\: \varphi$.
    By inductive hypothesis and Proposition~\ref{prp:oprmon}, this implies
    $\dual{\AsgFam[2]} \cmodels[][\forall\exists]\: \varphi$, which amounts to
    $\AsgFam[2] \cmodels[][\exists\forall]\: \varphi$.\qedhere
  \end{itemize}
\end{itemize}
\end{proof}

\recenv{ThmDlt}
\begin{proof}
  The fact that $\AsgFam \cmodels[][\alpha] \varphi$ \iff $\dual{\dual{\AsgFam}}
  \cmodels[][\alpha] \varphi$ immediately follows from $\AsgFam \equiv
  \dual{\dual{\AsgFam}}$ (Proposition~\ref{prp:dltinv}) and
  Corollary~\ref{cor:eqv}.

  We turn now on proving that $\AsgFam \cmodels[][\alpha] \varphi$ \iff
  $\dual{\AsgFam} \cmodels[][\dual{\alpha}] \varphi$, for all $\AsgFam \in
  \HAsgSet[\subseteq](\free{\varphi})$.
  The proof is by induction on the structure of the formula.
  \begin{itemize}
  \item[(base case)]
    If $\varphi \in \LTL$, then the claim follows immediately from the
    semantics and Lemma~\ref{lmm:flp}.

  \item[(inductive step)]
    If $\varphi = \neg \psi$, then we have: $\AsgFam \cmodels[][\alpha]\:
    \varphi \iffExpl{sem.}
    \AsgFam {\notcmodels[][\dual{\alpha}]}\: \psi \iffExpl{ind.hp.}
    \flipFun{\AsgFam} \notcmodels[][\alpha]\: \psi \iffExpl{sem.}
    \flipFun{\AsgFam} \cmodels[][\dual{\alpha}]\: \varphi$.

    If $\varphi = \varphi_1 \wedge \varphi_2$, then we have:
    \begin{itemize}
    \item $\AsgFam \cmodels[][\exists\forall]\: \varphi
      \iffExpl{Thm.~\ref{thm:dlt}}
\flipFun{\flipFun{\AsgFam}}
      \cmodels[][\exists\forall]\: \varphi \iffExpl{sem.}
      \flipFun{\AsgFam} \cmodels[][\forall\exists]\: \varphi$; and

    \item $\AsgFam \cmodels[][\forall\exists]\: \varphi \iffExpl{sem.}
      \flipFun{\AsgFam} \cmodels[][\exists\forall]\: \varphi$.

    \end{itemize}

    If $\varphi = \varphi_1 \vee \varphi_2$, then we have:
    \begin{itemize}
    \item $\AsgFam \cmodels[][\exists\forall]\: \varphi \iffExpl{sem.}
      \flipFun{\AsgFam} \cmodels[][\forall\exists]\: \varphi$; and

    \item $\AsgFam \cmodels[][\forall\exists]\: \varphi
      \iffExpl{Thm.~\ref{thm:dlt}}
\flipFun{\flipFun{\AsgFam}}
      \cmodels[][\forall\exists]\: \varphi \iffExpl{sem.}
      \flipFun{\AsgFam} \cmodels[][\exists\forall]\: \varphi$.

    \end{itemize}

    If $\varphi = \exists \apElm \colon\! \spcElm \ldotp \psi$, then we have:
    \begin{itemize}
    \item $ \AsgFam \cmodels[][\exists\forall]\: \varphi
      \iffExpl{Thm.~\ref{thm:dlt}}
\flipFun{\flipFun{\AsgFam}}
      \cmodels[][\exists\forall]\: \varphi \iffExpl{sem.}
      \flipFun{\AsgFam} \cmodels[][\forall\exists]\: \varphi$; and

    \item $ \AsgFam \cmodels[][\forall\exists]\: \varphi \iffExpl{sem.}
\flipFun{\AsgFam}
      \cmodels[][\exists\forall]\: \varphi$.
    \end{itemize}

    If $\varphi = \forall \apElm \colon\! \spcElm \ldotp \psi$, then we have:
    \begin{itemize}
    \item $ \AsgFam \cmodels[][\exists\forall]\: \varphi \iffExpl{sem.}
\flipFun{\AsgFam}
      \cmodels[][\forall\exists]\: \varphi$;

    \item $ \AsgFam \cmodels[][\forall\exists]\: \varphi
      \iffExpl{Thm.~\ref{thm:dlt}}
\flipFun{\flipFun{\AsgFam}}
      \cmodels[][\forall\exists]\: \varphi \iffExpl{sem.}
      \flipFun{\AsgFam} \cmodels[][\exists\forall]\: \varphi$. \qedhere
    \end{itemize}
  \end{itemize}

\end{proof}

\recenv{LmmDML}
\begin{proof}
  Thanks to Corollary~\ref{cor:intinv}, it suffices to prove the
  equivalence for $\equiv^{\alpha}$ for some $\alpha \in \{ \exists\forall,
  \forall\exists \}$.
  Let $\varphi$ be a \QPTL formula and $\AsgFam \in
  \HAsgSet(\free{\varphi})$ a hyperassignment.
  \begin{enumerate}
  \item[\ref{lmm:dml(neg)})]
    $\AsgFam \cmodels[][\exists\forall]\: \neg \neg
    \varphi \iffExpl{sem.}  \AsgFam
    \notcmodels[][\forall\exists] \neg \varphi \iffExpl{sem.}
    \AsgFam \cmodels[][\exists\forall]\: \varphi$;
  \item[\ref{lmm:dml(crem)}-\ref{lmm:dml(dass)})]
    Trivial and omitted.
  \item[\ref{lmm:dml(con)})]
     $ \AsgFam
    \cmodels[][\exists\forall]\: \varphi_1 \wedge \varphi_2 \iffExpl{sem.}
    \forall (\AsgFam[1], \AsgFam[2]) \in \parFun[]{\AsgFam}$ it holds
$(\AsgFam[1] \neq \emptyset$ and $\AsgFam[1] \cmodels[][\exists\forall]
\varphi_{1})$ \ or $(\AsgFam[2] \neq \emptyset$ and $
    \AsgFam[2] \cmodels[][\exists\forall] \varphi_{2})
    \iffExpl{1.,sem.}
    \forall (\AsgFam[1], \AsgFam[2]) \in \parFun[]{\AsgFam}$ it holds $
(\AsgFam[1] \neq \emptyset$ and $\AsgFam[1] \notcmodels[][\forall\exists] \neg
\varphi_{1})$ \
    or $(\AsgFam[2] \neq \emptyset$ and $ \AsgFam[2]
\notcmodels[][\forall\exists] \neg \varphi_{2}) \iffExpl{sem.}
    \AsgFam \notcmodels[][\forall\exists]\: \neg \varphi_1 \vee \neg
    \varphi_2 \iffExpl{sem.}  \AsgFam \cmodels[][\exists\forall]\: \neg(\neg
    \varphi_1 \vee \neg \varphi_2)$;


  \item[\ref{lmm:dml(dis)})]
    $ \AsgFam \cmodels[][\forall\exists]\:
    \varphi_1 \vee \varphi_2 \iffExpl{sem.}
    \exists (\AsgFam[1], \AsgFam[2]) \in \parFun[]{\AsgFam}$ s.t.
    $(\AsgFam[1] \neq \emptyset$ implies $\AsgFam[1] \cmodels[][\forall\exists]
    \varphi_{1})$ \ and $(\AsgFam[2] \neq \emptyset$ implies $\AsgFam[2]
    \cmodels[][\forall\exists] \varphi_{2}) \iffExpl{1.,sem.}
    \exists (\AsgFam[1], \AsgFam[2])  \in \parFun[]{\AsgFam}$ s.t.
    $(\AsgFam[1] \neq \emptyset$ implies $\AsgFam[1]
    \notcmodels[][\exists\forall] \neg \varphi_{1})$ \ and $(\AsgFam[2] \neq
    \emptyset$ implies $ \AsgFam[2] \notcmodels[][\exists\forall] \neg \varphi_{2})
    \iffExpl{sem.}
    \AsgFam \notcmodels[][\exists\forall]\: \neg \varphi_1 \wedge \neg
    \varphi_2 \iffExpl{sem.}  \AsgFam \cmodels[][\forall\exists]\: \neg(\neg
    \varphi_1 \wedge \neg \varphi_2)$;


  \item[\ref{lmm:dml(exs)})]
  $ \AsgFam
  \cmodels[][\exists\forall]\: \exists \apElm \colon\! \spcElm \ldotp \psi \iffExpl{sem.}
    \extFun[\spcElm]{\AsgFam, \apElm}
  \cmodels[][\exists\forall]\: \psi
  \iffExpl{1.,sem.}
    \extFun[\spcElm]{\AsgFam, \apElm}
  \notcmodels[][\forall\exists]\: \neg \psi
  \iffExpl{sem.}
    \AsgFam \notcmodels[][\forall\exists]\: \forall
  \apElm \colon\! \spcElm \ldotp \neg \psi
  \iffExpl{sem.}
    \AsgFam \cmodels[][\exists\forall]\: \neg \forall \apElm
  \colon\! \spcElm \ldotp \neg \psi$;


  \item[\ref{lmm:dml(all)})]
  $\AsgFam
  \cmodels[][\forall\exists]\: \forall \apElm \colon\! \spcElm \ldotp \psi \iffExpl{sem.}
    \extFun[\spcElm]{\AsgFam, \apElm}
  \cmodels[][\forall\exists]\: \psi
  \iffExpl{1.,sem.}
    \extFun[\spcElm]{\AsgFam, \apElm}
  \notcmodels[][\exists\forall]\: \neg \psi
  \iffExpl{sem.}
    \AsgFam \notcmodels[][\exists\forall]\: \exists
  \apElm \colon\! \spcElm \ldotp \neg \psi
  \iffExpl{sem.}
  \AsgFam \cmodels[][\forall\exists]\: \neg \exists \apElm \colon\!
  \spcElm \ldotp \neg \psi$.
%
%
  \qedhere
  \end{enumerate}

\end{proof}

\recenv{PrpEvlMon}
\begin{proof}
  The proof proceeds by induction on the length of the quantification prefix
  \qntElm.

  If $\qntElm = \varepsilon$ (base case), then we have
  $\evlFun[\alpha](\AsgFam[1], \qntElm) = \AsgFam[1] \sqsubseteq \AsgFam[2] =
  \evlFun[\alpha](\AsgFam[2], \qntElm)$.

  If $\qntElm = \Qnt^{\spcElm} \apElm . \qntElm'$ (inductive step), then we
  distinguish two cases.
  \begin{itemize}
  \item If $\alpha$ and $\Qnt$ are coherent, then we have
    $\evlFun[\alpha](\AsgFam[1], \qntElm) =
    \evlFun[\alpha](\extFun[\spcElm]{\AsgFam[1], \apElm}, \qntElm')$ and
    $\evlFun[\alpha](\AsgFam[2], \qntElm) =
    \evlFun[\alpha](\extFun[\spcElm]{\AsgFam[2], \apElm}, \qntElm')$.
    By Proposition~\ref{prp:oprmon}, $\AsgFam[1] \sqsubseteq \AsgFam[2]$ implies
    $\extFun[\spcElm]{\AsgFam[1], \apElm} \sqsubseteq \extFun[\spcElm]{\AsgFam[2],
      \apElm}$ and, by inductive hypothesis,
    $\evlFun[\alpha](\extFun[\spcElm]{\AsgFam[1], \apElm}, \qntElm') \sqsubseteq
    \evlFun[\alpha](\extFun[\spcElm]{\AsgFam[2], \apElm}, \qntElm')$, hence the
    thesis.

  \item If $\alpha$ and $\Qnt$ are not coherent, then we have
    $\evlFun[\alpha](\AsgFam[1], \qntElm) =
    \evlFun[\alpha](\dual{\extFun[\spcElm]{\dual{\AsgFam[1]}, \apElm}}, \qntElm')$
    and $\evlFun[\alpha](\AsgFam[2], \qntElm) =
    \evlFun[\alpha](\dual{\extFun[\spcElm]{\dual{\AsgFam[2]}, \apElm}}, \qntElm')$.
    By Proposition~\ref{prp:oprmon}, $\AsgFam[1]
    \sqsubseteq \AsgFam[2]$ implies $\dual{\extFun[\spcElm]{\dual{\AsgFam[1]},
        \apElm}} \sqsubseteq \dual{\extFun[\spcElm]{\dual{\AsgFam[2]}, \apElm}}$, and, by
    inductive hypothesis, $\evlFun[\alpha](\dual{\extFun[\spcElm]{\dual{\AsgFam[1]},
        \apElm}}, \qntElm') \sqsubseteq
    \evlFun[\alpha](\dual{\extFun[\spcElm]{\dual{\AsgFam[2]}, \apElm}}, \qntElm')$,
    hence the thesis. \qedhere
  \end{itemize}
\end{proof}

\recenv{LmmPrfEvl}
\begin{proof}

  \begin{itemize}
  \item[(base case)] When $\qntElm = \varepsilon$ the claim is trivial.

  \item[(inductive step)]

  \begin{itemize}
  \item $\AsgFam \cmodels[][\exists\forall] \exists^{\spcElm} \apElm \ldotp
    \qntElm \ldotp \psi \Leftrightarrow \extFun[\spcElm]{\AsgFam, \apElm}
    \cmodels[][\exists\forall] \qntElm \ldotp \psi \Leftrightarrow
    \evlFun[\exists\forall](\extFun[\spcElm]{\AsgFam, \apElm}, \qntElm)
    \cmodels[][\exists\forall] \psi \Leftrightarrow \evlFun[\exists\forall](\AsgFam,
    \exists^{\spcElm} \apElm \ldotp \qntElm) \cmodels[][\exists\forall] \psi$.

  \item $\AsgFam \cmodels[][\forall\exists] \exists^{\spcElm} \apElm \ldotp
    \qntElm \ldotp \psi \Leftrightarrow \dual{\AsgFam} \cmodels[][\exists\forall]
    \exists^{\spcElm} \apElm \ldotp \qntElm \ldotp \psi \Leftrightarrow
    \extFun[\spcElm]{\dual{\AsgFam}, \apElm} \cmodels[][\exists\forall] \qntElm
    \ldotp \psi \Leftrightarrow \dual{\extFun[\spcElm]{\dual{\AsgFam}, \apElm}}
    \cmodels[][\forall\exists] \qntElm \ldotp \psi \Leftrightarrow
    \evlFun[\forall\exists](\dual{\extFun[\spcElm]{\dual{\AsgFam}, \apElm}}, \qntElm)
    \cmodels[][\forall\exists] \psi \Leftrightarrow
    \evlFun[\forall\exists](\AsgFam, \exists^{\spcElm} \apElm \ldotp \qntElm)
    \cmodels[][\forall\exists] \psi$.

  \item $\AsgFam \cmodels[][\exists\forall] \forall^{\spcElm} \apElm \ldotp
    \qntElm \ldotp \psi \Leftrightarrow \dual{\AsgFam} \cmodels[][\forall\exists]
    \forall^{\spcElm} \apElm \ldotp \qntElm \ldotp \psi \Leftrightarrow
    \extFun[\spcElm]{\dual{\AsgFam}, \apElm} \cmodels[][\forall\exists] \qntElm
    \ldotp \psi \Leftrightarrow \dual{\extFun[\spcElm]{\dual{\AsgFam}, \apElm}}
    \cmodels[][\exists\forall] \qntElm \ldotp \psi \Leftrightarrow
    \evlFun[\exists\forall](\dual{\extFun[\spcElm]{\dual{\AsgFam}, \apElm}}, \qntElm)
    \cmodels[][\exists\forall] \psi \Leftrightarrow \evlFun[\exists\forall](\AsgFam,
    \forall^{\spcElm} \apElm \ldotp \qntElm) \cmodels[][\exists\forall] \psi$.

  \item $\AsgFam \cmodels[][\forall\exists] \forall^{\spcElm} \apElm \ldotp
    \qntElm \ldotp \psi \Leftrightarrow \extFun[\spcElm]{\AsgFam, \apElm}
    \cmodels[][\forall\exists] \qntElm \ldotp \psi \Leftrightarrow
    \evlFun[\forall\exists](\extFun[\spcElm]{\AsgFam, \apElm}, \qntElm)
    \cmodels[][\forall\exists] \psi \Leftrightarrow \evlFun[\forall\exists](\AsgFam,
    \forall^{\spcElm} \apElm \ldotp \qntElm) \cmodels[][\forall\exists] \psi$.
    \qedhere
  \end{itemize}

  \end{itemize}
\end{proof}

We now introduce an alternative to $\evlFun$ that we use to prove theorems presented in this paper. To use this alternative, we prove propositions that have not been expressed in the main article.

\addtocounter{proposition}{1}
\defenv{proposition}[][Prp][Evl]
  {
  Let $\AsgFam \in \HAsgSet(\PSet)$ with $\PSet \!\subseteq\! \APSet$, $\spcElm
  \in \SpcSet$, $\apElm \in \APSet \!\setminus\! \PSet$, and $\Psi \subseteq
  \AsgSet(\PSet \cup \{ \apElm \})$.
  There exists $\WSet \in \evlFun[\alpha](\AsgFam, \Qnt^{\spcElm} \apElm)$ such
  that $\WSet \subseteq \Psi$ \iff the following conditions hold true:
  \begin{enumerate}
    \item\labelx{coh}
      there exist $\FFun \in \FncSet[\spcElm](\PSet)$ and $\XSet \in \AsgFam$
      such that $\extFun{\XSet, \FFun, \apElm} \allowbreak \subseteq \Psi$,
      whenever $\alpha$ and $\Qnt$ are coherent;
    \item\labelx{nch}
      for all $\FFun \in \FncSet[\spcElm](\PSet)$, there is $\XSet \in \AsgFam$
      such that $\extFun{\XSet, \FFun, \apElm} \allowbreak \subseteq \Psi$,
      whenever $\alpha$ and $\Qnt$ are not coherent.
  \end{enumerate}
  }
\begin{proof}

  If $\alpha$ and $\Qnt$ are coherent, then $ \evlFun[\alpha](\AsgFam,
  \Qnt^{\spcElm} \apElm) = \extFun[\spcElm]{\AsgFam, \apElm} = {\set{
\extFun{\XSet,
        \FFun, \apElm} }{ \XSet \in \AsgFam, \FFun \in
\FncSet[\spcElm](\ap{\AsgFam}) }}
  $, and we have: $\exists \WSet \in \evlFun[\alpha](\AsgFam, \Qnt^{\spcElm}
  \apElm) \ . \ \WSet \subseteq \Psi \iffExpl{}\exists \FFun \in
  \FncSet[\spcElm](\PSet) \ \exists \XSet \in \AsgFam \ . \ \extFun{\XSet, \FFun,
    \apElm} \subseteq \Psi$.

  If $\alpha$ and $\Qnt$ are not coherent, then we proceed as follows: $\exists
  \WSet \in \evlFun[\alpha](\AsgFam, \Qnt^{\spcElm} \apElm) \ . \ \WSet \subseteq
  \Psi \iffExpl{} \exists \WSet \in \dual{\extFun[\spcElm]{\dual{\AsgFam},
\apElm}}
  \ . \ \WSet \subseteq \Psi \iffExpl{} \exists \chcFun \in
  \ChcSet{\extFun[\spcElm]{\dual{\AsgFam}, \apElm}} \ . \ \img{\chcFun} =
  \set{\chcFun(\ZSet)}{\ZSet \in \extFun[\spcElm]{\dual{\AsgFam}, \apElm}}
  \subseteq \Psi \iffExpl{} \exists \chcFun \in
  \ChcSet{\extFun[\spcElm]{\dual{\AsgFam}, \apElm}} \ . \ \forall \ZSet \in
  \extFun[\spcElm]{\dual{\AsgFam}, \apElm} \ . \ \chcFun(\ZSet) \in \Psi
\iffExpl{}
  \forall \ZSet \in \extFun[\spcElm]{\dual{\AsgFam}, \apElm} \ . \ \exists
  \asgElm[\ZSet] \in \ZSet \ . \ \asgElm[\ZSet] \in \Psi \iffExpl{} \forall
\GFun
  \in \FncSet[\spcElm](\PSet) \ . \ \forall \YSet \in \dual{\AsgFam} \ . \
\exists
  \asgElm[\GFun,\YSet] \in \extFun{\YSet, \GFun, \apElm} \ . \
  \asgElm[\GFun,\YSet] \in \Psi \iffExpl{} \forall \GFun \in
  \FncSet[\spcElm](\PSet) \ . \ \forall \Lambda \in \ChcSet{\AsgFam} \ . \
\exists
  \asgElm[\GFun,\Lambda] \in \extFun{\img{\Lambda}, \GFun, \apElm} \ . \
  \asgElm[\GFun,\Lambda] \in \Psi \iffExpl{} \forall \GFun \in
  \FncSet[\spcElm](\PSet) \ . \ \forall \Lambda \in \ChcSet{\AsgFam} \ . \
\exists
  \asgElm[\GFun,\Lambda] \in \extFun{\set{\Lambda(\XSet)}{\XSet \in \AsgFam},
    \GFun, \apElm} \ . \ \asgElm[\GFun,\Lambda] \in \Psi \iffExpl{} \forall
\GFun
  \in \FncSet[\spcElm](\PSet) \ . \ \forall \Lambda \in \ChcSet{\AsgFam} \ . \
  \exists \XSet \in \AsgFam \ . \ \extFun{\Lambda(\XSet), \GFun, \apElm} \in
\Psi
  \iffExpl{} \forall \GFun \in \FncSet[\spcElm](\PSet) \ . \ \exists \XSet \in
  \AsgFam \ . \ \forall \asgElm \in \XSet \ . \ \extFun{\asgElm, \GFun, \apElm}
  \in \Psi \iffExpl{} \forall \GFun \in \FncSet[\spcElm](\PSet) \ . \ \exists
\XSet
  \in \AsgFam \ . \ \extFun{\XSet, \GFun, \apElm} \subseteq \Psi$.
\end{proof}

We now introduce $\nevlFun$. First, we define it only on $\Qnt^{\spcElm} \apElm$.
 $\nevlFun[\alpha](\AsgFam, \Qnt^{\spcElm} \apElm) \defeq
 \extFun[\spcElm]{\AsgFam,
 \apElm}$, if $\alpha$ and $\Qnt$ are coherent, and $\nevlFun[\alpha](\AsgFam,
 \Qnt^{\spcElm} \apElm) \defeq \set{ \extFun{\eth, \apElm} }{ \eth \in
 \FncSet[\spcElm](\ap{\AsgFam}) \to \AsgFam }$, otherwise, where $\extFun{\eth,
 \apElm} \defeq \bigcup_{\FFun \in \dom{\eth}} \extFun{\eth(\FFun), \FFun,
 \apElm}$.

 \defenv{proposition}[][Prp][NrmEvlBas]
   {
   If $\AsgFam[1] \!\equiv\! \AsgFam[2]$ then $\nevlFun[\alpha](\AsgFam[1],
 \Qnt^{\spcElm}
   \apElm) \equiv \evlFun[\alpha](\AsgFam[2], \Qnt^{\spcElm} \apElm)$, for all
   $\AsgFam[1], \AsgFam[2] \in \HAsgSet$, $\Qnt \in \{ \exists, \forall \}$, $\spcElm
   \in \SpcSet$, and $\apElm \in \APSet \setminus \ap{\AsgFam}$.
   }
\begin{proof}
  Assume $\alpha$ and $\Qnt$ to be coherent.
  Then, $\nevlFun[\alpha](\AsgFam[1], \Qnt^{\spcElm} \apElm) =
  \extFun[\spcElm]{\AsgFam[1], \apElm} \equiv \extFun[\spcElm]{\AsgFam[2], \apElm} =
  \evlFun[\alpha](\AsgFam[2], \Qnt^{\spcElm} \apElm)$.

  Assume, now, $\alpha$ and $\Qnt$ not to be coherent.
  Then, $\nevlFun[\alpha](\AsgFam[1], \Qnt^{\spcElm} \apElm) = \set{
    \bigcup_{\FFun \in \FncSet[\spcElm](\ap{\AsgFam[1]})} \extFun{\eth(\FFun), \FFun,
      \apElm} }{ \eth \in \FncSet[\spcElm](\ap{\AsgFam[1]}) \to \AsgFam[1] }$.

  In order to prove that $\nevlFun[\alpha](\AsgFam[1], \Qnt^{\spcElm} \apElm)
  \sqsubseteq \evlFun[\alpha](\AsgFam[2], \Qnt^{\spcElm} \apElm)$, we let $\Psi \in
  \nevlFun[\alpha](\AsgFam[1], \Qnt^{\spcElm} \apElm)$ and we show that there is $W
  \in \evlFun[\alpha](\AsgFam[2], \Qnt^{\spcElm} \apElm)$ with $W \subseteq \Psi$.
  By the definition of $\nevlFun[\alpha](\AsgFam[1], \Qnt^{\spcElm} \apElm)$, we
  have that $\Psi = \bigcup_{\FFun \in \FncSet[\spcElm](\ap{\AsgFam[1]})}
  \extFun{\eth(\FFun), \FFun, \apElm}$ for some $\eth \in
  \FncSet[\spcElm](\ap{\AsgFam[1]}) \to \AsgFam[1]$.
  This means that for every $\FFun \in \FncSet[\spcElm](\ap{\AsgFam[1]})$ there
  is $\XSet_{\FFun} \defeq \eth(\FFun) \in \AsgFam[1]$ such that
  $\extFun{\XSet_{\FFun}, \FFun, \apElm} \subseteq \Psi$.
  By Proposition~\ref{prp:evl}, there exists $\WSet' \in
  \evlFun[\alpha](\AsgFam[1], \Qnt^{\spcElm} \apElm)$ such that $\WSet' \subseteq
  \Psi$.
  Since, due to Proposition~\ref{prp:evlmon}, $\AsgFam[1] \equiv \AsgFam[2]$
  implies $\evlFun[\alpha](\AsgFam[1], \Qnt^{\spcElm} \apElm) \equiv
  \evlFun[\alpha](\AsgFam[2], \Qnt^{\spcElm} \apElm)$, we have that there is $\WSet
  \in \evlFun[\alpha](\AsgFam[2], \Qnt^{\spcElm} \apElm)$ such that $\WSet
  \subseteq \WSet' \subseteq \Psi$.

  In order to prove the converse, \ie, $\evlFun[\alpha](\AsgFam[2],
  \Qnt^{\spcElm} \apElm) \sqsubseteq \nevlFun[\alpha](\AsgFam[1], \Qnt^{\spcElm}
  \apElm)$, let $\Psi \in \evlFun[\alpha](\AsgFam[2], \Qnt^{\spcElm} \apElm)$.
  By instantiating $\WSet$ with $\Psi$ in Proposition~\ref{prp:evl}, we have
  that for all $\FFun \in \FncSet[\spcElm](\ap{\AsgFam[2]})$, there is
  $\XSet'_{\FFun} \in \AsgFam[2]$ such that $\extFun{\XSet'_{\FFun}, \FFun,
    \apElm} \allowbreak \subseteq \Psi$.
  By $\AsgFam[1] \equiv \AsgFam[2]$, we have that there is $\XSet_{\FFun} \in
  \AsgFam[1]$ such that $\XSet_{\FFun} \subseteq \XSet'_{\FFun}$, which, in turn,
  implies $\extFun{\XSet_{\FFun}, \FFun, \apElm} \allowbreak \subseteq
  \extFun{\XSet'_{\FFun}, \FFun, \apElm} \allowbreak \subseteq \Psi$.
  Now, define $\eth$ as $\eth(\FFun) \defeq \XSet_{\FFun} \in \AsgFam[1]$ for
  every $\FFun \in \FncSet[\spcElm](\ap{\AsgFam[1]}) =
  \FncSet[\spcElm](\ap{\AsgFam[2]})$.
  Clearly, both of the following hold: $\bigcup_{\FFun \in
    \FncSet[\spcElm](\ap{\AsgFam[1]})} \extFun{\eth(\FFun), \FFun, \apElm} \in
  \nevlFun[\alpha](\AsgFam[1], \Qnt^{\spcElm} \apElm)$ and $\bigcup_{\FFun \in
    \FncSet[\spcElm](\ap{\AsgFam[1]})} \extFun{\eth(\FFun), \FFun, \apElm} \subseteq
  \Psi$; hence the thesis.
\end{proof}

We extend the definition of $\nevlFun$ to any quantifier prefix.
 \begin{itemize}
   \item
     $\nevlFun[\alpha](\AsgFam, \epsilon) \defeq \AsgFam$;
   \item
     $\nevlFun[\alpha](\AsgFam, \Qnt^{\spcElm} \apElm \ldotp \qntElm) \defeq
     \nevlFun[\alpha](\nevlFun[\alpha](\AsgFam, \Qnt^{\spcElm} \apElm),
 \qntElm)$.
 \end{itemize}

 \[
   \nevlFun[\alpha](\qntElm) \defeq \nevlFun[\alpha](\{ \{ \emptyfun \} \},
   \qntElm)
 \]

 Functions $\evlFun$ and $\nevlFun$ preserve equivalence between
 hyperassignments, as stated in the next proposition.
 Thus, in what follows, we will use the two functions interchangeably, as they
 are only involved in functions and relations that are invariant with respect to
 equivalence between hyperassignments.

 \defenv{proposition}[][Prp][NrmEvlInd]
   {
   If $\AsgFam[1] \!\equiv\! \AsgFam[2]$ then $\nevlFun[\alpha](\AsgFam[1], \qntElm)
   \!\equiv\! \evlFun[\alpha](\AsgFam[2], \qntElm)$, for all $\AsgFam[1], \AsgFam[2] \in
   \HAsgSet$ and $\qntElm \in \QntSet$, with $\ap{\AsgFam} \cap \ap{\qntElm} =
   \emptyset$.
   }
\begin{proof}
  The proof proceeds by induction on the length of the quantification prefix
  \qntElm.

  (base case) If $\qntElm = \varepsilon$, then we have
  $\nevlFun[\alpha](\AsgFam[1], \qntElm) = \AsgFam[1] \equiv \AsgFam[2] = \evlFun[\alpha](\AsgFam[2],
  \qntElm)$.
%

  (inductive step) If $\qntElm = \Qnt^{\spcElm} \apElm . \qntElm'$, then
  $\nevlFun[\alpha](\AsgFam[1], \qntElm) = \nevlFun[\alpha](\nevlFun[\alpha](\AsgFam[1],
  \Qnt^{\spcElm} \apElm), \qntElm')$.
  The thesis follows by a straightforward application of
  Proposition~\ref{prp:nrmevlbas} and the inductive hypothesis.
\end{proof}


 Let $\PSet, \PSet' \subseteq \APSet$, $\asgElm \in \AsgSet(\PSet)$, and $\XSet
 \subseteq \AsgSet(\PSet)$.
 We define $(\asgElm \setminus \PSet') \defeq \asgElm \rst\, (\PSet \setminus
 \PSet')$ and $\XSet \setminus \PSet' \defeq \set{ \asgElm \setminus \PSet' }{
 \asgElm \in \XSet }$.
 Notice that $\XSet \subseteq \YSet$ implies $\XSet \setminus \{ \apElm \}
 \subseteq \YSet \setminus \{ \apElm \}$, for all $\XSet, \YSet \subseteq
 \AsgSet(\PSet)$, $\PSet \subseteq \APSet$, and $\apElm \in \PSet$, and
 $(\bigcup_{\XSet \in \AsgFam} \XSet) \setminus \{ \apElm \} = \bigcup_{\XSet \in
   \AsgFam} (\XSet \setminus \{ \apElm \})$ for all $\AsgFam \in
 \HAsgSet(\PSet)$,
 $\PSet \subseteq \APSet$, and $\apElm \in \PSet$.

 \defenv{proposition}[][Prp][NrmEvlRst]
   {
   $\YSet \setminus \ap{\qntElm} \subseteq \bigcup \AsgFam$, for all $\YSet \in
   \nevlFun[\alpha](\AsgFam, \qntElm)$, $\AsgFam \in \HAsgSet(\PSet)$, and
   $\qntElm \in \QntSet$, with $\PSet \subseteq \APSet$ and $\ap{\qntElm} \cap
   \PSet = \emptyset$.
 }

\begin{proof}
  The proof proceeds by induction on the length of the quantification prefix
  \qntElm.

  (base case) If $\qntElm = \varepsilon$, then we have
  $\nevlFun[\alpha](\AsgFam, \qntElm) = \AsgFam$, and the claim follows trivially.

  If $\qntElm = \Qnt^{\spcElm} \apElm$, let $\YSet \in \nevlFun[\alpha](\AsgFam,
  \qntElm)$.
  If $\alpha$ and $\Qnt$ are coherent, then $\nevlFun[\alpha](\AsgFam, \qntElm)
  = \extFun[\spcElm]{\AsgFam, \apElm}$.
  Thus, $\YSet = \extFun{\XSet, \FFun, \apElm}$ for some $\XSet \in \AsgFam$ and
  $\FFun \in \FncSet[\spcElm](\ap{\AsgFam})$.
  Trivially, $\YSet \setminus \set{p}{} = \XSet$, and we are done.
  If, on the other hand, $\alpha$ and $\Qnt$ are not coherent, then
  $\nevlFun[\alpha](\AsgFam, \qntElm) = \set{ \bigcup_{\FFun \in
      \FncSet[\spcElm](\ap{\AsgFam})} \extFun{\eth(\FFun), \FFun, \apElm} }{ \eth \in
    \FncSet[\spcElm](\ap{\AsgFam}) \to \AsgFam }$.
  Thus, $\YSet = \bigcup_{\FFun \in \FncSet[\spcElm](\ap{\AsgFam})}
  \extFun{\eth(\FFun), \FFun, \apElm} $ for some $\eth \in
  \FncSet[\spcElm](\ap{\AsgFam}) \to \AsgFam$.
  Clearly, $\YSet \setminus \set{\apElm}{} = \bigcup_{\FFun \in
    \FncSet[\spcElm](\ap{\AsgFam})} \eth(\FFun)$.
  Since $\eth(\FFun) \in \AsgFam$ for every $\FFun \in
  \FncSet[\spcElm](\ap{\AsgFam})$, it holds that $\bigcup_{\FFun \in
    \FncSet[\spcElm](\ap{\AsgFam})} \eth(\FFun) \subseteq \bigcup \AsgFam$, hence the
  thesis.

  (inductive step) If $\qntElm = \Qnt^{\spcElm} \apElm . \qntElm'$ (with
  $\qntElm' \neq \varepsilon$), then $\nevlFun[\alpha](\AsgFam, \qntElm) =
  \nevlFun[\alpha](\nevlFun[\alpha](\AsgFam, \Qnt^{\spcElm} \apElm), \qntElm')$.
  Let $\YSet \in \nevlFun[\alpha](\AsgFam, \qntElm)$.
  By inductive hypothesis, $\YSet \setminus \ap{\qntElm'} \subseteq
  \bigcup_{\XSet \in \nevlFun[\alpha](\AsgFam, \Qnt^{\spcElm} \apElm)} \XSet$.
  Since by inductive hypothesis it also holds that $\XSet \setminus \{ \apElm \}
  \subseteq \bigcup \AsgFam$ for every $\XSet \in \nevlFun[\alpha](\AsgFam,
  \Qnt^{\spcElm} \apElm)$, we have that $\bigcup_{\XSet \in
    \nevlFun[\alpha](\AsgFam, \Qnt^{\spcElm} \apElm)}(\XSet \setminus \{ \apElm \})
  \subseteq \bigcup \AsgFam$.
  Therefore, it holds $\YSet \setminus \ap{\qntElm} = (\YSet \setminus
  \ap{\qntElm'}) \setminus \{ \apElm \} \subseteq (\bigcup_{\XSet \in
    \nevlFun[\alpha](\AsgFam, \Qnt^{\spcElm} \apElm)} \XSet) \setminus \{ \apElm \} =
  \bigcup_{\XSet \in \nevlFun[\alpha](\AsgFam, \Qnt^{\spcElm} \apElm)} (\XSet
  \setminus \{ \apElm \}) \subseteq \bigcup \AsgFam$.
\end{proof}





\subsection{Proofs for Section~\ref{sec:canqnt}}
\label{app:canqnt}

A \emph{path} $\pthElm \!\in\! \PthSet \!\subseteq\! \PosSet[][\infty]$ is a
finite or infinite sequence of positions compatible with the move relation,
\ie, $((\pthElm)_{i}, (\pthElm)_{i + 1}) \in \MovRel$, for all $i \in
\numco{0}{\card{\pthElm} - 1}$; it is \emph{initial} if $\card{\pthElm} > 0$
and $\fst{\pthElm} = \iposElm$.
A \emph{history} for player $\alpha \in \{ \PlrSym, \OppSym \}$ is a finite
initial path $\hstElm \in \HstSet[\alpha] \subseteq \PthSet \cap (\PosSet[][*]
\cdot \PosSet[\alpha])$ terminating in an $\alpha$-position.
A \emph{play} $\playElm \in \PlaySet \subseteq \PthSet \cap \PosSet[][\omega]$
is just an infinite initial path.

A \emph{strategy} for player $\alpha \in \{ \PlrSym, \OppSym \}$ is a function
$\strElm[\alpha] \in \StrSet[\alpha] \subseteq \HstSet[\alpha] \to \PosSet$
mapping each $\alpha$-history $\hstElm \in \HstSet[\alpha]$ to a position
$\strElm[\alpha](\hstElm) \in \PosSet$ compatible with the move relation, \ie,
$(\lst{\hstElm}, \allowbreak \strElm[\alpha](\hstElm)) \in \MovRel$.
A path $\pthElm \in \PthSet$ is \emph{compatible} with a pair of strategies
$(\pstrElm, \ostrElm) \in \PStrSet \times \OStrSet$ if, for all $i \in
\numco{0}{\card{\pthElm} - 1}$, it holds that $(\pthElm)_{i + 1} =
\pstrElm((\pthElm)_{\leq i})$, if $(\pthElm)_{i} \in \PPosSet$, and
$(\pthElm)_{i + 1} = \ostrElm((\pthElm)_{\leq i})$, otherwise.
The \emph{play function} $\playFun \colon \PStrSet \times \OStrSet \to
\PlaySet$ returns, for each pair of strategies $(\pstrElm, \ostrElm) \in
\PStrSet \times \OStrSet$, the unique play $\playFun(\pstrElm, \ostrElm) \in
\PlaySet$ compatible with them.

The \emph{observation function} $\obsFun \colon \PthSet \to \ObsSet[][\infty]$
associates with each path $\pthElm \in \PthSet$ the ordered sequence $\wElm
\defeq \obsFun(\pthElm) \in \ObsSet[][\infty]$ of all observable positions
occurring in it; formally, there exists a monotone bijection $\fFun \colon
\numco{0}{\card{\wElm}} \!\to\! \set{\! j \!\in\! \numco{0}{\card{\pthElm}} }{
(\pthElm)_{j} \!\in\! \ObsSet \!}$ satisfying the equality $(\wElm)_{i} =
(\pthElm)_{\fFun(i)}$, for all $i \in \numco{0}{\card{\wElm}}$.

Eloise \emph{wins} the game if there exists a strategy $\pstrElm \in \PStrSet$
such that $\obsFun(\playFun(\pstrElm, \ostrElm)) \in \WinSet$, for all
$\ostrElm \in \OStrSet$.
Similarly, Abelard \emph{wins} the game if there exists a strategy $\ostrElm
\in \OStrSet$ such that $\obsFun(\playFun(\pstrElm, \ostrElm)) \in
\dual{\WinSet}$, for all $\pstrElm \in \PStrSet$.

\recenv{ThmQntGamI}
\begin{proof}
  Let $\Game[\qntElm][\Psi]$ be the game defined as prescribed in
  Construction~\ref{cns:qntgami}.
  Obviously, this is a Borelian game, due to the hypothesis on the property
  $\Psi$.

  Before continuing, observe that, due to the specific structure of the game,
  every history $\hstElm \cdot \posElm \in \HstSet[\alpha]$ is bijectively
  correlated with the sequence of positions $\obsFun(\hstElm) \cdot \posElm \in
  \ObsSet^{*} \cdot \PosSet[\alpha]$, for any player $\alpha \in \{ \PlrSym,
  \OppSym \}$.
  In other words, the functions $\shrFun[\alpha] \colon \HstSet[\alpha] \to
  \ObsSet^{*} \cdot \PosSet[\alpha]$ defined as $\shrFun[\alpha](\hstElm \cdot
  \posElm) \defeq \obsFun(\hstElm) \cdot \posElm$ are bijective.

  Thanks to this observation, it is thus immediate to show that, for each
  strategy $\pstrElm \in \PStrSet$, there is a unique function $\der{\pstrElm}
  \colon \ObsSet^{*} \cdot \PPosSet \to \PosSet$ and, \viceversa, for each
  function $\der{\pstrElm} \colon \ObsSet^{*} \cdot \PPosSet \to \PosSet$, there
  is a unique strategy $\pstrElm \in \PStrSet$ such that
  \[
    \der{\pstrElm}(\shrFun[\PlrSym](\hstElm)) = \pstrElm(\hstElm), \text{ for
    all histories } \hstElm \in \PHstSet.
  \]
  Similarly, for each strategy $\ostrElm \in \OStrSet$, there is a unique
  function $\der{\ostrElm} \colon \ObsSet^{*} \cdot \OPosSet \to \PosSet$ and,
  \viceversa, for each function $\der{\ostrElm} \colon \ObsSet^{*} \cdot
  \OPosSet \to \PosSet$, there is a unique strategy $\ostrElm \in \OStrSet$
  satisfying the equality
    \[
    \der{\ostrElm}(\shrFun[\OppSym](\hstElm)) = \ostrElm(\hstElm), \text{ for
    all histories } \hstElm \in \OHstSet.
  \]

  We can now proceed with the proof of the two properties.
  \begin{itemize}
    \item\textbf{[\ref{thm:qntgami(elo)}]}
      Since Eloise wins the game, she has a winning strategy, \ie, there is
      $\pstrElm \in \PStrSet$ such that $\obsFun(\playFun(\pstrElm, \ostrElm))
      \in \WinSet$, for all $\ostrElm \in \OStrSet$.
      We want to prove that there exists $\ESet \in
      \nevlFun[\exists\forall](\CSym[\forall\exists](\qntElm))$ such that
      $\ESet \subseteq \Psi$.

      First, recall that $\CSym[\forall\exists](\qntElm) = \forall^{\BSym}
      \vec{\apElm} \ldotp \exists^{\vec{\spcElm}} \vec{\qapElm}$, for some
      vectors of atomic propositions $\vec{\apElm}, \vec{\qapElm} \in
      \APSet[][*]$ and quantifier specifications $\vec{\spcElm} \in
      \SpcSet^{\card{\vec{\qapElm}}}$.
      Moreover, thanks to Propositions~\ref{prp:evl} and~\ref{prp:nrmevlind},
      the following claim can be proved.
      \begin{claim}
        $\ESet \subseteq \Psi$, for some $\ESet \in
        \nevlFun[\exists\forall](\CSym[\forall\exists](\qntElm))$, \iff there
        exists $\vec{\FFun} \in \FncSet[\vec{\spcElm}](\vec{\apElm})$ such that
        $\extFun{\asgElm, \vec{\FFun}, \vec{\qapElm}} \in \Psi$, for all
        $\asgElm \in \AsgSet(\vec{\apElm})$.
      \end{claim}
      Due to the above characterization of the existence of a set $\ESet \in
      \nevlFun[\exists\forall](\CSym[\forall\exists](\qntElm))$ such that $\ESet
      \subseteq \Psi$, the thesis can be proved by defining a suitable vector of
      functors $\vec{\FFun} \in \FncSet[\vec{\spcElm}](\vec{\apElm})$.

      Consider an arbitrary assignment $\asgElm \in \AsgSet(\vec{\apElm})$ and
      define the function $\der{\ostrElm}[][\asgElm] \colon \ObsSet^{*} \cdot
      \OPosSet \to \PosSet$ as follows, for all finite sequences of observable
      positions $\wrdElm \in \ObsSet^{*}$ and Abelard's positions $\valElm \in
      \OPosSet$:
      \[
        \der{\ostrElm}[][\asgElm](\wrdElm \cdot \valElm) \defeq
        \begin{cases}
          \emptyfun,
            & \text{if } \valElm \in \ObsSet; \\
          {\valElm}[\xapElm \mapsto \asgElm(\xapElm)(\card{\wrdElm})],
            & \text{otherwise};
        \end{cases}
      \]
      where $\xapElm \in \vec{\apElm}$ is the atomic proposition at position
      $\#(\valElm)$ in the prefix $\qntElm$, \ie, $(\qntElm)_{\#(\valElm)} =
      \forall^{\BSym} \xapElm$.
      Due to the bijective correspondence previously described, there is a
      unique strategy $\ostrElm[][\asgElm] \!\in\! \OStrSet$ such that
      $\ostrElm[][\asgElm](\hstElm) =
      \der{\ostrElm}[][\asgElm](\shrFun[\OppSym](\hstElm))$, for all histories
      $\hstElm \in \OHstSet$.
      Obviously, the induced play $\playElm[][\asgElm] \defeq \playFun(\pstrElm,
      \ostrElm[][\asgElm])$ is won by Eloise, \ie, $\wrdElm[][\asgElm] \defeq
      \obsFun(\playElm[][\asgElm]) \in \WinSet$.

      Thanks to all the infinite sequences $\wrdElm[][\asgElm]$, one for each
      assignment $\asgElm \in \AsgSet(\vec{\apElm})$, we can define every
      component $(\vec{\FFun})_{i}$ of the vector of functors $\vec{\FFun} \in
      (\FncSet(\vec{\apElm}))^{\card{\vec{\qapElm}}}$ as follows, for all
      instants of time $t \in \SetN$, where $i \in
      \numco{0}{\card{\vec{\qapElm}}}$:
      \[
        (\vec{\FFun})_{i}(\asgElm)(t) \defeq\,
        (\wrdElm[][\asgElm])_{t}((\vec{\qapElm})_{i}).
      \]
      It is not too hard to show that this functor complies with the vector of
      quantifier specifications $\vec{\spcElm}$.
      \begin{claim}
        $\vec{\FFun} \in \FncSet[\vec{\spcElm}](\vec{\apElm})$.
      \end{claim}

      At this point, for all assignments $\asgElm \in \AsgSet(\vec{\apElm})$,
      let $\asgElm[\vec{\FFun}] \defeq \allowbreak \extFun{\asgElm, \vec{\FFun},
      \vec{\qapElm}}$.
      We can easily argue that $\asgElm[\vec{\FFun}] \in \Psi$.
      Indeed, by construction of the strategy $\ostrElm[][\asgElm]$ and the
      vector of functors $\vec{\FFun}$, it holds that
      $\asgElm[\vec{\FFun}](\xapElm)(t) = (\wrdElm[][\asgElm])_{t}(\xapElm)$,
      for all instants of time $t \in \SetN$ and atomic propositions $\xapElm
      \in \vec{\apElm} \cdot \vec{\qapElm}$.
      Hence, $\wrdFun(\asgElm[\vec{\FFun}]) = \wrdElm[][\asgElm]$, which
      implies $\asgElm[\vec{\FFun}] \in \Psi$, since $\wrdElm[][\asgElm] \in
      \WinSet$.

    \item\textbf{[\ref{thm:qntgami(abe)}]}
      Since Abelard wins the game, he has a winning strategy, \ie, there is
      $\ostrElm \in \OStrSet$ such that $\obsFun(\playFun(\pstrElm, \ostrElm))
      \not\in \WinSet$, for all $\pstrElm \in \PStrSet$.
      We want to prove that, for all $\ESet \in
      \nevlFun[\exists\forall](\CSym[\exists\forall](\qntElm))$, it holds that
      $\ESet \not\subseteq \Psi$.

      First, recall that $\CSym[\exists\forall](\qntElm) = \exists^{\BSym}
      \vec{\qapElm} \ldotp \forall^{\vec{\spcElm}} \vec{\apElm}$, for some
      vectors of atomic propositions $\vec{\apElm}, \vec{\qapElm} \in
      \APSet[][*]$ and quantifier specifications $\vec{\spcElm} \in
      \SpcSet^{\card{\vec{\apElm}}}$.
      Moreover, thanks to Propositions~\ref{prp:evl} and~\ref{prp:nrmevlind},
      the following claim can be proved.
      \begin{claim}
        $\ESet \not\subseteq \Psi$, for all $\ESet \in
        \nevlFun[\exists\forall](\CSym[\exists\forall](\qntElm))$, \iff there
        exists $\vec{\GFun} \in \FncSet[\vec{\spcElm}](\vec{\qapElm})$ such that
        $\extFun{\asgElm, \vec{\GFun}, \vec{\apElm}} \not\in \Psi$, for all
        $\asgElm \in \AsgSet(\vec{\qapElm})$.
      \end{claim}
      Due to the above characterization of non-existence of a set $\ESet \in
      \nevlFun[\exists\forall](\CSym[\exists\forall](\qntElm))$ such that $\ESet
      \subseteq \Psi$, the thesis can be proved by defining a suitable vector of
      functors $\vec{\GFun} \in \FncSet[\vec{\spcElm}](\vec{\qapElm})$.

      Consider an arbitrary assignment $\asgElm \in \AsgSet(\vec{\qapElm})$ and
      define the function $\der{\pstrElm}[][\asgElm] \colon \ObsSet^{*} \cdot
      \PPosSet \to \PosSet$ as follows, for all finite sequences of observable
      positions $\wrdElm \in \ObsSet^{*}$ and Eloise's positions $\valElm \in
      \PPosSet$:
      \[
        \der{\pstrElm}[][\asgElm](\wrdElm \cdot \valElm) \defeq
        {\valElm}[\xapElm \mapsto \asgElm(\xapElm)(\card{\wrdElm})],
      \]
      where $\xapElm \in \vec{\qapElm}$ is the atomic proposition at position
      $\#(\valElm)$ in the prefix $\qntElm$, \ie, $(\qntElm)_{\#(\valElm)} =
      \exists^{\BSym} \xapElm$.
      Due to the bijective correspondence previously described, there is a
      unique strategy $\pstrElm[][\asgElm] \!\in\! \PStrSet$ such that
      $\pstrElm[][\asgElm](\hstElm) =
      \der{\pstrElm}[][\asgElm](\shrFun[\PlrSym](\hstElm))$, for all histories
      $\hstElm \in \PHstSet$.
      Obviously, the induced play $\playElm[][\asgElm] \defeq
      \playFun(\pstrElm[][\asgElm], \ostrElm)$ is won by Abelard, \ie,
      $\wrdElm[][\asgElm] \defeq \obsFun(\playElm[][\asgElm]) \not\in \WinSet$.

      Thanks to all the infinite sequences $\wrdElm[][\asgElm]$, one for each
      assignment $\asgElm \in \AsgSet(\vec{\qapElm})$, we can define every
      component $(\vec{\GFun})_{i}$ of the vector of functors $\vec{\GFun} \in
      (\FncSet(\vec{\qapElm}))^{\card{\vec{\apElm}}}$ as follows, for all
      instants of time $t \in \SetN$, where $i \in
      \numco{0}{\card{\vec{\apElm}}}$:
      \[
        (\vec{\GFun})_{i}(\asgElm)(t) \defeq\,
        (\wrdElm[][\asgElm])_{t}((\vec{\apElm})_{i}).
      \]
      It is not too hard to show that this functor complies with the vector of
      quantifier specifications $\vec{\spcElm}$.
      \begin{claim}
        $\vec{\GFun} \in \FncSet[\vec{\spcElm}](\vec{\qapElm})$.
      \end{claim}

      At this point, for all assignments $\asgElm \in \AsgSet(\vec{\qapElm})$,
      let $\asgElm[\vec{\GFun}] \defeq \allowbreak \extFun{\asgElm, \vec{\GFun},
      \vec{\apElm}}$.
      We can easily argue that $\asgElm[\vec{\GFun}] \not\in \Psi$.
      Indeed, by construction of the strategy $\pstrElm[][\asgElm]$ and the
      vector of functors $\vec{\GFun}$, it holds that
      $\asgElm[\vec{\GFun}](\xapElm)(t) = (\wrdElm[][\asgElm])_{t}(\xapElm)$,
      for all instants of time $t \in \SetN$ and atomic propositions $\xapElm
      \in \vec{\qapElm} \cdot \vec{\apElm}$.
      Hence, $\wrdFun(\asgElm[\vec{\GFun}]) = \wrdElm[][\asgElm]$, which
      implies $\asgElm[\vec{\GFun}] \not\in \Psi$, since $\wrdElm[][\asgElm]
      \not\in \WinSet$.
      \qedhere
  \end{itemize}
\end{proof}

$\SpcSet[\BSym]$ is the set of behavioral quantifier specifications, \ie,
quantifier specifications of the form $\BSym \cup \sdep{\PSet[\SSym]}$ for some
$\PSet[\SSym] \subseteq \APSet$.

\defenv{proposition}[][Prp][QntSwpBas]
  {
  $\evlFun[\QEA](\AsgFam, \forall^{\BSym} \apElm \ldotp\!
  \exists^{\spcElm \cup \sdep{\apElm}} \qapElm)  \!\sqsubseteq\!
  \evlFun[\QEA](\AsgFam, \exists^{\spcElm} \qapElm \ldotp\!
  \allowbreak \forall^{\BSym} \apElm)$ and $\evlFun[\QAE](\AsgFam,
  \exists^{\BSym} \apElm \ldotp\!
  \forall^{\spcElm \cup \sdep{\apElm}} \qapElm) \!\sqsubseteq\!
  \evlFun[\QAE](\AsgFam, \forall^{\spcElm} \qapElm \ldotp\!
  \allowbreak \exists^{\BSym} \apElm)$, for all $\AsgFam \!\in\!
  \HAsgSet(\PSet)$ with $\PSet \!\subseteq\! \APSet$, $\apElm, \qapElm \in
  \APSet \setminus \PSet$, and $\spcElm \in \SpcSet[\BSym]$.
  }
\begin{proof}
  Due to the specific definition of the normal evaluation function
  $\nevlFun[\exists\forall](\AsgFam, \qntElm)$, and by exploiting
  Propositions~\ref{prp:evl} and~\ref{prp:nrmevlind}, the following claim can
  be proved.
  \begin{claim}
    $\nevlFun[\exists\forall](\AsgFam, \forall^{\BSym} \apElm \ldotp\!
    \exists^{\spcElm \cup \sdep{\apElm}} \qapElm)  \!\sqsubseteq\!
    \nevlFun[\exists\forall](\AsgFam, \exists^{\spcElm} \qapElm \ldotp\!
    \forall^{\BSym} \apElm)$ \iff, for all $\JFun \!\in\! \FncSet[\spcElm \cup
    \sdep{\apElm}](\ap{\AsgFam} \cup \{ \apElm \})$, $\eth \!\in\!
    \FncSet[\BSym](\ap{\AsgFam}) \!\to\! \AsgFam$, and $\GFun \!\in\!
    \FncSet[\BSym](\ap{\AsgFam} \cup \{ \qapElm \})$, there exists $\FFun
    \!\in\! \FncSet[\spcElm](\ap{\AsgFam})$ and $\XSet \!\in\! \AsgFam$ such that
    $\extFun{\extFun{\XSet, \FFun, \qapElm}, \GFun, \apElm} \subseteq
    \extFun{\extFun{\eth, \apElm}, \JFun, \qapElm}$.
  \end{claim}
%
  An analogous claim can be proved stating that the same characterization also
  holds for $\nevlFun[\QAE](\AsgFam, \exists^{\BSym} \apElm \ldotp\!
    \forall^{\spcElm \cup \sdep{\apElm}} \qapElm)  \!\sqsubseteq\!
    \nevlFun[\QAE](\AsgFam, \forall^{\spcElm} \qapElm \ldotp\!
    \exists^{\BSym} \apElm)$.
  Thanks to such characterizations, the thesis can be shown by choosing a
  suitable functor $\FFun$ and set of assignments $\XSet$ in dependence of the
  functors $\JFun$ and $\GFun$ and the selection map $\eth$.

  In order to define the functor $\FFun$, let us inductively construct, for
  every given assignment $\asgElm \in \AsgSet(\ap{\AsgFam})$, the following
  infinite families of assignments $\{ \aasgElm[t][\asgElm] \in
  \AsgSet(\ap{\AsgFam} \cup \{ \apElm \}) \}_{t \in \SetN}$, Boolean values $\{
  \vElm[t][\asgElm] \!\in\! \SetB \}_{t \in \SetN}$, and assignments $\{
  \basgElm[t][\asgElm] \!\in\! \AsgSet(\ap{\AsgFam} \allowbreak \cup \{ \qapElm
  \}) \}_{t \in \SetN}$:
  \begin{itemize}
    \item
      as base step $t \!=\! 0$, we choose $\aasgElm[0][\asgElm] \in
      \AsgSet(\ap{\AsgFam} \cup \{ \apElm \})$ as an arbitrary assignment for
      which the equality $\aasgElm[0][\asgElm] \rst \ap{\AsgFam} = \asgElm$
      holds true, the Boolean value $\vElm[0][\asgElm] \in \SetB$ as
      $\JFun(\aasgElm[0][\asgElm])(0)$, \ie, $\vElm[0][\asgElm] \defeq
      \JFun(\aasgElm[0][\asgElm])(0)$, and $\basgElm[0][\asgElm] \in
      \AsgSet(\ap{\AsgFam} \cup \{ \qapElm \})$ as an arbitrary assignment with
      $\basgElm[0][\asgElm] \rst \ap{\AsgFam} = \asgElm$ such that, at time $0$
      on the variable $\qapElm$, assumes $\vElm[0][\asgElm]$ as value, \ie,
      $\basgElm[0][\asgElm](\qapElm)(0) = \vElm[0][\asgElm]$;
    \item
      as inductive step $t > 0$, we derive the assignment $\aasgElm[t][\asgElm]
      \in \AsgSet(\ap{\AsgFam} \cup \{ \apElm \})$ from $\GFun(\basgElm[t -
      1][\asgElm])$, \ie, $\aasgElm[t][\asgElm] \defeq {\asgElm}[\apElm \mapsto
      \GFun(\basgElm[t - 1][\asgElm])]$, and the Boolean value
      $\vElm[t][\asgElm] \in \SetB$ from $\JFun(\aasgElm[t][\asgElm])(t)$, \ie,
      $\vElm[t][\asgElm] \defeq \JFun(\aasgElm[t][\asgElm])(t)$; moreover, we
      choose $\basgElm[t][\asgElm] \in \AsgSet(\ap{\AsgFam} \cup \{ \qapElm \})$
      as an arbitrary assignment with $\basgElm[t][\asgElm] \rst \ap{\AsgFam} =
      \asgElm$ such that, on the variable $\qapElm$, is equal to $\basgElm[t -
      1][\asgElm]$ up to time $t$ excluded and assumes $\vElm[t][\asgElm]$ as
      value at time $t$, \ie, $\basgElm[t][\asgElm](\qapElm)(h) = \basgElm[t -
      1][\asgElm](\qapElm)(h)$, for all $h \in \numco{0}{t}$, and
      $\basgElm[t][\asgElm](\qapElm)(t) = \vElm[t][\asgElm]$.
  \end{itemize}

  Thanks to the infinite family of Boolean values $\{ \vElm[t][\asgElm] \in
  \SetB \}_{t \in \SetN}$, one for each assignment $\asgElm \in
  \AsgSet(\ap{\AsgFam})$, we can define the functor $\FFun \in
  \FncSet(\ap{\AsgFam})$ as follows, for every instant of time $t \in \SetN$:
  \[
    \FFun(\asgElm)(t) \defeq \vElm[t][\asgElm].
  \]
  It is easy to show that this functor complies with the quantifier
  specification $\spcElm$, since the functor $\JFun$, from which $\FFun$ is
  derived, is compliant with the quantifier specification $\spcElm \cup
  \sdep{\apElm}$.
  \begin{claim}
    $\FFun \in \FncSet[\spcElm](\ap{\AsgFam})$.
  \end{claim}

  Before continuing, let us first introduce the functor $\HFun \in
  \FncSet(\ap{\AsgFam})$ as follows, for every assignment $\asgElm \in
  \AsgSet(\ap{\AsgFam})$:
  \[
    \HFun(\asgElm) \defeq \extFun{\extFun{\asgElm, \FFun, \qapElm}, \GFun,
    \apElm}(\apElm).
  \]
  It is quite immediate to verify that such a functor is behavioral.
  \begin{claim}
    $\HFun \in \FncSet[\BSym](\ap{\AsgFam})$.
  \end{claim}

  At this point, consider the set of assignments $\XSet \defeq \eth(\HFun)$.
  Thanks to the specific definitions of the two functors $\FFun$ and $\HFun$,
  the following claim can be proved.
  \begin{claim}
    $\extFun{\extFun{\XSet, \FFun, \qapElm}, \GFun, \apElm} \subseteq
    \extFun{\extFun{\XSet, \HFun, \apElm}, \JFun, \qapElm}$.
  \end{claim}

  Now, it is obvious that $\extFun{\XSet, \HFun, \apElm} \subseteq \extFun{\eth,
  \apElm}$, due to the definition of the latter and the choice of the set
  $\XSet$, which immediately implies $\extFun{\extFun{\XSet, \HFun, \apElm},
  \JFun, \qapElm} \subseteq \extFun{\extFun{\eth, \apElm}, \JFun, \qapElm}$.
  Therefore, $\extFun{\extFun{\XSet, \FFun, \qapElm}, \GFun, \apElm} \subseteq
  \extFun{\extFun{\eth, \apElm}, \JFun, \qapElm}$, which concludes the proof.
\end{proof}

\recenv{PrpQntSwpInd}
\begin{proof}
  First, we notice that, by repeatedly applying Proposition~\ref{prp:qntswpbas},
  we are able to derive the following claim.
  \begin{claim}
    $\nevlFun[\QEA](\AsgFam, \forall^{\BSym} \vec{\apElm}
    \ldotp\!
    \exists^{\BSym \cup \sdep{\vec{\apElm}}} \vec{\qapElm})
    \!\sqsubseteq\!
    \nevlFun[\QEA](\AsgFam, \exists^{\BSym} \vec{\qapElm}
    \ldotp\!
    \allowbreak \forall^{\BSym} \vec{\apElm})$ and
    $\nevlFun[\QAE](\AsgFam, \exists^{\BSym} \vec{\apElm}
    \ldotp\!
    \forall^{\BSym \cup \sdep{\vec{\apElm}}} \vec{\qapElm})
    \!\sqsubseteq\!
    \nevlFun[\QAE](\AsgFam, \forall^{\BSym} \vec{\qapElm}
    \ldotp\!
    \allowbreak \exists^{\BSym} \vec{\apElm})$, for all $\AsgFam \!\in\!
    \HAsgSet(\PSet)$ with $\PSet \!\subseteq\! \APSet$ and $\vec{\apElm},
    \vec{\qapElm} \subseteq \APSet \setminus \PSet$.
  \end{claim}
  Let us first focus on proving $\!\nevlFun[\alpha](\AsgFam,
  \CSym[\dual{\alpha}](\qntElm)) \!\sqsubseteq\! \nevlFun[\alpha](\AsgFam,
  \qntElm)$ and consider the case $\alpha = \exists\forall$.
  First, let us rewrite $\qntElm$ as $\forall^{\BSym} \vec{\apElm}[0] \ldotp
  (\exists^{\BSym} \vec{\qapElm}[i] \ldotp \forall^{\BSym} \vec{\apElm}[i])_{i =
    1}^{k} \ldotp \exists^{\BSym} \vec{\qapElm}[k + 1]$, where $| \vec{\apElm}[0] |,
  | \vec{\qapElm}[k+1] | \geq 0$, $k \geq 0$, and $| \vec{\qapElm}[i] |, |
  \vec{\apElm}[i] | \geq 1$ for all $i \in \{ 1, \ldots, k \}$.  The proof proceeds by induction on the
  \emph{$\exists\forall$-alternation degree} of $\qntElm$, that is, $k$.

  If $k=0$ (base case), then $\CSym[\forall\exists](\qntElm) = \qntElm$, and we
  are done.
  If $k>0$ (inductive step), then $\nevlFun[\exists\forall](\AsgFam, \qntElm) =
  \nevlFun[\exists\forall] ( \allowbreak \nevlFun[\exists\forall] (
  \nevlFun[\exists\forall] (\AsgFam, \forall^{\BSym} \vec{\apElm}[0] \ldotp
  (\exists^{\BSym} \vec{\qapElm}[i] \ldotp \forall^{\BSym} \vec{\apElm}[i])_{i =
    1}^{k-1}), \exists^{\BSym} \vec{\qapElm}[k] \ldotp \forall^{\BSym}
  \vec{\apElm}[k] ), \exists^{\BSym} \vec{\qapElm}[k + 1] )$.
  Thanks to the above claim and to the monotonicity of \nevlFun, we have
  $\nevlFun[\exists\forall](\AsgFam, \qntElm) \sqsupseteq \nevlFun[\exists\forall]
  ( \allowbreak \nevlFun[\exists\forall] ( \allowbreak \nevlFun[\exists\forall] (
  \allowbreak \AsgFam, \forall^{\BSym} \vec{\apElm}[0] \ldotp (\exists^{\BSym}
  \vec{\qapElm}[i] \ldotp \forall^{\BSym} \vec{\apElm}[i])_{i = 1}^{k-1}),
  \forall^{\BSym} \vec{\apElm}[k] \ldotp \exists^{\BSym \cup
    \sdep{\vec{\apElm}[k]}} \vec{\qapElm}[k]), \exists^{\BSym} \vec{\qapElm}[k + 1]
  ) = \nevlFun[\exists\forall]( \allowbreak \nevlFun[\exists\forall]( \allowbreak
  \AsgFam, \qntElm'), \exists^{\BSym \cup \sdep{\vec{\apElm}[k]}} \vec{\qapElm}[k]
  \ldotp \exists^{\BSym} \vec{\qapElm}[k + 1])$, where $\qntElm' = \forall^{\BSym}
  \vec{\apElm}[0] \ldotp (\exists^{\BSym} \vec{\qapElm}[i] \ldotp \forall^{\BSym}
  \vec{\apElm}[i])_{i = 1}^{k-2} \ldotp \exists^{\BSym} \vec{\qapElm}[k-1] \ldotp
  \forall^{\BSym} \vec{\apElm}[k-1] \ldotp \forall^{\BSym} \vec{\apElm}[k]$ has
  $\exists\forall$-alternation degree $k-1$.
  Thus, by inductive hypothesis, it holds that
  $\nevlFun[\exists\forall](\AsgFam, \CSym[\forall\exists](\qntElm'))
  \!\sqsubseteq\! \nevlFun[\exists\forall](\AsgFam, \qntElm')$.
  The thesis follows by observing that $\CSym[\forall\exists](\qntElm') \ldotp
  \exists^{\BSym \cup \sdep{\vec{\apElm}[k]}} \vec{\qapElm}[k] \ldotp
  \exists^{\BSym} \vec{\qapElm}[k + 1] = \CSym[\forall\exists](\qntElm)$.

  In order to prove that $\!\nevlFun[\alpha](\AsgFam,
  \CSym[\dual{\alpha}](\qntElm)) \!\sqsubseteq\! \nevlFun[\alpha](\AsgFam,
  \qntElm)$ holds when $\alpha = \QAE$, it suffices to replace quantifier
  $\forall$ with $\exists$, and vice versa, inside the quantifier prefixes, and to
  switch $\QEA$ and $\QAE$ throughout the previous proof. We omit the details.

  Finally, in order to prove that $\nevlFun[\alpha](\AsgFam, \qntElm)
  \!\sqsubseteq\! \nevlFun[\alpha](\AsgFam, \CSym[\alpha](\qntElm))$ for every
  $\alpha \in \{ \QAE, \QAE \}$, we first state two auxiliary results.
  \begin{claim}\label{clm:nevldualization}
    $\nevlFun[\alpha](\dual{\AsgFam}, \dual{\qntElm}) \equiv
    \dual{\nevlFun[\alpha](\AsgFam, \qntElm)}$, for all $\AsgFam \in \HAsgSet$ and
    $\qntElm \in \QntSet$.
  \end{claim}
  \begin{proof}
    The proof is by induction on the length of $\qntElm$.
    If $\qntElm = \varepsilon$, then the claim follows immediately.
    Let $\qntElm = \Qnt^{\spcElm} \apElm \ldotp \qntElm'$, then we distinguish
    two cases.

    If $\alpha$ and $\Qnt$ are coherent, then $\nevlFun[\alpha](\dual{\AsgFam},
    \dual{\qntElm}) \equiv
    \nevlFun[\alpha](\dual{\extFun[\spcElm]{\dual{\dual{\AsgFam}}, \apElm}},
    \dual{\qntElm'})$ (recall that $\nevlFun[\alpha](\AsgFam, \Qnt^{\spcElm}\apElm) \equiv
    \evlFun[\alpha](\AsgFam, \Qnt^{\spcElm}\apElm)$ due to
    Proposition~\ref{prp:nrmevlbas}).
    By
    Propositions~\ref{prp:dltinv},~\ref{prp:oprmon},~\ref{prp:evlmon},
    and~\ref{prp:nrmevlind}, we have that
    $\nevlFun[\alpha](\dual{\extFun[\spcElm]{\dual{\dual{\AsgFam}}, \apElm}},
    \dual{\qntElm'}) \equiv \nevlFun[\alpha](\dual{\extFun[\spcElm]{\AsgFam,
        \apElm}}, \dual{\qntElm'})$.
    Moreover, we have that $\dual{\nevlFun[\alpha](\AsgFam, \qntElm)} =
    \dual{\nevlFun[\alpha](\extFun[\spcElm]{\AsgFam, \apElm}, \qntElm')}$.
    By inductive hypothesis, it holds that
    $\nevlFun[\alpha](\dual{\extFun[\spcElm]{\AsgFam, \apElm}}, \dual{\qntElm'})
    \equiv \dual{\nevlFun[\alpha](\extFun[\spcElm]{\AsgFam, \apElm}, \qntElm')}$,
    hence the thesis.

    If $\alpha$ and $\Qnt$ are not coherent, then
    $\nevlFun[\alpha](\dual{\AsgFam}, \dual{\qntElm}) =
    \nevlFun[\alpha](\extFun[\spcElm]{\dual{\AsgFam}, \apElm}, \dual{\qntElm'})$.
    Moreover, $\dual{\nevlFun[\alpha](\AsgFam, \qntElm)} \equiv
    \dual{\nevlFun[\alpha](\dual{\extFun[\spcElm]{\dual{\AsgFam}, \apElm}},
      \qntElm')}$.
    By inductive hypothesis,
    $\dual{\nevlFun[\alpha](\dual{\extFun[\spcElm]{\dual{\AsgFam}, \apElm}},
      \qntElm')} \equiv \nevlFun[\alpha](\dual{\dual{\extFun[\spcElm]{\dual{\AsgFam},
          \apElm}}}, \dual{\qntElm'}) \equiv
    \nevlFun[\alpha](\extFun[\spcElm]{\dual{\AsgFam}, \apElm}, \dual{\qntElm'}) =
    \nevlFun[\alpha](\dual{\AsgFam}, \dual{\qntElm})$
  \end{proof}

  It is not difficult to convince oneself that the following claim holds.
  \begin{claim}\label{clm:prefix}
    $\CSym[\dual{\alpha}](\dual{\qntElm}) = \dual{\CSym[\alpha](\qntElm)}$ for
    all $\qntElm \in \QntSet$.
  \end{claim}

  Now, by instantiating $\AsgFam$ and $\qntElm$ with $\dual{\AsgFam}$ and
  $\dual{\qntElm}$, respectively, in the first ``squared inclusion'' of this
  proposition, proved in the first part of this proof, we know that
  $\nevlFun[\alpha](\dual{\AsgFam}, \CSym[\dual{\alpha}](\dual{\qntElm}))
  \sqsubseteq \nevlFun[\alpha](\dual{\AsgFam}, \dual{\qntElm})$.
  By Claims~\ref{clm:nevldualization} and~\ref{clm:prefix}, we have
  $\dual{\nevlFun[\alpha](\AsgFam, \CSym[\alpha](\qntElm))} \sqsubseteq
  \dual{\nevlFun[\alpha](\AsgFam, \qntElm)}$.
  The thesis follows from Proposition~\ref{prp:oprmon}.
\end{proof}

\addtocounter{proposition}{5}
\defenv{proposition}[][Prp][QntSwpIndI]
  {
  $\!\evlFun[\alpha](\AsgFam, \CSym[\dual{\alpha}](\qntElm[1] \ldotp
  \qntElm[2])) \!\sqsubseteq\! \YFam \!\sqsubseteq\! \evlFun[\alpha](\AsgFam,
  \qntElm[1] \ldotp \qntElm[2])$, for all $\YFam \in \{
  \evlFun[\alpha](\AsgFam, \qntElm[1] \ldotp \CSym[\dual{\alpha}](\qntElm[2])),
  \evlFun[\alpha](\AsgFam, \CSym[\dual{\alpha}](\qntElm[1]) \ldotp \qntElm[2])
  \}$, $\AsgFam \in \HAsgSet$ and $\qntElm[1], \qntElm[2] \in \QntSet[\BSym]$,
  with $\ap{\qntElm[1] \ldotp \qntElm[2]} \cap \ap{\AsgFam} = \emptyset$.
  }
\begin{proof}
  We will make use of the following result.

  \begin{claim}\label{clm:prefixswpI}
    $\CSym[\alpha](\CSym[\alpha](\qntElm[1]) \ldotp \qntElm[2]) =
    \CSym[\alpha](\qntElm[1] \ldotp \qntElm[2]) = \CSym[\alpha](\qntElm[1] \ldotp
    \CSym[\alpha](\qntElm[2]))$ for all $\qntElm[1], \qntElm[2] \in \QntSet$.
  \end{claim}

  \begin{remark}
    Notice that due to the fact that $\CSym[\alpha](\qntElm[1]) \notin
    \QntSet[\BSym]$ and $\CSym[\alpha](\qntElm[2]) \notin \QntSet[\BSym]$, the two
    expressions $\CSym[\alpha](\CSym[\alpha](\qntElm[1]) \ldotp \qntElm[2])$ and
    $\CSym[\alpha](\qntElm[1] \ldotp \CSym[\alpha](\qntElm[2]))$ are not well
    defined.
    Thus, in order for the above claim to make sense, we need generalized
    definitions for functions $\CSym[\QEA]$ and $\CSym[\QAE]$ to cope with
    quantifier prefixes featuring arbitrary quantifier specifications.
  \end{remark}

  \begin{enumerate}
  \item In order to show that $\nevlFun[\alpha](\AsgFam, \qntElm[1] \ldotp
    \CSym[\dual{\alpha}](\qntElm[2])) \sqsubseteq \nevlFun[\alpha](\AsgFam,
    \qntElm[1] \ldotp \qntElm[2])$, we proceed as follows:
    $\nevlFun[\alpha](\AsgFam, \qntElm[1] \ldotp \qntElm[2]) =
    \nevlFun[\alpha](\nevlFun[\alpha](\AsgFam, \qntElm[1]), \qntElm[2]) \sqsupseteq
    \nevlFun[\alpha](\nevlFun[\alpha](\AsgFam, \qntElm[1]), \CSym[\dual{\alpha}]
    (\qntElm[2])) = \nevlFun[\alpha](\AsgFam, \qntElm[1] \ldotp \CSym[\dual{\alpha}]
    (\qntElm[2]))$.

  \item In order to show that $\nevlFun[\alpha](\AsgFam,
    \CSym[\dual{\alpha}](\qntElm[1]) \ldotp \qntElm[2]) \sqsubseteq
    \nevlFun[\alpha](\AsgFam, \qntElm[1] \ldotp \qntElm[2])$, we proceed as follows:
    $\nevlFun[\alpha](\AsgFam, \qntElm[1] \ldotp \qntElm[2]) =
    \nevlFun[\alpha](\nevlFun[\alpha](\AsgFam, \qntElm[1]), \qntElm[2]) \sqsupseteq
    \nevlFun[\alpha](\nevlFun[\alpha](\AsgFam, \CSym[\dual{\alpha}] (\qntElm[1])),
    \qntElm[2]) = \nevlFun[\alpha](\AsgFam, \CSym[\dual{\alpha}] (\qntElm[1]) \ldotp
    \qntElm[2])$.

  \item In order to show that $\nevlFun[\alpha](\AsgFam,
    \CSym[\dual{\alpha}](\qntElm[1] \ldotp \qntElm[2])) \sqsubseteq
    \nevlFun[\alpha](\AsgFam, \qntElm[1] \ldotp \CSym[\dual{\alpha}](\qntElm[2]))$,
    we proceed as follows: $\nevlFun[\alpha](\AsgFam,
    \CSym[\dual{\alpha}](\qntElm[1] \ldotp \qntElm[2])) = \nevlFun[\alpha](\AsgFam,
    \CSym[\dual{\alpha}](\qntElm[1] \ldotp \CSym[\dual{\alpha}](\qntElm[2])))
    \sqsubseteq \nevlFun[\alpha](\AsgFam, \qntElm[1] \ldotp
    \CSym[\dual{\alpha}](\qntElm[2]))$.

  \item In order to show that $\nevlFun[\alpha](\AsgFam,
    \CSym[\dual{\alpha}](\qntElm[1] \ldotp \qntElm[2])) \sqsubseteq
    \nevlFun[\alpha](\AsgFam, \CSym[\dual{\alpha}](\qntElm[1]) \ldotp \qntElm[2])$,
    we proceed as follows: $\nevlFun[\alpha](\AsgFam,
    \CSym[\dual{\alpha}](\qntElm[1] \ldotp \qntElm[2])) = \nevlFun[\alpha](\AsgFam,
    \CSym[\dual{\alpha}](\CSym[\dual{\alpha}](\qntElm[1]) \ldotp \qntElm[2]))
    \sqsubseteq \nevlFun[\alpha](\AsgFam, \CSym[\dual{\alpha}](\qntElm[1]) \ldotp
    \qntElm[2])$.\qedhere
  \end{enumerate}
\end{proof}

\defenv{proposition}[][Prp][QntSwpIndII]
  {
  $\!\evlFun[\alpha](\AsgFam, \qntElm[1] \ldotp \qntElm[2]) \!\sqsubseteq\!
  \YFam \!\sqsubseteq\! \evlFun[\alpha](\AsgFam, \CSym[\alpha](\qntElm[1]
  \ldotp \qntElm[2]))$, for all $\YFam \in \{ \evlFun[\alpha](\AsgFam,
  \qntElm[1] \ldotp \CSym[\alpha](\qntElm[2])), \evlFun[\alpha](\AsgFam,
  \CSym[\alpha](\qntElm[1]) \ldotp \qntElm[2]) \}$, $\AsgFam \in \HAsgSet$ and
  $\qntElm[1], \qntElm[2] \in \QntSet[\BSym]$, with $\ap{\qntElm[1] \ldotp
  \qntElm[2]} \cap \ap{\AsgFam} = \emptyset$.
  }
\begin{proof}
  \begin{enumerate}
  \item In order to show that $\nevlFun[\alpha](\AsgFam, \qntElm[1] \ldotp
    \qntElm[2]) \sqsubseteq \nevlFun[\alpha](\AsgFam, \qntElm[1] \ldotp
    \CSym[\alpha](\qntElm[2]))$, we proceed as follows: $\nevlFun[\alpha](\AsgFam,
    \qntElm[1] \ldotp \qntElm[2]) = \nevlFun[\alpha](\nevlFun[\alpha](\AsgFam,
    \qntElm[1]), \qntElm[2]) \sqsubseteq \nevlFun[\alpha](\nevlFun[\alpha](\AsgFam,
    \qntElm[1]), \CSym[\alpha] (\qntElm[2])) = \nevlFun[\alpha](\AsgFam,
    \qntElm[1] \ldotp \CSym[\alpha] (\qntElm[2]))$.

  \item In order to show that $\nevlFun[\alpha](\AsgFam, \qntElm[1] \ldotp
    \qntElm[2]) \sqsubseteq \nevlFun[\alpha](\AsgFam,
    \CSym[\alpha](\qntElm[1]) \ldotp \qntElm[2])$, we proceed as follows:
    $\nevlFun[\alpha](\AsgFam, \qntElm[1] \ldotp \qntElm[2]) =
    \nevlFun[\alpha](\nevlFun[\alpha](\AsgFam, \qntElm[1]), \qntElm[2]) \sqsubseteq
    \nevlFun[\alpha](\nevlFun[\alpha](\AsgFam, \CSym[\alpha] (\qntElm[1])),
    \qntElm[2]) = \nevlFun[\alpha](\AsgFam, \CSym[\alpha] (\qntElm[1]) \ldotp
    \qntElm[2])$.

  \item In order to show that $\nevlFun[\alpha](\AsgFam, \qntElm[1] \ldotp
    \CSym[\alpha](\qntElm[2])) \sqsubseteq \nevlFun[\alpha](\AsgFam,
    \CSym[\alpha](\qntElm[1] \ldotp \qntElm[2]))$, we proceed as follows: $
    \nevlFun[\alpha](\AsgFam, \qntElm[1] \ldotp \CSym[\alpha](\qntElm[2]))
    \sqsubseteq \nevlFun[\alpha](\AsgFam, \CSym[\alpha](\qntElm[1] \ldotp
    \CSym[\alpha](\qntElm[2]))) = \nevlFun[\alpha](\AsgFam,
    \CSym[\alpha](\qntElm[1] \ldotp \qntElm[2]))$.

  \item In order to show that $\nevlFun[\alpha](\AsgFam,
    \CSym[\alpha](\qntElm[1]) \ldotp \qntElm[2]) \sqsubseteq
    \nevlFun[\alpha](\AsgFam, \CSym[\alpha](\qntElm[1] \ldotp \qntElm[2])) $, we
    proceed as follows: $\nevlFun[\alpha](\AsgFam, \CSym[\alpha](\qntElm[1]) \ldotp
    \qntElm[2]) \sqsubseteq \nevlFun[\alpha](\AsgFam,
    \CSym[\alpha](\CSym[\alpha](\qntElm[1]) \ldotp \qntElm[2])) =
    \nevlFun[\alpha](\AsgFam, \CSym[\alpha](\qntElm[1] \ldotp \qntElm[2]))$.\qedhere
  \end{enumerate}
\end{proof}

\recenv{ThmQntGamII}
\begin{proof}
  First of all, recall that the game $\Game[\QStr]$ of
  Construction~\ref{cns:qntgamii} is obtained from the game
  $\Game[\der{\qntElm}][\der{\Psi}]$ given in Construction~\ref{cns:qntgami},
  where
  \begin{itemize}
    \item
      $\der{\qntElm} \defeq \forall \vec{\apElm} \ldotp \qntElm[][\bullet]
      \ldotp \qntElm$ and
    \item
      $\der{\Psi} \defeq \Psi \cup \set{ \asgElm \in \AsgSet(\PSet) }{
      \asgElm \rst {\vec{\apElm}} \,\not\in \XSet }$,
  \end{itemize}
  with $\vec{\apElm} \defeq \ap{\AsgFam} \setminus \ap{\qntElm[][\bullet]}$ and
  $\PSet \defeq \ap{\qntElm} \cup \ap{\AsgFam}$ and
  being $\tuple {\qntElm[][\bullet]} {\XSet}$ a generator for
  $\AsgFam$.

  We can now proceed with the proof of the two properties.
  \begin{itemize}
    \item\textbf{[\ref{thm:qntgamii(elo)}]}
      If Eloise wins the game, by Theorem~\ref{thm:qntgami}, there exists a set
      of assignments $\der{\ESet} \in
      \nevlFun[\exists\forall](\CSym[\forall\exists](\der{\qntElm}))$ such that
      $\der{\ESet} \subseteq \der{\Psi}$.
      Thanks to Propositions~\ref{prp:evlmon} and~\ref{prp:qntswpindi}, we can
      easily prove the following inclusion between normal evaluations.
      \begin{claim}
        $\nevlFun[\exists\forall](\CSym[\forall\exists](\der{\qntElm}))
        \sqsubseteq \nevlFun[\exists\forall](\AsgFam,
        \CSym[\forall\exists](\qntElm))$.
      \end{claim}
      Due to the specific definition of the ordering $\sqsubseteq$ between
      hyperassignments, it follows that the above inclusion necessarily implies
      the existence of a set of assignments $\ESet \in
      \nevlFun[\exists\forall](\AsgFam, \CSym[\forall\exists](\qntElm))$ such
      that $\ESet \!\subseteq\! \der{\ESet}$.
      Therefore, $\ESet \!\subseteq\! \der{\Psi}$.

      At this point, we can immediately prove that $\ESet \subseteq \Psi$, being
      $\der{\Psi} = \Psi \cup \set{ \asgElm \in \AsgSet(\PSet) }{ \asgElm \rst
      {\vec{\apElm}} \,\not\in \XSet }$, thanks to the following claim, which
      can be derived from Proposition~\ref{prp:nrmevlrst}.
      \begin{claim}
        $\ESet \cap \set{ \asgElm \in \AsgSet(\PSet) }{ \asgElm \rst
        {\vec{\apElm}} \,\not\in \XSet } = \emptyset$, for all $\ESet \in
        \nevlFun[\exists\forall](\AsgFam, \CSym[\forall\exists](\qntElm))$.
      \end{claim}

    \item\textbf{[\ref{thm:qntgamii(abe)}]}
      If Abelard wins the game, by Theorem~\ref{thm:qntgami}, for all sets of
      assignments $\der{\ESet} \in
      \nevlFun[\exists\forall](\CSym[\exists\forall](\der{\qntElm}))$ it holds
      that $\der{\ESet} \not\subseteq \der{\Psi}$.

      First, given $\AsgFam[1], \AsgFam[2] \in \HAsgSet$,
      with $\ap{\AsgFam[1]} = \ap{\AsgFam[2]}$, and $\XSet \subseteq
      \AsgSet(\PSet)$ for some $\PSet \subseteq \APSet$, we define $\AsgFam[1]
      \sqsubseteq_{\XSet} \AsgFam[2]$ if and only if for every $\XSet[1] \in
      \AsgFam[1]$ there is $\XSet[2] \in \AsgFam[2]$ such that $\XSet[2] \setminus
      \set{ \asgElm \in \AsgSet } { \asgElm \rst {\PSet} \in \XSet } \subseteq
      \XSet[1]$.
      The following claim holds.

      \begin{claim}
        Let $\AsgFam[1], \AsgFam[2] \in \HAsgSet$ and $\XSet \subseteq
        \AsgSet(\PSet)$ for some $\PSet \subseteq \APSet$.
        Then, $\AsgFam[1] \sqsubseteq_{\XSet} \AsgFam[2]$ implies
        $\evlFun[\alpha](\AsgFam[1], \qntElm) \sqsubseteq_{\XSet}
        \evlFun[\alpha](\AsgFam[2], \qntElm)$, for all $\qntElm \in \QntSet$
        with $\ap{\AsgFam[1]} \cap \ap{\qntElm} = \emptyset$.
      \end{claim}

      \begin{proof}
        The proof is by induction on the length of $\qntElm$.

        If $\qntElm = \varepsilon$ (base case), then the claim follows
        trivially.

        Let $\qntElm = \Qnt^{\spcElm} \apElm \ldotp \qntElm'$.
        If $\alpha$ and $\Qnt$ are coherent, then $\evlFun[\alpha](\AsgFam[i],
        \qntElm) = \evlFun[\alpha](\extFun[\spcElm]{\AsgFam[i], \apElm}, \qntElm')$, for
        every $i \in \{ 1,2 \}$.
        We show that $\extFun[\spcElm]{\AsgFam[1], \apElm} \sqsubseteq_{\XSet}
        \extFun[\spcElm]{\AsgFam[2], \apElm} $ holds, and the thesis follows by applying
        the inductive hypothesis.
        Since $\extFun[\spcElm]{\AsgFam[i], \apElm} = \set{ \extFun{\XSet[i],
            \FFun[i], \apElm} }{ \XSet[i] \in \AsgFam[1], \FFun_i \in
          \FncSet[\spcElm](\ap{\AsgFam[i]}) }$ ($i \in \{ 1,2 \}$), we have to show that
        for every $\XSet[1] \in \AsgFam[1]$ and $\FFun[1] \in
        \FncSet[\spcElm](\ap{\AsgFam[1]})$ ($= \FncSet[\spcElm](\ap{\AsgFam[2]})$) there
        are $\XSet[2] \in \AsgFam[2]$ and $\FFun[2] \in
        \FncSet[\spcElm](\ap{\AsgFam[1]})$ such that $\extFun{\XSet[2], \FFun[2],
          \apElm} \setminus \set{ \asgElm \in \AsgSet } { \asgElm \rst {\PSet} \in \XSet }
        \subseteq \extFun{\XSet[1], \FFun[1], \apElm}$.
        It is easy to see that for every $\XSet[1] \in \AsgFam[1]$ and $\FFun[1]
        \in \FncSet[\spcElm](\ap{\AsgFam[1]})$, it holds that $\extFun{f(\XSet[1]),
          \FFun[1], \apElm} \setminus \set{ \asgElm \in \AsgSet } { \asgElm \rst {\PSet}
          \in \XSet } \subseteq \extFun{\XSet[1], \FFun[1], \apElm}$, where $f :
        \AsgFam[1] \rightarrow \AsgFam[2]$ is a witness for $\AsgFam[1]
        \sqsubseteq_{\XSet} \AsgFam[2]$.

        Instead, if $\alpha$ and $\Qnt$ are not coherent, then
        $\evlFun[\alpha](\AsgFam[i], \qntElm) \equiv
        \evlFun[\alpha](\nevlFun[\alpha](\AsgFam[i], \Qnt^{\spcElm} \apElm), \qntElm')$,
        for every $i \in \{ 1,2 \}$.
        We show that $\nevlFun[\alpha](\AsgFam[1], \Qnt^{\spcElm} \apElm)
        \sqsubseteq_{\XSet} \nevlFun[\alpha](\AsgFam[2], \Qnt^{\spcElm} \apElm)$ holds,
        and the thesis follows by applying the inductive hypothesis.
        Since $\nevlFun[\alpha](\AsgFam[i], \Qnt^{\spcElm} \apElm) = \set{
          \extFun{\eth_i, \apElm} }{ \eth_i \in \FncSet[\spcElm](\ap{\AsgFam[i]}) \to
          \AsgFam[i] }$, where $\extFun{\eth_i, \apElm} = \bigcup_{\FFun \in \dom{\eth_i}}
        \extFun{\eth_i(\FFun), \FFun, \apElm}$ ($i \in \{ 1,2 \}$), we have to show that
        for every $\eth_1 \in \FncSet[\spcElm](\ap{\AsgFam[1]}) \to \AsgFam[1]$ there is
        $\eth_2 \in \FncSet[\spcElm](\ap{\AsgFam[1]}) \to \AsgFam[2]$ (recall that
        $\ap{\AsgFam[1]} = \ap{\AsgFam[2]}$) such that $\extFun{\eth_2, \apElm}
        \setminus \set{ \asgElm \in \AsgSet } { \asgElm \rst {\PSet} \in \XSet }
        \subseteq \extFun{\eth_1, \apElm}$.
        To this end, we define a function $g: (\FncSet[\spcElm](\ap{\AsgFam[1]})
        \to \AsgFam[1]) \rightarrow (\FncSet[\spcElm](\ap{\AsgFam[1]}) \to \AsgFam[2])$
        as follows.
        For every $\eth_1 \in \FncSet[\spcElm](\ap{\AsgFam[1]}) \to \AsgFam[1]$
        and $\FFun \in \FncSet[\spcElm](\ap{\AsgFam[1]})$, we define $g(\eth_1)(\FFun) =
        f(\eth_1(\FFun))$.
        Clearly, since $f$ is a witness for $\AsgFam[1] \sqsubseteq_{\XSet}
        \AsgFam[2]$, it holds that $g(\eth_1)(\FFun) \setminus \set{ \asgElm \in \AsgSet
        } { \asgElm \rst {\PSet} \in \XSet } \subseteq \eth_1(\FFun)$, for every $\FFun
        \in \FncSet[\spcElm](\ap{\AsgFam[1]})$.
        Thus, for every $\eth_1 \in \FncSet[\spcElm](\ap{\AsgFam[1]}) \to
        \AsgFam[1]$, the following holds: $\extFun{g(\eth_1), \apElm} \setminus \set{
          \asgElm \in \AsgSet } { \asgElm \rst {\PSet} \in \XSet } = \left( \bigcup_{\FFun
            \in \dom{g(\eth_1)}} \extFun{g(\eth_1)(\FFun), \FFun, \apElm} \right) \setminus
        \set{ \asgElm \in \AsgSet } { \asgElm \rst {\PSet} \in \XSet } = \bigcup_{\FFun
          \in \dom{g(\eth_1)}} (\extFun{g(\eth_1)(\FFun), \FFun, \apElm} \setminus \set{
          \asgElm \in \AsgSet } { \asgElm \rst {\PSet} \in \XSet } )\subseteq \bigcup_{\FFun \in
          \dom{g(\eth_1)}} \extFun{g(\eth_1)(\FFun) \setminus \set{ \asgElm \in \AsgSet }
          { \asgElm \rst {\PSet} \in \XSet }, \FFun, \apElm}$

        $ \subseteq \bigcup_{\FFun \in \dom{\eth_1}} \extFun{\eth_1(\FFun),
          \FFun, \apElm} = \extFun{\eth_1, \apElm}$ (notice that $\dom{g(\eth_1)} =
        \dom{\eth_1}$).
      \end{proof}

      Now, notice that, by definition of generator, $\XSet \subseteq
      \AsgSet(\vec{\apElm})$.
      Let $\dual{\XSet} \defeq \AsgSet(\vec{\apElm}) \setminus \XSet$.
      Thanks to Propositions~\ref{prp:evlmon},
%
%
      it is possible to prove the following claim, by observing that, clearly,
      $\{ \XSet\} \sqsubseteq_{\dual \XSet} \{ \AsgSet(\vec{\apElm}) \}$ holds, as
      $\AsgSet(\vec{\apElm}) \setminus \set{ \asgElm \in \AsgSet } { \asgElm \rst
        {\vec{\apElm}} \in \dual {\XSet} } = \XSet$.
      \begin{claim}
        $\evlFun[\exists\forall](\AsgFam, \CSym[\exists\forall](\qntElm))
        \sqsubseteq_{\dual{\XSet}}
        \evlFun[\exists\forall](\CSym[\exists\forall](\der{\qntElm}))$.
      \end{claim}
%
       %
      Due to the specific definition of the ordering $\sqsubseteq$ between
      hyperassignments, it follows that the above inclusion necessarily implies
      the non existence of a set of assignments $\ESet \in
      \evlFun[\exists\forall](\AsgFam, \CSym[\exists\forall](\qntElm))$ such
      that $\ESet \subseteq \der{\Psi}$.
      Indeed, assume, towards a contradiction that there is $\ESet \in
      \evlFun[\exists\forall](\AsgFam, \CSym[\exists\forall](\qntElm))$ such that
      $\ESet \subseteq \der{\Psi}$.
      By the above claim, there is $\der{\ESet} \in
      \evlFun[\exists\forall](\CSym[\exists\forall](\der{\qntElm}))$ such that
      $\der{\ESet} \setminus \set{ \asgElm \in \AsgSet } { \asgElm \rst {\vec{\apElm}}
        \in \dual{\XSet} } \subseteq \ESet \subseteq \der{\Psi} $.
      Since $\der{\ESet} \cap \set{ \asgElm \in \AsgSet } { \asgElm \rst
        {\vec{\apElm}} \in \dual{\XSet} } \subseteq \set{ \asgElm \in \AsgSet(\PSet) }{
        \asgElm \rst {\vec{\apElm}} \,\not\in \XSet } \subseteq \der{\Psi}$, we have
      that $\der{\ESet} \subseteq \der{\Psi}$, which is in contradiction with Abelard
      winning the game.
      Hence, $\ESet \not\subseteq \der{\Psi}$ holds for all $\ESet \in
      \evlFun[\exists\forall](\AsgFam, \CSym[\exists\forall](\qntElm))$, which implies
      $\ESet \not\subseteq \Psi$, for all $\ESet \in \evlFun[\exists\forall](\AsgFam,
      \CSym[\exists\forall](\qntElm))$, being $\Psi \subseteq \der{\Psi}$.
      \qedhere
  \end{itemize}
\end{proof}

\recenv{ThmFrmCanFor}
\begin{proof}
  By Proposition~\ref{prp:nrmevlind} and Lemma~\ref{lmm:prfevl}, we have that
  $\AsgFam \cmodels[][\alpha] \qntElm \psi$ \iff
  $\nevlFun[\alpha](\AsgFam,\qntElm) \cmodels[][\alpha] \psi$, $\AsgFam
  \cmodels[][\alpha] \CSym[\forall\exists](\qntElm) \psi$ \iff
  $\nevlFun[\alpha](\AsgFam,\CSym[\forall\exists](\qntElm)) \cmodels[][\alpha]
  \psi$, and $\AsgFam \cmodels[][\alpha] \CSym[\exists\forall](\qntElm) \psi$ \iff
  $\nevlFun[\alpha](\AsgFam,\CSym[\exists\forall](\qntElm)) \cmodels[][\alpha]
  \psi$.

  If $\alpha = \QEA$, then, by Proposition~\ref{prp:qntswpind}, it holds
  that $\nevlFun[\alpha](\AsgFam,\CSym[\forall\exists](\qntElm)) \sqsubseteq
  \nevlFun[\alpha](\AsgFam,\qntElm) \sqsubseteq
  \nevlFun[\alpha](\AsgFam,\CSym[\exists\forall](\qntElm))$.
  Therefore, by Theorem~\ref{thm:hypref}, we have that $\AsgFam
  \cmodels[][\alpha] \CSym[\forall\exists](\qntElm) \psi$ implies $\AsgFam
  \cmodels[][\alpha] \qntElm \psi$, which, in turn, implies $\AsgFam
  \cmodels[][\alpha] \CSym[\exists\forall](\qntElm) \psi$.

  If $\alpha = \QAE$, then, by Proposition~\ref{prp:qntswpind}, it holds
  that $\nevlFun[\alpha](\AsgFam,\CSym[\exists\forall](\qntElm)) \sqsubseteq
  \nevlFun[\alpha](\AsgFam,\qntElm) \sqsubseteq
  \nevlFun[\alpha](\AsgFam,\CSym[\forall\exists](\qntElm))$.
  Therefore, by Theorem~\ref{thm:hypref}, we have that $\AsgFam
  \cmodels[][\alpha] \CSym[\exists\forall](\qntElm) \psi$ implies $\AsgFam
  \cmodels[][\alpha] \qntElm \psi$, which, in turn, implies $\AsgFam
  \cmodels[][\alpha] \CSym[\forall\exists](\qntElm) \psi$.

  We are now left to show that $\AsgFam \cmodels[][\QEA]
  \CSym[\QEA](\qntElm) \psi$ implies $\AsgFam \cmodels[][\QEA]
  \CSym[\QAE](\qntElm) \psi$ and that $\AsgFam \cmodels[][\QAE]
  \CSym[\QAE](\qntElm) \psi$ implies $\AsgFam \cmodels[][\QAE]
  \CSym[\QEA](\qntElm) \psi$.

  In order to prove the former, we proceed as follows.
  Let $\Psi \defeq \set{\asgElm \in \AsgSet(\free{\psi})}{ \asgElm
    \cmodels[\LTL] \psi }$ and notice that $\Psi$ is Borelian~\cite{PP04}.
  Observe that $\QStr = \tuple {\AsgFam} {\qntElm}
  {\Psi}$ is a quantification-game schema.
  Since $\AsgFam$ is Borelian behavioral, then there is a generator $\tuple
  {\tilde{\qntElm}} {\XSet}$ for it.
  Therefore, the game $\Game[\der{\qntElm}][\der{\Psi}]$ is a $\QStr$-game,
  where $\vec{p} \defeq \ap{\AsgFam} \setminus \ap{\tilde{\qntElm}}$,
  $\der{\qntElm} \defeq \forall \vec{\apElm} \ldotp \tilde{\qntElm} \ldotp
  \qntElm$, and $\der{\Psi} \defeq \Psi \cup \set{ \asgElm \in
    \AsgSet(\free{\psi}) }{ \asgElm \rst {\vec{\apElm}} \,\not\in \XSet }$.
  By Proposition~\ref{prp:nrmevlind} and Item~\ref{thm:semadqi(ea)} of
  Theorem~\ref{thm:semadqi}, $\AsgFam \cmodels[][\QEA]
  \CSym[\exists\forall](\qntElm) \psi$ implies that there exists $\XSet \in
  \nevlFun[\QEA](\AsgFam,\CSym[\QEA](\qntElm))$ such that $\XSet \subseteq \Psi$.
  By Item~\ref{thm:qntgami(abe)} of Theorem~\ref{thm:qntgamii}, it holds that
  Abelard does not win the game $\Game[\der{\qntElm}][\der{\Psi}]$.
  Now, notice that $\Game[\der{\qntElm}][\der{\Psi}]$ is Borelian; thus, as a
  consequence of Martin's determinacy theorem~\cite{Mar75,Mar85},
  $\Game[\der{\qntElm}][\der{\Psi}]$ is won by Eloise.
  Hence, from Item~\ref{thm:qntgami(elo)} of Theorem~\ref{thm:qntgamii}, it
  follows that there exists $\XSet \in
  \nevlFun[\QEA](\AsgFam,\CSym[\QAE](\qntElm))$ such that $\asgElm \cmodels[\LTL]
  \psi$ for all $\asgElm \in \XSet$, which, by Proposition~\ref{prp:nrmevlind} and
  Item~\ref{thm:semadqi(ea)} of Theorem~\ref{thm:semadqi}, implies $\AsgFam
  \cmodels[][\QEA] \CSym[\QAE](\qntElm) \psi$.

  Finally, we turn to showing that $\AsgFam \cmodels[][\QAE]
  \CSym[\QAE](\qntElm) \psi$ implies $\AsgFam \cmodels[][\QAE]
  \CSym[\QEA](\qntElm) \psi$.
  By~Proposition~\ref{prp:qntswpind}, it holds that
  $\nevlFun[\QAE](\AsgFam, \CSym[\QEA](\qntElm)) \sqsubseteq
  \nevlFun[\QAE](\AsgFam, \CSym[\QAE](\qntElm))$.
  By Theorem~\ref{thm:hypref}, we have that $\nevlFun[\QAE](\AsgFam,
  \CSym[\QAE](\qntElm)) \cmodels[][\QAE] \psi$ implies $\nevlFun[\QAE](\AsgFam,
  \CSym[\QEA](\qntElm)) \cmodels[][\QAE] \psi$.
  The thesis follows from Lemma~\ref{lmm:prfevl} and
  Proposition~\ref{prp:nrmevlind}.
%
%
%
%
%
%
%
%
\end{proof}





\subsection{Proofs for Section~\ref{sec:decprb}}
\label{app:decprb}

	\recenv{ThmSatGam}

\begin{proof}



Let $\varphi = \quantificationPrefix \LTLformula$ be a behavioral \QPTL
prenex sentence with $\quantificationPrefix$ a quantification prefix and $\LTLformula$
an \LTL formula. The idea of the proof is to construct a parity game from the
automaton recognizing models of $\LTLformula$ and the quantification prefix that
will be equivalent to the game $\Game[\qntElm][\psi]$ defined in
Construction~\ref{cns:qntgami}.

	From $\LTLformula$, we construct a non-deterministic B\"uchi automaton
$A_\LTLformula$ recognizing models of $\LTLformula$ using the
Vardi-Wolper construction~\cite{VW86a}. We will consider a complete determinized
parity Automaton $\Automaton=\tuple{\Astates}{\Ainitialstate}{\Aalphabet}{
\Atransition}{\Aacceptingcondition}$ equivalent to  $A_\LTLformula$ (that can be
obtained via a Safra-like determinization procedure~\cite{Pit06}) with
	\begin{itemize}
		\item $\Astates$ is the set of states of $\Automaton$,
		\item $\Ainitialstate \in \Astates$ is the initial state,
		\item $\Aalphabet = \ValSet$ is the alphabet of $\Automaton$,
		\item $\Atransition: \Astates\times\Aalphabet\rightarrow\Astates$ is the transition function,
		\item $\Aacceptingcondition$ is the parity condition.
	\end{itemize}

	The parity game associated is defined as $\Game[\varphi]\defeq\tuple{\Arena}{\ObsSet}{\WinSet}$ with arena $\Arena \defeq \tuple {\PPosSet} {\OPosSet} {\iposElm} {\MovRel}$ and is constructed as follows:
	\begin{itemize}
		\item
		the set of positions $\APosSet \subseteq \Astates\times\QPosSet$ contains exactly the pairs of one state of the automaton $\Automaton$ and a valuation $\valElm \in \ValSet$ that is a position of $\Game[\quantificationPrefix][\LTLformula]$;
		\item
		the set of Eloise's positions $\PPosSet \subseteq \APosSet$ only contains
		the positions $(\Astate,\valElm) \in \APosSet$ for $\valElm$ is an Eloise's position in $\Game[\quantificationPrefix][\LTLformula]$;
		\item
		the initial position $\iposElm \defeq (\Ainitialstate,\emptyfun)$ is just the initial state of $\Automaton$ paired with the initial state of $\Game[\quantificationPrefix][\LTLformula]$;
		\item
		the move relation $\MovRel \subseteq \APosSet \times \APosSet$ contains
		exactly those pairs of positions $((\Astate_1,\valElm[1]), (\Astate_2,\valElm[2])) \in \APosSet
		\times \APosSet$ such that:
		\begin{itemize}
			\item $(\valElm[1],\valElm[2])$ is a move in $\Game[\qntElm][\psi]$;
			\item if $\valElm[2] = \emptyfun$ then $\Astate_2=\Atransition(\Astate_1,\valElm[1])$, otherwise, $\Astate_1=\Astate_2$;
		\end{itemize}
		\item
		the set of observable positions $\ObsSet \defeq \Astates\times\set{\emptyfun}{}$;
		\item
		the winning condition is deduced from the accepting condition of the automaton $\Automaton$. More precisely, the priority of a position $(\Astate,\emptyfun)$ is defined as the priority of $\Astate$.
	\end{itemize}

	We first show that the arena of the game $\Game[\varphi]$ is very similar to the arena of $\Game[\quantificationPrefix][\LTLformula]$ trough the definition of a bijection between initial paths on $\Game[\varphi]$ (denoted $\iQPaths$) and initial paths on $\Game[\quantificationPrefix][\LTLformula]$ (denoted $\iAPaths$) $\PathToPath:\iQPaths\rightarrow\iAPaths$ .

	We consider the morphism $\PathToPathmorph\colon\APosSet\to\QPosSet$ defined on $\Apos$ by the following.
	\[\PathToPathmorph(\Apos)\defeq\Qpos\]

	We now define $\PathToPath$ on $\Apath\in\iAPaths$ by the following.
	\[\PathToPath(\Apath)\defeq\PathToPathmorph(\Apath)\]
	\begin{claim}
		\label{claimPathToPath}
		The function $\PathToPath$ between $\iQPaths$ and $\iAPaths$ is a bijection.
	\end{claim}

	\begin{proof}
		First, we show that $\PathToPathmorph$ is indeed a morphism:
		\begin{itemize}
			\item there is a move $\Aposq1\Aposq2$ for some $\Astate\in\Astates$ in $\Game[\varphi]$ \iff there is a move $\Qpos[1]\Qpos[2]$ in $\Game[\quantificationPrefix][\LTLformula]$,
			\item there is a move $\Apos\AposA'$ for some $\valElm\in\ValSet(\ap{\qntElm})$ in $\Game[\varphi]$ \iff there is a move $\Qpos\emptyfun$ in $\Game[\quantificationPrefix][\LTLformula]$.
		\end{itemize}
		We conclude that if $\Apos[1]\Apos[2]$ is a move in $\Game[\varphi]$, then $\PathToPathmorph(\Apos[1]\Apos[2])$ is a move in $\Game[\quantificationPrefix][\LTLformula]$. Then $\PathToPath$ is well defined.

		The inverse function is also well defined because we consider only initial paths and the automaton $\Automaton$ is deterministic and complete. Thus, from an history $\aQhistory$ of $\Game[\quantificationPrefix][\LTLformula]$, there is only one state $\Astate$ reached by $\Automaton$ by reading $\obsFun(\aQhistory)$. Formally $\PathToPath^{-1}$ can be defined inductively on a $\alpha$-history $\aQhistory=\aQhistory'\valElm$ as follows:
		\begin{itemize}
			\item $\QgameToAgame(\emptyfun)=(\Ainitialstate,\emptyfun')$
			\item if $\last{\QgameToAgame(\aQhistory'})=\Apos$ with $\dom{\valElm}= \ap{(\quantificationPrefix)_{< \#(\valElm)}}$, then $\QgameToAgame(\aQhistory)=\QgameToAgame(\aQhistory')(\Atransition(\Astate,\valElm),\emptyfun)$,
			\item otherwise, $\last{\QgameToAgame(\aQhistory')}=(\Astate,\valElm')$ for some $\Astate\in\Astates$ thus $\QgameToAgame(\aQhistory)=\QgameToAgame(\aQhistory')\Apos$.
		\end{itemize}
	with $\last{\wrdElm}$ being the last character of the finite word $\wrdElm$.

		We just showed that $\PathToPath$ is a bijection between initial paths $\iAPaths$ and initial paths $\iQPaths$.
	\end{proof}

	We want to state that a play in $\Game[\varphi]$ is won by Eloise \iff its image by $\PathToPath$ is also won by Eloise. To do so, we first characterize the plays in $\Game[\varphi]$ that Eloise wins by linking them with an execution of $\Automaton$. To precise this link, we associate an assignment $\asgElm$ to a play $\Aplay$ of $\Game[\varphi]$.

	We define a function $\PlayToAsg:\APlays\to\ValSet(\ap{\quantificationPrefix})^{\omega}$ as follow where $\Aplay=(\Aplay)_{i\in\SetN}$ is a play on $\Game[\varphi]$.
	\[\begin{array}{l}
	\PlayToAsg(\Aplay) \in \ValSet(\ap{\quantificationPrefix})^{\omega}\\
	\PlayToAsg(\Aplay)_t \defeq \snd(\Aplay_{(t+1)\times(\card{\ap{\quantificationPrefix}}+1)-1})
	\end{array}
	\]



	\begin{claim}
		\label{claimPlayToAsg}
		The function $\PlayToAsg$ is a bijection and a play $\Aplay$ is won by Eloise \iff $\PlayToAsg(\Aplay)$ is recognized by the automaton $\Automaton[\psi]$.
	\end{claim}
	\begin{proof}
		We first prove that $\PlayToAsg$ is a bijection by proving it is well defined and it is a surjection and an injection. By definition of the game $\Game[\varphi]$, for every $t\in\SetN$, it holds that $\Aplay_{(t+1)\times(\card{\ap{\quantificationPrefix}}+1)}$ has the form $(\Astate,\emptyfun)$. Thus, the previous positions necessarily has the form $(\Astate', \valElm)$ with $\valElm\in\ValSet(\ap{\qntElm})$. Then $\PlayToAsg$ is well defined.

		To show that $\PlayToAsg$ is a bijection, we show that every word $\wordAsg$ is the image of exactly one play. Consider a word $\wordAsg\in\ValSet(\ap{\quantificationPrefix})^{\omega}$. For every natural $t$ and every position $(\Astate,\emptyfun)$ there exist a unique finite path of $\card{\ap{\quantificationPrefix}}$ moves reaching the position $(\Astate, \wordAsg_t)$, by definition of the game $\Game[\varphi]$. This path $\Apath$ is defined as $\Apath_0=(\Astate,\emptyfun)$ and $\Apath_{i+1}=(\Astate,\valElm_{i+1})$ with $\valElm_{i+1}=(\wordAsg_t)_{<i+1}$ for every $i<\card{\ap{\quantificationPrefix}}$. Thanks to this construction, we can build step by step the full play, starting at $(\Ainitialstate,\emptyfun)$ with $t=0$.
		\begin{enumerate}
			\item From a state of the form $(\Astate,\emptyfun)$, goes to $(\Astate,\wordAsg_t)$ in $\card{\ap{\quantificationPrefix}}$ moves.
			\item The next move in necessarily from $(\Astate,\wordAsg_t)$ to $(\Atransition(\Astate,\wordAsg_t),\emptyfun)$.
			\item Go to $1)$ with $t:=t+1$
		\end{enumerate}
		Steps $1)$ and $2)$ use $\card{\ap{\quantificationPrefix}}+1$ moves. Positions of the form $(\Astate,\wordAsg_t)$ are reached at the end of the first step. Thus we have $\Aplay_{(t+1)\times(\card{\ap{\quantificationPrefix}}+1)-1}=(\Astate,\wordAsg_t)$.

		Because of the parity condition of the game $\Game[\varphi]$, it is clear that a play $\Aplay$ is won by Eloise if and only if $\PlayToAsg(\Aplay)$ is accepted by the automaton $\Automaton$.
	\end{proof}

	We now explicit the link between $\PathToPath$ and $\PlayToAsg$: the word associated with a play in $\Game[\varphi]$ is the sequence of observable position in the associated play in $\Game[\quantificationPrefix][\LTLformula]$.

	\begin{claim}
		\label{claimPlayToAsgAndPathToPath}
		Given a play $\Aplay$ in the game $\Game[\varphi]$, we have for every $t$ natural number, $\PlayToAsg(\Aplay)_t=\obsFun(\PathToPath(\Aplay))_t$.
	\end{claim}
	\begin{proof}
		This claim derive directly from definitions of $\PathToPath$, $\PlayToAsg$ and $\Game[\quantificationPrefix][\LTLformula]$.
	\end{proof}

	We now show that Eloise wins a play in $\Game[\varphi]$ \iff she wins the corresponding play in $\Game[\quantificationPrefix][\LTLformula]$.

	\begin{claim}
		The bijection $\PathToPath$ preserves the winner: $\Aplay$ is won by Eloise
in $\Game[\varphi]$ \iff $\PathToPath(\Aplay)$ is won by Eloise in $\Game[\quantificationPrefix][\LTLformula]$.
	\end{claim}

	\begin{proof}
		Consider a play $\Aplay$ in the game $\Game[\varphi]$. Thanks to Claim~\ref{claimPlayToAsg}, we know that $\Aplay$ is won by Eloise \iff $\PlayToAsg(\Aplay)$ is accepted by $\Automaton$ which means that $\wrdInvFun(\PlayToAsg(\Aplay))\models\psi$. Claim~\ref{claimPlayToAsgAndPathToPath} assure that $\PlayToAsg(\Aplay)(n)=\obsFun(\PathToPath(\Aplay))_n$. We can deduce that $\wrdInvFun(\obsFun(\PathToPath(\Aplay)))$ define an assignment that satisfy $\psi$. Then $\PathToPath(\Aplay)$ is won by Eloise.

		The reciprocity is ensured by the bijective property of $\PathToPath$ (Claim~\ref{claimPathToPath}.)
	\end{proof}

Because $\PathToPath$ is a bijection, we can derive a bijection between strategies on $\Game[\varphi]$ and strategies on $\Game[\quantificationPrefix][\LTLformula]$. It is not hard to see that a strategy for Eloise is winning \iff its associated strategy is also winning. Thus Eloise wins $\Game[\varphi]$ \iff she wins $\Game[\quantificationPrefix][\LTLformula]$.

The automaton $A_\LTLformula$ has a size exponential in the size of $\LTLformula$. The procedure to determinize in a parity automaton adds one exponential; thus $\card{\Automaton}=\AOmicron{\pow{\pow{\card{\LTLformula}}}}$. The quantification game have a size in $\AOmicron{\pow{\card{\qntElm}}}$. We conclude that the game constructed has a size in $\AOmicron{\pow{\card{\qntElm}}\pow{\pow{\card{\LTLformula}}}}=\AOmicron{\pow{\pow{\card{\varphi}}}}$. The game has the same number of priorities as the automaton $\Automaton$ which is in $\AOmicron{\pow{\card{\LTLformula}}}$.
\end{proof}




\end{document}
